\RequirePackage{snapshot}
\documentclass[11pt,twoside,british]{article}
\usepackage[T1]{fontenc}
\usepackage[latin9]{inputenc}
\usepackage[a4paper]{geometry}
\usepackage{fancyhdr}
\usepackage{color}
\usepackage{babel}
\usepackage{verbatim}
\usepackage{refstyle}
\usepackage{booktabs}
\usepackage{mathrsfs}
\usepackage{mathtools}
\usepackage{enumitem}
\usepackage{amsmath}
\usepackage{amsthm}
\usepackage{amssymb}
\usepackage{graphicx}
\usepackage{setspace}
\usepackage[authoryear,longnamesfirst]{natbib}
\usepackage{xargs}[2008/03/08]
\usepackage[bookmarks=true,bookmarksnumbered=false,bookmarksopen=false,
 breaklinks=false,pdfborder={0 0 0},pdfborderstyle={},backref=false,colorlinks=false]
 {hyperref}

\makeatletter

\AtBeginDocument{\providecommand\eqref[1]{\ref{eq:#1}}}
\AtBeginDocument{\providecommand\defref[1]{\ref{def:#1}}}
\AtBeginDocument{\providecommand\secref[1]{\ref{sec:#1}}}
\AtBeginDocument{\providecommand\assref[1]{\ref{ass:#1}}}
\AtBeginDocument{\providecommand\propref[1]{\ref{prop:#1}}}
\AtBeginDocument{\providecommand\subsecref[1]{\ref{subsec:#1}}}
\AtBeginDocument{\providecommand\tblref[1]{\ref{tbl:#1}}}
\AtBeginDocument{\providecommand\figref[1]{\ref{fig:#1}}}
\AtBeginDocument{\providecommand\thmref[1]{\ref{thm:#1}}}
\AtBeginDocument{\providecommand\exaref[1]{\ref{exa:#1}}}
\AtBeginDocument{\providecommand\appref[1]{\ref{app:#1}}}
\AtBeginDocument{\providecommand\suppref[1]{\ref{supp:#1}}}
\AtBeginDocument{\providecommand\lemref[1]{\ref{lem:#1}}}
\AtBeginDocument{\providecommand\enuref[1]{\ref{enu:#1}}}
\AtBeginDocument{\providecommand\remref[1]{\ref{rem:#1}}}
\providecommand{\tabularnewline}{\\}
\RS@ifundefined{subsecref}
  {\newref{subsec}{name = \RSsectxt}}
  {}
\RS@ifundefined{thmref}
  {\def\RSthmtxt{theorem~}\newref{thm}{name = \RSthmtxt}}
  {}
\RS@ifundefined{lemref}
  {\def\RSlemtxt{lemma~}\newref{lem}{name = \RSlemtxt}}
  {}

\numberwithin{equation}{section}
\numberwithin{figure}{section}
\numberwithin{table}{section}
\theoremstyle{remark}
\newtheorem*{notation*}{\protect\notationname}
\theoremstyle{definition}
\newtheorem{example}{\protect\examplename}[section]
\theoremstyle{plain}
\newtheorem{assumption}{\protect\assumptionname}
\theoremstyle{plain}
\newtheorem{prop}{\protect\propositionname}[section]
\theoremstyle{definition}
\newtheorem{defn}{\protect\definitionname}[section]
\theoremstyle{remark}
\newtheorem{rem}{\protect\remarkname}[section]
\theoremstyle{plain}
\newtheorem{thm}{\protect\theoremname}[section]
\theoremstyle{plain}
\newtheorem{lem}{\protect\lemmaname}[section]

\setlist[enumerate,1]{label=\upshape{(\roman*)}, ref=(\roman*)}
\setlist[enumerate,2]{label=\upshape{(\alph*)}, ref=(\alph*)}
\setlist[enumerate,3]{label=\upshape{\roman*.}, ref=\roman*}

\makeatletter
  \@ifpackageloaded{refstyle}{}{\usepackage{refstyle}}
\makeatother

\newref{thm}{name = Theorem~}
\newref{prop}{name = Proposition~}
\newref{lem}{name = Lemma~}
\newref{cor}{name = Corollary~}
\newref{ass}{name = }
\newref{exa}{name = Example~}
\newref{rem}{name = Remark~}
\newref{def}{name = Definition~}

\newref{eq}{refcmd = (\ref{#1})}

\newref{sec}{name = Section~}
\newref{sub}{name = Section~}
\newref{subsec}{name = Section~}
\newref{app}{name = Appendix~}
\newref{supp}{name = Appendix~}

\newref{fig}{name = Figure~}
\newref{tbl}{name = Table~}

\makeatletter
  \@ifpackageloaded{mathrsfs}{}{\usepackage{mathrsfs}}
\makeatother

\date{}

\usepackage{nextpage}

\allowdisplaybreaks[1]

\usepackage{scalefnt}
\usepackage[group-separator={,}]{siunitx}
\newcommand\smaller[2][0.85]{{\scalefont{#1}#2}}
\newcommand{\ass}[1]{{\upshape{\smaller[0.76]{#1}}}}

\newcommand{\assumpname}[1]{%
  \renewcommand{\theassumption}{\ass{#1}}%
}

\makeatletter
\newsavebox{\@brx}
\newcommand{\dbllangle}[1][]{\savebox{\@brx}{\(\m@th{#1\langle}\)}%
  \mathopen{\copy\@brx\kern-0.5\wd\@brx\usebox{\@brx}}}
\newcommand{\dblrangle}[1][]{\savebox{\@brx}{\(\m@th{#1\rangle}\)}%
  \mathclose{\copy\@brx\kern-0.5\wd\@brx\usebox{\@brx}}}
\makeatother

\usepackage{changepage}

\theoremstyle{definition}
\newtheorem{examplex}{Example}
\renewenvironment{example}
  {\pushQED{\qed}\examplex}
  {\popQED\endexamplex}

\numberwithin{examplex}{section}

\newcommand{\exname}[1]{%
  \renewcommand{\theexamplex}{{#1}}%
}

\newcounter{subremark}[rem]
\renewcommand{\thesubremark}{(\roman{subremark})}

\newcommand{\subremark}{%
  \refstepcounter{subremark}%
  \thesubremark{}.%
}

\newcounter{savedexnumber}

\newcommand{\saveexamplex}{%
  \setcounter{savedexnumber}{\value{examplex}}
}

\newcommand{\restoreexamplex}{%
  \setcounter{examplex}{\value{savedexnumber}}
  \numberwithin{examplex}{section}
}

\usepackage{mathdots}

\usepackage[usestackEOL]{stackengine}
\setstackgap{S}{5pt}
\newcommand{\authaffil}[2]{\Shortunderstack{#1\\\small{#2}}}

\usepackage{needspace}

\usepackage{bibunits}
\defaultbibliography{cksvar}
\defaultbibliographystyle{ecta}

\newcommand{\tightermath}[3]{%
#3%
}

\makeatother

\providecommand{\assumptionname}{Assumption}
\providecommand{\definitionname}{Definition}
\providecommand{\examplename}{Example}
\providecommand{\lemmaname}{Lemma}
\providecommand{\notationname}{Notation}
\providecommand{\propositionname}{Proposition}
\providecommand{\remarkname}{Remark}
\providecommand{\theoremname}{Theorem}

\begin{document}


\global\long\def\uwrite#1#2{\underset{#2}{\underbrace{#1}} }%

\global\long\def\blw#1{\ensuremath{\underline{#1}}}%

\global\long\def\abv#1{\ensuremath{\overline{#1}}}%

\global\long\def\vect#1{\mathbf{#1}}%


\global\long\def\smlseq#1{\{#1\} }%

\global\long\def\seq#1{\left\{  #1\right\}  }%

\global\long\def\smlsetof#1#2{\{#1\mid#2\} }%

\global\long\def\setof#1#2{\left\{  #1\mid#2\right\}  }%


\global\long\def\goesto{\ensuremath{\rightarrow}}%

\global\long\def\ngoesto{\ensuremath{\nrightarrow}}%

\global\long\def\uto{\ensuremath{\uparrow}}%

\global\long\def\dto{\ensuremath{\downarrow}}%

\global\long\def\uuto{\ensuremath{\upuparrows}}%

\global\long\def\ddto{\ensuremath{\downdownarrows}}%

\global\long\def\ulrto{\ensuremath{\nearrow}}%

\global\long\def\dlrto{\ensuremath{\searrow}}%


\global\long\def\setmap{\ensuremath{\rightarrow}}%

\global\long\def\elmap{\ensuremath{\mapsto}}%

\global\long\def\compose{\ensuremath{\circ}}%

\global\long\def\cont{C}%

\global\long\def\cadlag{D}%

\global\long\def\Ellp#1{\ensuremath{\mathcal{L}^{#1}}}%


\global\long\def\naturals{\ensuremath{\mathbb{N}}}%

\global\long\def\reals{\mathbb{R}}%

\global\long\def\complex{\mathbb{C}}%

\global\long\def\rationals{\mathbb{Q}}%

\global\long\def\integers{\mathbb{Z}}%


\global\long\def\abs#1{\ensuremath{\left|#1\right|}}%

\global\long\def\smlabs#1{\ensuremath{\lvert#1\rvert}}%
 
\global\long\def\bigabs#1{\ensuremath{\bigl|#1\bigr|}}%
 
\global\long\def\Bigabs#1{\ensuremath{\Bigl|#1\Bigr|}}%
 
\global\long\def\biggabs#1{\ensuremath{\biggl|#1\biggr|}}%

\global\long\def\norm#1{\ensuremath{\left\Vert #1\right\Vert }}%

\global\long\def\smlnorm#1{\ensuremath{\lVert#1\rVert}}%
 
\global\long\def\bignorm#1{\ensuremath{\bigl\|#1\bigr\|}}%
 
\global\long\def\Bignorm#1{\ensuremath{\Bigl\|#1\Bigr\|}}%
 
\global\long\def\biggnorm#1{\ensuremath{\biggl\|#1\biggr\|}}%

\global\long\def\floor#1{\left\lfloor #1\right\rfloor }%
\global\long\def\smlfloor#1{\lfloor#1\rfloor}%

\global\long\def\ceil#1{\left\lceil #1\right\rceil }%
\global\long\def\smlceil#1{\lceil#1\rceil}%


\global\long\def\Union{\ensuremath{\bigcup}}%

\global\long\def\Intsect{\ensuremath{\bigcap}}%

\global\long\def\union{\ensuremath{\cup}}%

\global\long\def\intsect{\ensuremath{\cap}}%

\global\long\def\pset{\ensuremath{\mathcal{P}}}%

\global\long\def\clsr#1{\ensuremath{\overline{#1}}}%

\global\long\def\symd{\ensuremath{\Delta}}%

\global\long\def\intr{\operatorname{int}}%

\global\long\def\cprod{\otimes}%

\global\long\def\Cprod{\bigotimes}%


\global\long\def\smlinprd#1#2{\ensuremath{\langle#1,#2\rangle}}%

\global\long\def\inprd#1#2{\ensuremath{\left\langle #1,#2\right\rangle }}%

\global\long\def\orthog{\ensuremath{\perp}}%

\global\long\def\dirsum{\ensuremath{\oplus}}%


\global\long\def\spn{\operatorname{sp}}%

\global\long\def\rank{\operatorname{rk}}%

\global\long\def\proj{\operatorname{proj}}%

\global\long\def\tr{\operatorname{tr}}%

\global\long\def\vek{\operatorname{vec}}%

\global\long\def\diag{\operatorname{diag}}%

\global\long\def\col{\operatorname{col}}%


\global\long\def\smpl{\ensuremath{\Omega}}%

\global\long\def\elsmp{\ensuremath{\omega}}%

\global\long\def\sigf#1{\mathcal{#1}}%

\global\long\def\sigfield{\ensuremath{\mathcal{F}}}%
\global\long\def\sigfieldg{\ensuremath{\mathcal{G}}}%

\global\long\def\flt#1{\mathcal{#1}}%

\global\long\def\filt{\mathcal{F}}%
\global\long\def\filtg{\mathcal{G}}%

\global\long\def\Borel{\ensuremath{\mathcal{B}}}%

\global\long\def\cyl{\ensuremath{\mathcal{C}}}%

\global\long\def\nulls{\ensuremath{\mathcal{N}}}%

\global\long\def\gauss{\mathfrak{g}}%

\global\long\def\leb{\mathfrak{m}}%


\global\long\def\prob{\ensuremath{\mathbb{P}}}%

\global\long\def\Prob{\ensuremath{\mathbb{P}}}%

\global\long\def\Probs{\mathcal{P}}%

\global\long\def\PROBS{\mathcal{M}}%

\global\long\def\expect{\ensuremath{\mathbb{E}}}%

\global\long\def\Expect{\ensuremath{\mathbb{E}}}%

\global\long\def\probspc{\ensuremath{(\smpl,\filt,\Prob)}}%


\global\long\def\iid{\ensuremath{\textnormal{i.i.d.}}}%

\global\long\def\as{\ensuremath{\textnormal{a.s.}}}%

\global\long\def\asp{\ensuremath{\textnormal{a.s.p.}}}%

\global\long\def\io{\ensuremath{\ensuremath{\textnormal{i.o.}}}}%

\newcommand\independent{\protect\mathpalette{\protect\independenT}{\perp}}
\def\independenT#1#2{\mathrel{\rlap{$#1#2$}\mkern2mu{#1#2}}}

\global\long\def\indep{\independent}%

\global\long\def\distrib{\ensuremath{\sim}}%

\global\long\def\distiid{\ensuremath{\sim_{\iid}}}%

\global\long\def\asydist{\ensuremath{\overset{a}{\distrib}}}%

\global\long\def\inprob{\ensuremath{\overset{p}{\goesto}}}%

\global\long\def\inprobu#1{\ensuremath{\overset{#1}{\goesto}}}%

\global\long\def\inas{\ensuremath{\overset{\as}{\goesto}}}%

\global\long\def\eqas{=_{\as}}%

\global\long\def\inLp#1{\ensuremath{\overset{\Ellp{#1}}{\goesto}}}%

\global\long\def\indist{\ensuremath{\overset{d}{\goesto}}}%

\global\long\def\eqdist{=_{d}}%

\global\long\def\wkc{\ensuremath{\rightsquigarrow}}%

\global\long\def\wkcu#1{\overset{#1}{\ensuremath{\rightsquigarrow}}}%

\global\long\def\plim{\operatorname*{plim}}%


\global\long\def\var{\operatorname{var}}%

\global\long\def\lrvar{\operatorname{lrvar}}%

\global\long\def\cov{\operatorname{cov}}%

\global\long\def\corr{\operatorname{corr}}%

\global\long\def\bias{\operatorname{bias}}%

\global\long\def\MSE{\operatorname{MSE}}%

\global\long\def\med{\operatorname{med}}%

\global\long\def\avar{\operatorname{avar}}%

\global\long\def\se{\operatorname{se}}%

\global\long\def\sd{\operatorname{sd}}%


\global\long\def\nullhyp{H_{0}}%

\global\long\def\althyp{H_{1}}%

\global\long\def\ci{\mathcal{C}}%


\global\long\def\simple{\mathcal{R}}%

\global\long\def\sring{\mathcal{A}}%

\global\long\def\sproc{\mathcal{H}}%

\global\long\def\Wiener{\ensuremath{\mathbb{W}}}%

\global\long\def\sint{\bullet}%

\global\long\def\cv#1{\left\langle #1\right\rangle }%

\global\long\def\smlcv#1{\langle#1\rangle}%

\global\long\def\qv#1{\left[#1\right]}%

\global\long\def\smlqv#1{[#1]}%


\global\long\def\trans{\mathsf{T}}%

\global\long\def\indic{\ensuremath{\mathbf{1}}}%

\global\long\def\Lagr{\mathcal{L}}%

\global\long\def\grad{\nabla}%

\global\long\def\pmin{\ensuremath{\wedge}}%
\global\long\def\Pmin{\ensuremath{\bigwedge}}%

\global\long\def\pmax{\ensuremath{\vee}}%
\global\long\def\Pmax{\ensuremath{\bigvee}}%

\global\long\def\sgn{\operatorname{sgn}}%

\global\long\def\argmin{\operatorname*{argmin}}%

\global\long\def\argmax{\operatorname*{argmax}}%

\global\long\def\Rp{\operatorname{Re}}%

\global\long\def\Ip{\operatorname{Im}}%

\global\long\def\deriv{\ensuremath{\mathrm{d}}}%

\global\long\def\diffnspc{\ensuremath{\deriv}}%

\global\long\def\diff{\ensuremath{\,\deriv}}%

\global\long\def\i{\ensuremath{\mathrm{i}}}%

\global\long\def\e{\mathrm{e}}%

\global\long\def\sep{,\ }%

\global\long\def\defeq{\coloneqq}%

\global\long\def\eqdef{\eqqcolon}%

\global\long\def\err{\varepsilon}%

\global\long\def\mset#1{\mathcal{#1}}%

\global\long\def\largedec#1{\mathbf{#1}}%

\global\long\def\z{\largedec z}%

\newcommandx\Ican[1][usedefault, addprefix=\global, 1=]{I_{#1}^{\ast}}%

\global\long\def\jsr{{\scriptstyle \mathrm{JSR}}}%

\global\long\def\cjsr{{\scriptstyle \mathrm{CJSR}}}%

\global\long\def\rsr{{\scriptstyle \mathrm{RJSR}}}%

\global\long\def\ctspc{\mathscr{M}}%

\global\long\def\b#1{\boldsymbol{#1}}%

\global\long\def\pseudy{\text{\ensuremath{\abv y}}}%

\global\long\def\regcoef{\kappa}%

\global\long\def\adj{\operatorname{adj}}%

\global\long\def\llangle{\dbllangle}%

\global\long\def\rrangle{\dblrangle}%

\global\long\def\smldblangle#1{\ensuremath{\llangle#1\rrangle}}%

\newcommand{\casecens}{{\upshape{(i)}}}

\newcommand{\caseclas}{{\upshape{(ii)}}}

\newcommand{\casestat}{{\upshape{(iii)}}}

\global\long\def\delcens{\mathrm{(i)}}%

\global\long\def\delclas{\mathrm{(ii)}}%

\title{Cointegration with Occasionally Binding Constraints}
\author{\authaffil{James A.\ Duffy\footnotemark[1]{}}{University of Oxford}\hspace{1cm}
\authaffil{Sophocles Mavroeidis\footnotemark[2]{}}{University
of Oxford} \hspace{1cm} \authaffil{Sam Wycherley\footnotemark[3]{}}{Stanford
University}}
\date{\vspace*{0.3cm}September 2025}

\maketitle
\renewcommand*{\thefootnote}{\fnsymbol{footnote}}

\footnotetext[1]{Department\ of Economics and Corpus Christi College;
\texttt{james.duffy@economics.ox.ac.uk}}

\footnotetext[2]{Department of Economics and University College;
\texttt{sophocles.mavroeidis@economics.ox.ac.uk}.}

\footnotetext[3]{Department of Economics; \texttt{wycherley@stanford.edu}}

\renewcommand*{\thefootnote}{\arabic{footnote}}

\setcounter{footnote}{0}
\begin{abstract}
\noindent In the literature on nonlinear cointegration, a long-standing
open problem relates to how a (nonlinear) vector autoregression, which
provides a unified description of the short- and long-run dynamics
of a vector of time series, can generate `nonlinear cointegration'
in the profound sense of those series sharing common nonlinear stochastic
trends. We consider this problem in the setting of the censored and
kinked structural VAR (CKSVAR), which provides a flexible yet tractable
framework within which to model time series that are subject to threshold-type
nonlinearities, such as those arising due to occasionally binding
constraints, of which the zero lower bound (ZLB) on short-term nominal
interest rates provides a leading example. We provide a complete characterisation
of how common linear and \emph{nonlinear} stochastic trends may be
generated in this model, via unit roots and appropriate generalisations
of the usual rank conditions, providing the first extension to date
of the Granger--Johansen representation theorem to a nonlinearly
cointegrated setting, and thereby giving the first successful treatment
of the open problem. The limiting common trend processes include regulated,
censored and kinked Brownian motions, none of which have previously
appeared in the literature on cointegrated VARs. Our results and running
examples illustrate that the CKSVAR is capable of supporting a far
richer variety of long-run behaviour than is a linear VAR, in ways
that may be particularly useful for the identification of structural
parameters. 
\end{abstract}
\vfill

\noindent We thank B.~Beare, P.~Bonomolo, H.~P.~Boswijk, A.~Bykhovskaya,
G.~Cavaliere, R.~Engle, J.~Gao, D.~Harris, X.~Jiao, S.~Johansen,
I.~Kasparis, D.~Kristensen, M.~Kulish, O.~Lieberman, O.~Linton,
Y.~Lu, T\@.~Magdalinos, J.~Morley, U.~M{\"u}ller, L.~Neri,
B.~Nielsen, A.~Onatski, M.~Plagborg-M{\o}ller, W.-K.~Seo, A.~Srakar,
J.~Stock, A.~M.~R.~Taylor, and participants in seminars at Aarhus,
BI Oslo, Cambridge, Cyprus, Melbourne, Monash, Oxford, St.\ Andrews,
Stanford, Sydney, the VTSS workshop, the May 2023 Harvard conference
in honour of Jim Stock and Mark Watson, and the 24th Zaragoza Workshop
in Time Series Econometrics, for helpful comments on earlier drafts
of this work. This research was supported by the European Research
Council via Consolidator Grant 647152.

\thispagestyle{plain}

\pagenumbering{roman}

\newpage{}

\thispagestyle{plain}

\setcounter{tocdepth}{2}

\tableofcontents{}

\newpage{}

\pagenumbering{arabic}

\newcommand{\thmcanonical}{2.1}

\newcommand{\secstability}{4}

\newcommand{\secstabilityabstract}{4.1}

\newcommand{\secregimes}{4.2}

\newcommand{\seclinearcoint}{5}

\newcommand{\thmcostat}{5.1}

\newcommand{\thmcanonvecm}{5.1}

\newcommand{\lemergodic}{B.2}

\section{Introduction}

Nonstationarity, in the form of highly persistent, randomly wandering
time series, is ubiquitous in macroeconomics and finance. It presents
both a challenge for inference (\citealp{Stock1994}; \citealp{Watson1994})
and an opportunity for the identification of dynamic causal effects
(\citealp{BlanchardQuah89}). The canonical framework for modelling
such series is as (common) stochastic trends generated by a linear
vector autoregression (VAR) with unit roots, underpinned by the powerful
Granger--Johansen representation theorem (\citealp{Joh95}, Ch.\ 4).
However, this framework is inadequate when even one of the variables
is subject to an occasionally binding constraint, such as the zero
lower bound (ZLB) constraint on short-term nominal interest rates,
which has recently gained particular prominence in macroeconomic policy
analysis (e.g.\ \citealp{Sum14BE}; \citealp{Will14}; \citealp{EMR19AEJM};
\citealp{Kocherlakota2019}).

In an autoregressive model, an occasionally binding constraint naturally
gives rise to nonlinearity in the form of multiple endogenously switching
regimes, whose presence significantly complicates the links between
unit roots and stochastic trends. While it has become increasingly
common to introduce stochastic trends into (linear) empirical models
of monetary policy by treating the `natural rate of interest' or
`trend inflation' as latent random walks (\citealp{LW03REStat};
\citealp{CS08AER}; \citealp{delNegroGiannoneGiannoniTambalotti2017};
\citealp{AndradeGaliLeBihanMatheron2019}; \citealp{BauerRudebusch2020};
\citealp{SchmittGroheUribe2022}), little is understood about what
would happen to the implied time series properties of the observable
series, such as the actual inflation rate and the nominal rate of
interest, if nonlinearities were introduced into such models via the
ZLB constraint. There is thus a pressing need to extend our understanding
of how stochastic trends may be modelled beyond the linear VAR framework,
to more general settings capable of accommodating such nonlinearities.

This paper addresses this problem in the setting of the censored and
kinked structural VAR (CKSVAR) model (\citealp{SM21,AMSV21}), which
provides a flexible yet tractable framework within which to model
time series that are subject to occasionally binding constraints,
and more generally to threshold-type nonlinearities. In this model,
which is otherwise like a linear structural VAR, one series is allowed
to enter differently according to whether it lies above or below a
threshold; e.g.\ in applications to monetary policy, this series
may be taken to represent the stance of monetary policy in each period,
which above zero coincides with the short-term policy rate, and below
zero corresponds to the (unobserved) `shadow rate'. We provide a
complete characterisation of how common linear and \emph{nonlinear}
stochastic trends may be generated in this model, via unit roots and
appropriate generalisations of the usual rank conditions, providing
the first extension to date of the Granger--Johansen representation
theorem to a nonlinearly cointegrated setting. Our results, which
describe the behaviour of both the short- and long-run components
of the CKSVAR, are foundational for frequentist inference in this
setting, in the presence of common stochastic trends (a treatment
of which will be given elsewhere).

As in a linear VAR, unit roots are unavoidable if one wishes to apply
the CKSVAR to series with stochastic trends. The usual criticism of
simply estimating a model in differences -- that this obliterates
the identifying long-run information carried by the cointegrating
relations -- is here magnified by the threshold nonlinearity in the
model, which dictates whether the affected variable (e.g.\ interest
rates) should enter in levels or differences, and in turn prescribes
appropriate forms for the other variables.\footnote{That these other variables (which enter the model linearly) cannot
generally be replaced by their first differences can be seen from
the vector error correction form of the CKSVAR given in \eqref{vecm-general}
below, from which it is evident that making this replacement (of $x_{t}$
by $\Delta x_{t}$) would amount to imposing the restriction that
$\Pi^{x}=0$.} To put it another way, that nonlinearity prevents series from being
simply `differenced to stationarity'. Moreover, we show that in
our setting, unit roots are not a mere technical nuisance: rather,
their presence may impart significant identifying power to the low
frequency behaviour of the series. For instance, the possibility that
cointegrating relations between series may change as one series crosses
a threshold, something accommodated by the CKSVAR, may be utilised
to test hypotheses about the relative effectiveness of unconventional
monetary policy (see Examples~\ref{exa:natrate} and \ref{exa:infldrift}
below; this is a problem that has been studied econometrically by
e.g.\ \citealp{GHP14JCMB}; \citealp{WX16JCMB}; \citealp{DGG20NBER};
and \citealp{ILMZ20}).

In analysing the CKSVAR with unit roots, we make a major contribution
to the literature on nonlinear cointegration. Here a long-standing
open problem relates to how a (nonlinear) vector autoregression, which
provides a unified description of the short- and long-run dynamics
of a collection of time series, can generate nonlinear cointegration
between those series -- where `nonlinear cointegration' is understood
in the profound sense of those series having common \emph{nonlinear}
stochastic trends with possibly \emph{nonlinear} cointegrating relations
between those trends. As discussed in the recent review by \citet{Tjo20EctRev},
despite the voluminous literature on the subject of `nonlinear cointegration',
this problem has yet to be addressed at any reasonable level of generality.\footnote{While \citet{CGT17JBES} make an important effort in this direction,
their results are limited to a first-order bivariate VAR with two
regimes, in which one of those regimes is delimited by a compact set,
and so makes a negligible contribution to the long-run behaviour of
the series generated by the model. Their results are thus markedly
different from those obtained below.} Within the framework of the CKSVAR, we provide a resolution of this
problem, showing that the model naturally gives rise to nonlinear
cointegration, generating nonlinear common trend processes -- \emph{censored},
\emph{regulated}, and \emph{kinked} Brownian motions (see \defref{nonlinearBMs}
below) -- that have not previously appeared in multivariate settings.

To clarify how our work relates to the existing literature on `nonlinear
cointegration', and to explain why we have been able to make progress
in an area that has previously seemed intractable, we briefly recall
the two main strands of that literature. One strand (\citealp{Tjo20EctRev},
p.~657) starts from the vector error correction model (VECM) representation
of a cointegrated VAR, and introduces nonlinearity into the error
correction mechanism; a prototypical model is
\begin{equation}
\Delta z_{t}=g(\beta^{\trans}z_{t-1})+\sum_{i=1}^{k-1}\Gamma_{i}\Delta z_{t-i}+u_{t},\label{eq:nlvecm}
\end{equation}
in which the usually linear loadings $\alpha[\beta^{\trans}z_{t-1}]=\alpha\xi_{t-1}$
on the equilibrium errors are replaced by a general nonlinear function.
In the original `threshold cointegration' conception of this model,
due to \citet{BF97IER}, $g$ is piecewise linear, i.e.\  $g(\xi_{t})=\sum_{i=1}^{m}\alpha^{(i)}\indic\{\xi_{t}\in\Xi_{i}\}\xi_{t}$,
where each of the $\alpha^{(i)}$'s correspond to different `regimes'.
The values of $\{\Gamma_{i}\}$ may also depend on lags of $\xi_{t}$
or $\Delta z_{t}$. (For regime-switching versions, including of
the smoothed variety, see e.g.\ \citealp{HS02JoE}; \citealp{Saik05JoE,Saik08ET};
and \citealp{Seo11ET}; for versions in which $g$ is allowed to be
a more general nonlinear function, but the $\{\Gamma_{i}\}$ matrices
are fixed, see \citealp{EM02JTSA}; and \citealp{KR10JoE,KR13ET}).
Notably, the nonlinearity in such models is wholly confined to the
short-run dynamics: as in a linearly cointegrated VAR, there remains
a globally defined cointegrating space spanned by the columns of $\beta$
(i.e.\ which is common to all `regimes'), and the limiting common
trends remain a (vector) Brownian motion.

The other strand (\citealp{Tjo20EctRev}, pp.~658--666) takes as
its starting point the triangular representation of a linearly cointegrated
system, and introduces nonlinearity directly into the common trends,
by specifying
\begin{align}
y_{t} & =f(x_{t})+\varepsilon_{yt}, & x_{t} & =x_{t-1}+\varepsilon_{xt}.\label{eq:nlcointreg}
\end{align}
Here $f(x_{t})$ replaces what would ordinarily be a linear function,
with the consequence that the weak limit of $Y_{n}(\lambda)\defeq n^{-1/2}y_{\smlfloor{n\lambda}}$
will now be a nonlinear transformation of the limiting Brownian motions
associated with $X_{n}(\lambda)\defeq n^{-1/2}x_{\smlfloor{n\lambda}}$.
The errors $\varepsilon_{t}=(\varepsilon_{yt},\varepsilon_{xt})$
may be weakly dependent and cross-correlated, permitting $\{x_{t}\}$
to be endogenous. The function $f$ is typically estimated via some
sort of regression, either parametrically (\citealp{PP99ET,PP01Ecta};
\citealp{CW15JoE}; \citealp{LTG16AS}) or nonparametrically (\citealp{KMT07AS};
\citealp{WP09Ecta,WP16ET}; \citealp{Duffy17ET}; \citealp{DK21AS});
there is also literature on specification testing in this setting
(e.g.\ \citealp{WP12AS}; \citealp{DGTY17JoE}; \citealp{WWZ18JoE};
\citealp{BRN20ET}). Notable variants have used $f$ to model transitions
between regimes with distinct linear cointegrating relations (\citealp{SC04ET};
\citealp{GP06OBES}), or allowed it to take the `functional coefficient'
form $\beta(w_{t})x_{t}$ (\citealp{Xia09JoE,PW23JoE}).

In developing a VAR model that exhibits both nonlinearity in its short-run
dynamics, as in \eqref{nlvecm}, and in the implied (long-run) common
trends, as in \eqref{nlcointreg}, this paper is the first to bridge
the remarkably wide gulf that has existed between these two strands
of the literature. The CKSVAR turns out to provide just enough flexibility
to accommodate meaningful departures from linear cointegration, while
retaining a sufficiently nice structure to be tractable. We show that
depending on the rank conditions imposed on (submatrices of) the autoregressive
polynomial evaluated at unity, the CKSVAR is capable of generating
three distinct kinds of nonlinear cointegration, which we term: (i)
regulated cointegration; (ii) kinked cointegration; and (iii) linear
cointegration in a nonlinear VECM.

At a technical level, our contribution consists of identifying the
alternative configurations of the model that give rise to cases~(i)--(iii),
which are essentially exhaustive of the possibilities here, and deriving
analogues of the Granger--Johansen representation theorem in these
three cases.\footnote{Our analysis is exhaustive with respect to the possibilities for generating
series that are integrated of order one (in a suitably extended sense
of the term; see \defref{ast-integration}) within the CKSVAR; higher
orders of integration are not considered here.} In analysing case~(i), our work relates to that of \citet{Cav05ET},
\citet{LLS11Bern}, \citet{GLY13JoE}, and most closely to \citet{BD22},
all of whom obtain convergence to regulated Brownian motions in univariate
models. Case~(ii), even in the univariate setting, does not appear
to have been considered by any previous literature. While cases~(i)
and (ii) describe phenomena that are entirely new to the literature,
case~(iii) holds under a configuration of the model that falls within
the very broad class of nonlinear VECMs considered by \citet{Saik08ET}:
but even here, our results regarding the ergodicity of the short-memory
components of the model extend his, insofar as we are able to exploit
certain properties of the CKSVAR that are not shared by all the models
encompassed by his general framework.

We illustrate the economic significance of cases~(i) and (ii) by
demonstrating how each may arise in our running example of a stylised
structural model of monetary policy in the presence of a zero lower
bound, contingent on the values taken by certain model parameters.
We further apply our results to determine the long-run properties
of the structural model of \citet{ABH23mimeo}, for which the presence
of nonlinear transformations of potentially stochastically trending
series (due to a `long-run Phillips curve') precludes the application
of any pre-existing version of the Granger--Johansen representation
theorem.

The remainder of the paper is organised as follows. \secref{model}
introduces the CKSVAR model and the stylised structural macroeconomic
models that we use as running examples. \secref{cointegration} develops
the heuristics of the model with unit roots, outlining the tripartite
classification noted above. The main results of this paper, which
extend the Granger--Johansen representation theorem to the CKSVAR
model, are given in \secref{representation}. All proofs appear in
the appendices.

\begin{notation*}
$e_{m,i}$ denotes the $i$th column of an $m\times m$ identity matrix;
when $m$ is clear from the context, we write this simply as $e_{i}$.
In a statement such as $f(a^{\pm},b^{\pm})=0$, the notation `$\pm$'
signifies that both $f(a^{+},b^{+})=0$ and $f(a^{-},b^{-})=0$ hold;
similarly, `$a^{\pm}\in A$' denotes that both $a^{+}$ and $a^{-}$
are elements of $A$. All limits are taken as $n\goesto\infty$ unless
otherwise stated. $\inprob$ and $\wkc$ respectively denote convergence
in probability and in distribution (weak convergence). We write `$X_{n}(\lambda)\wkc X(\lambda)$
on $D_{\reals^{m}}[0,1]$' to denote that $\{X_{n}\}$ converges
weakly to $X$, where these are considered as random elements of $D_{\reals^{m}}[0,1]$,
the space of cadlag functions $[0,1]\setmap\reals^{m}$, equipped
with the uniform topology; we denote this as $D[0,1]$ whenever the
value of $m$ is clear from the context. $\smlnorm{\cdot}$ denotes
the Euclidean norm on $\reals^{m}$, and the matrix norm that it induces.
For $X$ a random vector and $p\geq1$, $\smlnorm X_{p}\defeq(\expect\smlnorm X^{p})^{1/p}$.
\end{notation*}

\section{Model: the censored and kinked SVAR}

\label{sec:model}

We consider a VAR($k$) model in $p$ variables, in which one series,
$y_{t}$, enters with coefficients that differ according to whether
it is above or below a time-invariant threshold $b$, while the other
$p-1$ series, collected in $x_{t}$, enter linearly. Defining
\begin{align}
y_{t}^{+} & \defeq\max\{y_{t},b\} & y_{t}^{-} & \defeq\min\{y_{t},b\},\label{eq:y-threshold}
\end{align}
we specify that $(y_{t},x_{t})$ follow 
\begin{equation}
\phi_{0}^{+}y_{t}^{+}+\phi_{0}^{-}y_{t}^{-}+\Phi_{0}^{x}x_{t}=c+\sum_{i=1}^{k}[\phi_{i}^{+}y_{t-i}^{+}+\phi_{i}^{-}y_{t-i}^{-}+\Phi_{i}^{x}x_{t-i}]+u_{t}\label{eq:var-two-sided}
\end{equation}
or, more compactly,
\begin{equation}
\phi^{+}(L)y_{t}^{+}+\phi^{-}(L)y_{t}^{-}+\Phi^{x}(L)x_{t}=c+u_{t},\label{eq:var-pm}
\end{equation}
where 
\begin{align*}
\phi^{\pm}(L) & \defeq\phi_{0}^{\pm}-\sum_{i=1}^{k}\phi_{i}^{\pm}L^{i} & \Phi^{x}(L) & \defeq\Phi_{0}^{x}-\sum_{i=1}^{k}\Phi_{i}^{x}L^{i},
\end{align*}
for $\phi_{i}^{\pm}\in\reals^{p\times1}$ and $\Phi_{i}^{x}\in\reals^{p\times(p-1)}$,
and $L$ denotes the lag operator. As in a linear SVAR, $\{u_{t}\}$
may be an i.i.d.\ sequence of mutually orthogonal structural shocks,
but our results below also permit them to be cross-correlated or weakly
dependent. Through an appropriate redefinition of $y_{t}$ and $c$,
we may take $b=0$ without loss of generality, and will do so throughout
the sequel.\footnote{Defining $y_{b,t}\defeq y_{t}-b$, $y_{b,t}^{+}\defeq\max\{y_{b,t},0\}$,
$y_{b,t}^{-}\defeq\min\{y_{b,t},0\}$ and $c_{b}\defeq c-[\phi^{+}(1)+\phi^{-}(1)]b$,
we can rewrite \eqref{var-pm} as
\[
\phi^{+}(L)y_{b,t}^{+}+\phi^{-}(L)y_{b,t}^{-}+\Phi^{x}(L)x_{t}=c_{b}+u_{t}.
\]
} In this case, $y_{t}^{+}$ and $y_{t}^{-}$ respectively equal the
positive and negative parts of $y_{t}$, and $y_{t}=y_{t}^{+}+y_{t}^{-}$.
(Throughout the following, the notation `$a^{\pm}$' connotes $a^{+}$
and $a^{-}$ as objects associated respectively with $y_{t}^{+}$
and $y_{t}^{-}$, or their lags. If we want to instead denote the
positive and negative parts of some $a\in\reals$, we shall do so
by writing $[a]_{+}\defeq\max\{a,0\}$ or $[a]_{-}\defeq\min\{a,0\}$.)

Models of the form of \eqref{var-pm} have previously been employed
in the literature to account for the dynamic effects of censoring,
occasionally binding constraints, and endogenous regime switching.
\citet{SM21} proposed exactly this model, which he termed the \emph{censored
and kinked structural VAR} (CKSVAR) model, to describe the operation
of monetary policy during periods when a zero lower bound may bind
on the policy rate: in our notation, $y_{t}$ corresponds to his `shadow
rate', expressing the central bank's desired policy stance, and
$y_{t}^{+}$ to the actual policy rate. \citet{AMSV21} considered
a model in which one variable is subject to an occasionally binding
constraint, which although in its initial formulation is somewhat
more general, reduces to an instance of the CKSVAR once the conditions
necessary for the model to have a unique solution (for all values
of $u_{t}$) have been imposed (see their Proposition~1(i)). This
version of their model -- i.e.\ that in which the `private sector
regression functions' are piecewise linear and continuous -- is
thus accommodated by \eqref{var-pm}.\footnote{See also \citet{ACHSV21RED}, for a DSGE model with an occasionally
binding constraint, in which agents' decision rules are approximated
by functions with these properties.}

To put some economic flesh on the bones of the representation theory
developed in Sections~\ref{sec:cointegration} and \ref{sec:representation}
below, we here introduce the running example of a stylised structural
model of monetary policy in the presence of a zero lower bound (ZLB).
This model provides a simple, economically interpretable framework
in which we may illustrate the various forms of novel long-run behaviour
permitted by the CKSVAR, by considering alternate parametrisations
of the model. Moreover, as discussed in \secref{cointegration} below,
the model elucidates how that long-run behaviour may provide identifying
information on the relative effectiveness of unconventional monetary
policy (as compared with conventional rate-setting policy), i.e.\ on
whether the zero lower bound really constrains the ability of a central
bank to target inflation.

\needspace{3\baselineskip}
\begin{example}
\label{exa:monetary}Consider the following stylised structural model,
a simplified version of the model of \citet{ILMZ20},  consisting
of a composite IS and Phillips curve (PC) equation
\begin{align}
\pi_{t}-\abv{\pi}_{t} & =\theta[i_{t}^{+}+\mu i_{t}^{-}-(r_{t}^{\ast}+\abv{\pi}_{t})]+\varepsilon_{t}\label{eq:is-pc}
\end{align}
and a policy reaction function (Taylor rule)
\begin{equation}
i_{t}=(r_{t}^{\ast}+\abv{\pi}_{t})+\gamma(\pi_{t}-\abv{\pi}_{t}),\label{eq:taylor}
\end{equation}
where $r_{t}^{\ast}$ denotes the (real) natural rate of interest,
$\abv{\pi}_{t}$ the central bank's inflation target, $\pi_{t}$
inflation, and $\varepsilon_{t}$ a mean zero, i.i.d.\ innovation.
$i_{t}$ measures the stance of monetary policy; thus $i_{t}^{+}\defeq[i_{t}]_{+}$
gives the actual policy rate (constrained to be non-negative), and
$i_{t}^{-}\defeq[i_{t}]_{-}$ the desired stance of policy when the
ZLB binds, to be effected via some form of `unconventional' monetary
policy, such as long-term asset purchases. We maintain that $\gamma>0$
and $\theta<0$. The parameter $\mu\in[0,1]$ reflects the relative
efficacy of unconventional policy, with $\mu=1$ if this is as effective
as conventional policy.

To `close' the model, we consider two alternative specifications
for the underlying processes followed by $\{r_{t}^{\ast}\}$ and $\{\abv{\pi}_{t}\}$.
In the first of these \exname{\ref*{exa:monetary}a}\refstepcounter{examplex}(henceforth,
\label{exa:natrate}\textbf{Example~\theexamplex{}}) the inflation
target is assumed to be constant and is normalised to zero (i.e.\ $\abv{\pi}_{t}=\abv{\pi}=0$),
while the natural real rate of interest follows a random walk AR(1)
process (as in the model of \citealp{LW03REStat}),
\begin{equation}
r_{t}^{\ast}=r_{t-1}^{\ast}+\eta_{t}\label{eq:natural-rate}
\end{equation}
where $\eta_{t}$ is an i.i.d.\ mean zero innovation, possibly correlated
with $\varepsilon_{t}$. Substituting \eqref{taylor} into \eqref{is-pc}
and \eqref{natural-rate}, we render the system as a CKSVAR for $(i_{t},\pi_{t})$
as
\begin{equation}
\begin{bmatrix}1 & 1 & -\gamma\\
0 & \theta(1-\mu) & 1-\theta\gamma
\end{bmatrix}\begin{bmatrix}i_{t}^{+}\\
i_{t}^{-}\\
\pi_{t}
\end{bmatrix}=\begin{bmatrix}1 & 1 & -\gamma\\
0 & 0 & 0
\end{bmatrix}\begin{bmatrix}i_{t-1}^{+}\\
i_{t-1}^{-}\\
\pi_{t-1}
\end{bmatrix}+\begin{bmatrix}\eta_{t}\\
\varepsilon_{t}
\end{bmatrix}.\label{eq:cksvar-natrate}
\end{equation}
This model will provide an illustration of the second kind of nonlinear
cointegration developed in \secref{cointegration} below.

In the second variant of the model \exname{\ref*{exa:monetary}b}\refstepcounter{examplex}(henceforth,
\label{exa:infldrift}\textbf{Example~\theexamplex{}}) the natural
rate is assumed to be constant and, for simplicity of exposition,
normalised to zero (i.e.\ $r_{t}^{\ast}=r^{\ast}=0$), while the
inflation target is allowed to be time-varying, according to
\begin{equation}
\abv{\pi}_{t}=\abv{\pi}_{t-1}+\delta(\pi_{t-1}-\abv{\pi}_{t-1})+\eta_{t}\label{eq:infltarget}
\end{equation}
where $\delta\in(-1,0]$, and $\eta_{t}$ is an i.i.d.\ innovation
as above. When $\delta=0$, this corresponds to a model in which the
inflation target follows a pure random walk, possibly reflecting the
time-varying preferences of the central bank (cf.\ \citealp{CS08AER});
when $\delta<0$, the model allows past deviations of inflation from
target to feed back into the target, such that e.g.\ below-target
inflation induces an \emph{upward} revision of the inflation target.
Motivation for this aspect of the model comes from the manner in which
the ZLB may constrain policy to be excessively deflationary for a
sustained period, something that has prompted the literature to consider
the costs and benefits of adopting a higher inflation target (e.g.\ \citealp[pp.~207f.]{BDM10JMCB};
\citealp{CGW10RES}). Supposing additionally that $\gamma>1$, we
may put \eqref{is-pc}, \eqref{taylor} and \eqref{infltarget} in
the form of a CKSVAR as
\begin{equation}
\begin{bmatrix}-1 & -1 & \gamma\\
\varphi_{1} & \varphi_{\mu} & -\varphi_{1}
\end{bmatrix}\begin{bmatrix}i_{t}^{+}\\
i_{t}^{-}\\
\pi_{t}
\end{bmatrix}=\begin{bmatrix}\delta-1 & \delta-1 & \gamma-\delta\\
0 & 0 & 0
\end{bmatrix}\begin{bmatrix}i_{t-1}^{+}\\
i_{t-1}^{-}\\
\pi_{t-1}
\end{bmatrix}+(\gamma-1)\begin{bmatrix}\eta_{t}\\
\varepsilon_{t}
\end{bmatrix}\label{eq:cksvar-infldrift}
\end{equation}
where $\varphi_{\mu}\defeq(1-\mu\theta\gamma)-\theta(1-\mu)$ and
so $\varphi_{1}=1-\theta\gamma$. Depending on the assumptions made
on the model parameters (in particular $\delta$), this model is capable
of generating either of the first two types of nonlinear cointegration
discussed in \secref{cointegration}.
\end{example}
\addtocounter{examplex}{-2}

While both \citet{SM21} and \citet{AMSV21} motivate and interpret
\eqref{var-pm} as a structural model, empirically motivated reduced-form
models of this kind have also appeared in the literature, particularly
in the univariate ($p=1$) case of \eqref{var-pm}, which encompasses
the dynamic Tobit model (\citealt[p.~186]{Maddala83}; for applications,
see e.g.\ \citealp{DJ02FRB}; \citealp{DJH11}; \citealp{BMMV21JBF};
and \citealp{Byk21JBES}).
\begin{example}[univariate]
\label{exa:dyntobit} Consider \eqref{var-pm} with $p=1$ and $\phi_{0}^{+}=\phi_{0}^{-}=1$,
so that
\begin{equation}
y_{t}=c+\sum_{i=1}^{k}(\phi_{i}^{+}y_{t-i}^{+}+\phi_{i}^{-}y_{t-i}^{-})+u_{t}.\label{eq:univariate-case}
\end{equation}
In the nomenclature of \citet[Sec.~1]{BD22}, if $\phi_{i}^{-}=0$
for all $i\in\{1,\ldots,k\}$, so that only the positive part of $y_{t-i}$
enters the r.h.s., then 
\begin{equation}
y_{t}^{+}=\left[c+\sum_{i=1}^{k}\phi_{i}^{+}y_{t-i}^{+}+u_{t}\right]_{+}\label{eq:censored-Tobit}
\end{equation}
follows a `censored' dynamic Tobit.
\end{example}
We follow \citet{SM21} and \citet{AMSV21} in maintaining the following,
which are necessary and sufficient to ensure that \eqref{var-pm}
has a unique solution for $(y_{t},x_{t})$, for all possible values
of $u_{t}$. Define 
\[
\Phi_{0}\defeq\begin{bmatrix}\phi_{0}^{+} & \phi_{0}^{-} & \Phi_{0}^{x}\end{bmatrix}=\begin{bmatrix}\phi_{0,yy}^{+} & \phi_{0,yy}^{-} & \phi_{0,yx}^{\trans}\\
\phi_{0,xy}^{+} & \phi_{0,xy}^{-} & \Phi_{0,xx}
\end{bmatrix},
\]
$\Phi_{0}^{+}\defeq[\phi_{0}^{+},\Phi_{0}^{x}]$ and $\Phi_{0}^{-}\defeq[\phi_{0}^{-},\Phi_{0}^{x}]$.

\assumpname{DGP}
\begin{assumption}
\label{ass:dgp}~
\begin{enumerate}[label=\ass{\arabic*.}, ref=\ass{.\arabic*}, itemsep=1pt,topsep=2pt]
\item \label{enu:dgp:defn} $\{(y_{t},x_{t})\}$ are generated according
to \eqref{y-threshold}--\eqref{var-pm} with $b=0$, with (possibly
random) initial values $(y_{i},x_{i})$, for $i\in\{-k+1,\ldots,0\}$;
\item \label{enu:dgp:coherence} $\sgn(\det\Phi_{0}^{+})=\sgn(\det\Phi_{0}^{-})\neq0$.
\item \label{enu:dgp:wlog} $\Phi_{0,xx}$ is invertible, and
\[
\sgn\{\phi_{0,yy}^{+}-\phi_{0,yx}^{\trans}\Phi_{0,xx}^{-1}\phi_{0,xy}^{+}\}=\sgn\{\phi_{0,yy}^{-}-\phi_{0,yx}^{\trans}\Phi_{0,xx}^{-1}\phi_{0,xy}^{-}\}>0.
\]
\end{enumerate}
\end{assumption}
For a further discussion of these conditions, including why \assref{dgp}\ref{enu:dgp:wlog}
may be maintained without loss of generality when \assref{dgp}\ref{enu:dgp:coherence}
holds, see \citet[Sec.~2]{DMW23stat}. As in that paper, we shall
designate a CKSVAR as \emph{canonical} if\tightermath{0.8}{5}{
\begin{equation}
\Phi_{0}=\begin{bmatrix}1 & 1 & 0\\
0 & 0 & I_{p-1}
\end{bmatrix}\eqdef\Ican[p].\label{eq:canonical}
\end{equation}
}While it is not always the case that the reduced form of \eqref{var-pm}
corresponds directly to a canonical CKSVAR, by defining the canonical
variables
\begin{equation}
\begin{bmatrix}\tilde{y}_{t}^{+}\\
\tilde{y}_{t}^{-}\\
\tilde{x}_{t}
\end{bmatrix}\defeq\begin{bmatrix}\bar{\phi}_{0,yy}^{+} & 0 & 0\\
0 & \bar{\phi}_{0,yy}^{-} & 0\\
\phi_{0,xy}^{+} & \phi_{0,xy}^{-} & \Phi_{0,xx}
\end{bmatrix}\begin{bmatrix}y_{t}^{+}\\
y_{t}^{-}\\
x_{t}
\end{bmatrix}\eqdef P^{-1}\begin{bmatrix}y_{t}^{+}\\
y_{t}^{-}\\
x_{t}
\end{bmatrix},\label{eq:canon-vars}
\end{equation}
where $\bar{\phi}_{0,yy}^{\pm}\defeq\phi_{0,yy}^{\pm}-\phi_{0,yx}^{\trans}\Phi_{0,xx}^{-1}\phi_{0,xy}^{\pm}>0$
and $P^{-1}$ is invertible under \ref{ass:dgp}; and setting\tightermath{0.8}{5}{
\begin{equation}
\begin{bmatrix}\tilde{\phi}^{+}(\lambda) & \tilde{\phi}^{-}(\lambda) & \tilde{\Phi}^{x}(\lambda)\end{bmatrix}\defeq Q\begin{bmatrix}\phi^{+}(\lambda) & \phi^{-}(\lambda) & \Phi^{x}(\lambda)\end{bmatrix}P,\label{eq:canon-polys}
\end{equation}
}where 
\begin{equation}
Q\defeq\begin{bmatrix}1 & -\phi_{0,yx}^{\trans}\Phi_{0,xx}^{-1}\\
0 & I_{p-1}
\end{bmatrix},\label{eq:Q-canon}
\end{equation}
we obtain a canonical CKSVAR for $(\tilde{y}_{t},\tilde{x}_{t})$.
This is formalised by the following, which reproduces the first part
of Proposition~\thmcanonical{} in \citet{DMW23stat}.
\begin{prop}
\label{prop:canonical}Suppose \ref{ass:dgp} holds. Then there exist
$(\tilde{y}_{t},\tilde{x}_{t})$ such that \eqref{canon-vars}--\eqref{canon-polys}
hold, $\tilde{y}_{t}^{+}=\max\{\tilde{y}_{t},0\}$, $\tilde{y}_{t}^{-}=\min\{\tilde{y}_{t},0\}$
and
\begin{equation}
\tilde{\phi}^{+}(L)\tilde{y}_{t}^{+}+\tilde{\phi}^{-}(L)\tilde{y}_{t}^{-}+\tilde{\Phi}^{x}(L)\tilde{x}_{t}=\tilde{c}+\tilde{u}_{t},\label{eq:tildeVAR}
\end{equation}
is a canonical CKSVAR, where $\tilde{c}=Qc$ and $\tilde{u}_{t}=Qu_{t}$.
\end{prop}

To distinguish between a general CKSVAR in which possibly $\Phi_{0}\neq\Ican[p]$,
and its associated canonical form as given by \propref{canonical},
we shall refer to the former as the `structural form' of the CKSVAR.
Since the time series properties of a general CKSVAR are largely inherited
from its derived canonical form, we shall often work with this more
convenient representation of the system, and indicate this as follows.

\assumpname{DGP$^{\ast}$}
\begin{assumption}
\label{ass:dgp-canon}$\{(y_{t},x_{t})\}$ are generated by a canonical
CKSVAR, i.e.\ \assref{dgp} holds with $\Phi_{0}=[\phi_{0}^{+},\phi_{0}^{-},\Phi^{x}]=\Ican[p]$,
so that \eqref{var-two-sided} may be equivalently written as 
\begin{equation}
\begin{bmatrix}y_{t}\\
x_{t}
\end{bmatrix}=c+\sum_{i=1}^{k}\begin{bmatrix}\phi_{i}^{+} & \phi_{i}^{-} & \Phi_{i}^{x}\end{bmatrix}\begin{bmatrix}y_{t-i}^{+}\\
y_{t-i}^{-}\\
x_{t-i}
\end{bmatrix}+u_{t}.\label{eq:canon-var}
\end{equation}
\end{assumption}
\saveexamplex{}

\section{Unit roots and nonlinear cointegration: heuristics}

\label{sec:cointegration}

\subsection{Nonlinearity and cointegration}

\label{subsec:nonlinear-coint}

It is well known that a linear VAR can faithfully replicate the high
persistence, random wandering and long-run co-movement that is characteristic
of a great many macroeconomic time series, via the imposition of unit
autoregressive roots and the familiar rank conditions (\citealp{Joh95}).
The question thus arises as to whether, and how, such behaviour may
also be generated within a CKSVAR, so that the model might still be
applied to series for which a stationary CKSVAR would be inappropriate.
As we shall see, it is possible not only to accommodate linear cointegration
within the CKSVAR, but also to generate a variety of \emph{nonlinear}
forms of cointegration, owing to the richer class of common trend
processes that the model supports. Moreover, as our examples below
illustrate, such departures from linear cointegration may also aid
in the identification of structural parameters.

In developing the CKSVAR with unit roots, we shall find it necessary
to depart from the usual classification of processes according to
their orders of integration, since the nonlinearity in the model generally
prevents it from generating series that are difference stationary.
This is an issue commonly encountered in regime-switching cointegration
models: see e.g.\ the discussion in \citet[pp.~816f.]{GP06OBES},
where this motivates the definition of the `order of summability'
of a time series, and the allied notion of `co-summability' as a
generalisation of linear cointegration (\citealp{BRG14JoE}, pp.\ 335f.).
While those concepts could well be applied to the CKSVAR, the following
properties, which may be more easily verified, will suffice for our
purposes.
\begin{defn}
\label{def:ast-integration}Let $\{w_{t}\}_{t\in\naturals}$ be a
random sequence taking values in $\reals^{p}$. We say that $\{w_{t}\}$
is:
\begin{enumerate}[itemsep=2pt,topsep=3pt]
\item \emph{$\ast$-stationary}, denoted $w_{t}\sim I^{\ast}(0)$, if $\sup_{1\leq t\leq n}\smlnorm{w_{t}}=o_{p}(n^{1/2})$;
or
\item \emph{\label{enu:-star-integrated}$\ast$-integrated} (\emph{of order
one}), denoted $w_{t}\sim I^{\ast}(1)$, if $n^{-1/2}w_{\smlfloor{n\lambda}}\wkc\ell(\lambda)$
on $D_{\reals^{p}}[0,1]$, where $\ell$ is a non-degenerate stochastic
process with continuous sample paths;
\end{enumerate}
and analogously for subvectors (and individual elements) of $\{w_{t}\}$.
\end{defn}
The preceding relates to the discussion of `cointegration' in \citet{MW13JoE},
who consider a model with linear cointegration in which the common
trends belong to a broad class of processes satisfying a weak convergence
criterion (see their eq.\ (2)), which includes strict $I(1)$ and
local-to-unity models as special cases, so that the vector of time
series generated by the model is $I^{\ast}(1)$ as defined above.
The bounded unit root processes of \citet{Cav05ET} also provide an
example of a series that converges weakly upon standardisation by
$n^{-1/2}$, and so are $I^{\ast}(1)$, but not $I(1)$.

\saveexamplex{}

\exname{\ref*{exa:dyntobit}}
\begin{example}[univariate; ctd]
 Within the CKSVAR framework, a simple, non-trivial example of series
that are $I^{\ast}(d)$ but not $I(d)$ is provided by the censored
dynamic Tobit. When $\sum_{i=1}^{k}\phi_{i}^{+}=1$, this model has
a unit root, and \citet{BD22} show that $n^{-1/2}y_{\smlfloor{n\lambda}}^{+}\wkc Y^{+}(\lambda)$
on $D[0,1]$, where $Y^{+}$ is a Brownian motion regulated (at zero)
from below (their Theorem~3.2; see \defref{nonlinearBMs} below),
and $\smlnorm{\Delta y_{t}^{+}}_{2+\delta_{u}}$ is uniformly bounded
(their Lemma~B.2). Thus $\Delta y_{t}^{+}\sim I^{\ast}(0)$ and $y_{t}^{+}\sim I^{\ast}(1)$,
even though, due to the nonlinearity in the model, neither series
are $I(d)$ for any $d$.
\end{example}
\restoreexamplex{}

Here we also need an enlarged notion of `cointegration' that is
sufficiently general to encompass the possibilities of nonlinear cointegration
accommodated by the CKSVAR model, particularly for the analysis of
`case\ \caseclas{}' below, such as is provided by the following
(cf.\ \citealp{GP06OBES}, p.~817).
\begin{defn}
\label{def:coint} Let $\mathscr{Z}^{+}\defeq\{(y,x)\in\reals^{p}\mid y\geq0\}$
and $\mathscr{Z}^{-}\defeq\{(y,x)\in\reals^{p}\mid y\leq0\}$, $r^{\pm}\in\{0,\ldots,p-1\}$,
and $\beta^{\pm}\in\reals^{p\times r^{\pm}}$ have full column rank.
Suppose $z_{t}=(y_{t},x_{t}^{\trans})^{\trans}\sim I^{\ast}(1)$,
but
\[
\indic\{z_{t}\in\mathscr{Z}^{(i)}\}\theta^{\trans}z_{t}\sim I^{\ast}(0)\iff\theta\in\spn\beta^{(i)}
\]
for $(i)\in\{+,-\}$. Then $z_{t}$ is said to be \emph{cointegrated}
\emph{on} $\mathscr{Z}^{(i)}$, with $r^{(i)}$ the \emph{cointegrating
rank on} $\mathscr{Z}^{(i)}$, $\spn\beta^{(i)}$ the \emph{cointegrating
space} \emph{on $\mathscr{Z}^{(i)}$}, and any (nonzero) element of
$\spn\beta^{(i)}$ a \emph{cointegrating vector on} $\mathscr{Z}^{(i)}$,
for $(i)\in\{+,-\}$. If $\beta^{(i)}$ does not depend on $(i)$,
we drop the `on $\mathscr{Z}^{(i)}$' qualifiers.
\end{defn}

\subsection{The CKSVAR with unit roots}

Our next step is to rewrite the CKSVAR, as in \eqref{var-two-sided}
or \eqref{var-pm} above, in the form of a vector error-correction
model (VECM). Define the autoregressive polynomials\tightermath{0.8}{5}{
\[
\Phi^{\pm}(\lambda)\defeq\begin{bmatrix}\phi^{\pm}(\lambda) & \Phi^{x}(\lambda)\end{bmatrix},
\]
}and let $\Gamma_{i}^{\pm}\defeq-\sum_{j=i+1}^{k}\Phi_{j}^{\pm}\eqdef[\gamma_{i}^{\pm},\Gamma^{x}]$
for $i\in\{1,\ldots,k-1\}$, so that $\Gamma^{\pm}(\lambda)\defeq\Phi_{0}^{\pm}-\sum_{i=1}^{k-1}\Gamma_{i}^{\pm}\lambda^{i}$
is such that
\[
\Phi^{\pm}(\lambda)=\Phi^{\pm}(1)\lambda+\Gamma^{\pm}(\lambda)(1-\lambda).
\]
Set $\pi^{\pm}\defeq-\phi^{\pm}(1)$ and $\Pi^{x}\defeq-\Phi^{x}(1)$.
Then
\begin{equation}
\Phi_{0}\begin{bmatrix}\Delta y_{t}^{+}\\
\Delta y_{t}^{-}\\
\Delta x_{t}
\end{bmatrix}=c+\begin{bmatrix}\pi^{+} & \pi^{-} & \Pi^{x}\end{bmatrix}\begin{bmatrix}y_{t-1}^{+}\\
y_{t-1}^{-}\\
x_{t-1}
\end{bmatrix}+\sum_{i=1}^{k-1}\begin{bmatrix}\gamma_{i}^{+} & \gamma_{i}^{-} & \Gamma_{i}^{x}\end{bmatrix}\begin{bmatrix}\Delta y_{t-i}^{+}\\
\Delta y_{t-i}^{-}\\
\Delta x_{t-i}
\end{bmatrix}+u_{t},\label{eq:vecm-general}
\end{equation}
where $\Delta\defeq1-L$ denotes the difference operator, and for
clarity we note that $\Delta y_{t}^{+}=y_{t}^{+}-y_{t-1}^{+}$ (rather
than being the positive part of $\Delta y_{t}$). In the case of a
canonical CKSVAR, \eqref{vecm-general} helpfully reduces to
\begin{equation}
\begin{bmatrix}\Delta y_{t}\\
\Delta x_{t}
\end{bmatrix}=c+\begin{bmatrix}\pi^{+} & \pi^{-} & \Pi^{x}\end{bmatrix}\begin{bmatrix}y_{t-1}^{+}\\
y_{t-1}^{-}\\
x_{t-1}
\end{bmatrix}+\sum_{i=1}^{k-1}\begin{bmatrix}\gamma_{i}^{+} & \gamma_{i}^{-} & \Gamma_{i}^{x}\end{bmatrix}\begin{bmatrix}\Delta y_{t-i}^{+}\\
\Delta y_{t-i}^{-}\\
\Delta x_{t-i}
\end{bmatrix}+u_{t}.\label{eq:vecm}
\end{equation}
While our main results apply to the general CKSVAR, the reader may
find it helpful to work through the remainder of this section under
the supposition that $(y_{t},x_{t})$ are generated by a canonical
CKSVAR.

Just as in a linear (cointegrated) VAR, which corresponds to the special
case of \eqref{vecm} in which $\pi^{+}=\pi^{-}$ and $\gamma_{i}^{+}=\gamma_{i}^{-}$
for all $i\in\{1,\ldots,k-1\}$, the long-run dynamics will be governed
by the matrix of coefficients on the lagged levels. More precisely,
there are two such $p\times p$ matrices, $\Pi^{+}$ and $\Pi^{-}$,
defined by
\begin{equation}
\Pi^{\pm}\defeq[\pi^{\pm},\Pi^{x}]=-\Phi^{\pm}(1).\label{eq:PIpm}
\end{equation}
Although the canonical CKSVAR technically has $2^{k}$ distinct autoregressive
`regimes' (corresponding to the possible sign patterns of $\b y_{t-1}\defeq(y_{t-1},\ldots,y_{t-k})^{\trans}$;
see also \citealp{DMW23stat}, Sec.\ \secregimes{}), the behaviour
of the CKSVAR with unit roots depends largely on the two regimes in
which the elements of $\b y_{t-1}$ are all either positive or negative,
which we shall loosely refer to as the `positive' and `negative'
regimes, to which $\Pi^{+}$ and $\Pi^{-}$ correspond. This simplification
occurs because whenever $y_{t}\sim I^{\ast}(1)$, it spends most of
its time away from the origin, so that all elements of $\b y_{t-1}$
will have the same sign almost all of the time.

Our baseline assumptions on the CKSVAR with unit roots may now be
stated.

\assumpname{CVAR}
\begin{assumption}
\label{ass:co}~
\begin{enumerate}[label=\ass{\arabic*.}, ref=\ass{.\arabic*}, itemsep=1pt,topsep=2pt]
\item \label{enu:co:roots}$\det\Phi^{\pm}(\lambda)$ has $q^{\pm}\in\{1,\ldots,p\}$
unit roots, and all others outside the unit circle; and
\item \label{enu:co:rank}$\rank\Pi^{\pm}=r^{\pm}=p-q^{\pm}$.
\end{enumerate}
\end{assumption}
\assumpname{DET}
\begin{assumption}
\label{ass:det}$c\in\spn\Pi^{+}\intsect\spn\Pi^{-}$.
\end{assumption}
\assumpname{ERR}
\begin{assumption}
\label{ass:coerr} $\{u_{t}\}_{t\in\naturals}$ is a random sequence
in $\reals^{p}$, such that $\sup_{t\in\naturals}\smlnorm{u_{t}}_{2+\delta_{u}}<\infty$
for some $\delta_{u}>0$, and
\begin{equation}
U_{n}(\lambda)\defeq n^{-1/2}\sum_{t=1}^{\smlfloor{n\lambda}}u_{t}\wkc U(\lambda)\label{eq:Unwkc}
\end{equation}
on $D[0,1]$, where $U$ is a Brownian motion in $\reals^{p}$ with
(positive definite) variance $\Sigma$.
\end{assumption}
\begin{rem}
\label{rem:coerr}By the functional martingale central limit theorem,
under the stated moment condition on $\{u_{t}\}$, a sufficient but
not necessary condition for \eqref{Unwkc} is that $\{u_{t}\}$ be
a stationary and ergodic martingale difference sequence with $\Sigma\defeq\expect u_{t}u_{t}^{\trans}$.
However, \assref{coerr} also allows $\{u_{t}\}$ to be weakly dependent,
such as if e.g.\ $u_{t}=\sum_{i=0}^{\infty}\Theta_{i}\eta_{t-i}$,
where $\sum_{i=0}^{\infty}\smlnorm{\Theta_{i}}<\infty$ and $\{\eta_{t}\}$
is a mean zero, i.i.d.\ process with $\smlnorm{\eta_{t}}_{2+\delta_{u}}<\infty$.
\end{rem}
In a linear VAR (i.e.\ where $\Phi^{+}(\lambda)=\Phi^{-}(\lambda)$)
under \assref{co} the Granger--Johansen representation theorem (\citealp[Thm.~4.2 and Cor.~4.3]{Joh95})
implies that $z_{t}\defeq(y_{t},x_{t}^{\trans})^{\trans}\sim I^{\ast}(1)$
and that there exists a full column rank matrix $\beta\in\reals^{p\times r}$
such that $\beta^{\trans}z_{t}\sim I^{\ast}(0)$, where $r=r^{+}=r^{-}$.
Thus $\{z_{t}\}$ is cointegrated in the sense of \defref{coint},
with cointegrating space $\spn\beta$. Here \assref{det} specialises
to $c\in\spn\Phi(1)$, preventing the model from generating any common
\emph{deterministic} trends between the series; it has the same effect
more generally when $\Pi^{+}\neq\Pi^{-}$. (We shall relax this condition,
so as to allow for deterministic trends, in \subsecref{deterministic}
below.)

While linear cointegration may occur in a CKSVAR, other phenomena
are possible, depending on the ranks of $\Pi^{+}$, $\Pi^{-}$ and
$\Pi^{x}$. Within the framework of $I^{\ast}(0)$ and $I^{\ast}(1)$
processes, as delimited by \assref{co}, there are three possibilities,
each of which generate profoundly different trajectories for $\{y_{t}\}$.
These are characterised in \tblref{PI}; in each case $r$ (without
a superscript) is defined such that it is possible to regard the system
as having a cointegrating rank of at least $r$ in each `regime'.
Since, $\Pi^{+}$ and $\Pi^{-}$ differ by only their first column,
$r^{+}$ and $r^{-}$ may differ by at most one, so that the case
where $r^{+}\neq r^{-}$ may be identified with case~\casecens{}
in the table without loss of generality.\footnote{If we were to instead take $r^{+}=r^{-}+1$ in this case, the characteristic
behaviour of the series in the positive and negative regimes, with
$y_{t}^{+}\sim I^{\ast}(1)$ and $y_{t}^{-}\sim I^{\ast}(0)$, would
now be reversed, but the key properties of the model would otherwise
be unaltered.}

\begin{table}
\begin{centering}
\begin{tabular}{ccccccc}
\toprule 
\addlinespace
\textbf{Case} & \textbf{$r^{+}$} & \textbf{$r^{-}$} & \textbf{$\rank\Pi^{x}$} &  & \textbf{$\pi^{+}\in\spn\Pi^{x}$} & \textbf{$\pi^{-}\in\spn\Pi^{x}$}\tabularnewline\addlinespace
\midrule
\casecens{} & $r$ & $r+1$ & $r$ &  & yes & no\tabularnewline\addlinespace
\caseclas{} & $r$ & $r$ & $r$ &  & yes & yes\tabularnewline\addlinespace
\casestat{} & $r$ & $r$ & $r-1$ &  & no & no\tabularnewline\addlinespace
\bottomrule
\end{tabular}
\par\end{centering}
\caption{Possible configurations of $\Pi^{\pm}=[\pi^{\pm},\Pi^{x}]$ and $r^{\pm}=\protect\rank\Pi^{\pm}$}
\label{tbl:PI}
\end{table}

\subsection{Common trends and long-run trajectories}

\label{subsec:common-trends}

To develop some intuition for the properties of the model in these
three cases, in advance of the representation theory developed in
the next section, it is helpful to regard the long-run behaviour of
the processes as being characterised by a \emph{space of common trends}
$\ctspc$, defined as the set of non-stochastic (i.e.\ $u_{t}\defeq0$,
$\forall t$) steady state solutions to \eqref{vecm-general}:
\begin{align*}
\ctspc & \defeq\{(y,x)\in\reals^{p}\mid\pi^{+}[y]_{+}+\pi^{-}[y]_{-}+\Pi^{x}x=0\}\\
 & =\{(y,x)\in\mathscr{Z}^{+}\mid\pi^{+}y+\Pi^{x}x=0\}\union\{(y,x)\in\mathscr{Z}^{-}\mid\pi^{-}y+\Pi^{x}x=0\}\\
 & \eqdef\ctspc^{+}\union\ctspc^{-}.
\end{align*}
This space defines the domain of the limiting processes
\begin{equation}
Z_{n}(\lambda)\defeq\begin{bmatrix}Y_{n}(\lambda)\\
X_{n}(\lambda)
\end{bmatrix}\defeq\frac{1}{n^{1/2}}\begin{bmatrix}y_{\smlfloor{n\lambda}}\\
x_{\smlfloor{n\lambda}}
\end{bmatrix}\wkc\begin{bmatrix}Y(\lambda)\\
X(\lambda)
\end{bmatrix}\eqdef Z(\lambda).\label{eq:YnXn}
\end{equation}

In a linear VAR, $\ctspc=\ker\Pi$ is a $q$-dimensional linear subspace
of $\reals^{p}$, the orthogonal complement of which is the cointegrating
space; and $Z$ is a $p$-dimensional Brownian motion with rank $q$
covariance matrix -- i.e.\ it is a rank $q$ linear function of
$U$ in \eqref{Unwkc} above -- taking values in $\ctspc$. Whereas
in the CKSVAR, $\ctspc$ no longer need be a linear subspace, but
is instead a linear cone\footnote{Recall that a set $S\subset\reals^{p}$ is termed a \emph{linear cone}
if $\lambda s\in S$ for every $s\in S$ and $\lambda\geq0$, i.e.\ if
$S$ is closed under multiplication by \emph{non-negative} scalars.} formed from the union of $\ctspc^{+}$ and $\ctspc^{-}$, the vectors
orthogonal to which are the cointegrating vectors $\beta^{+}$ and
$\beta^{-}$ on the half spaces $\mathscr{Z}^{+}$ and $\mathscr{Z}^{-}$
respectively (recall \defref{coint} above). While $Z$ remains a
function of $U$, that function need not be linear: indeed, in addition
to (linear) Brownian motions (BMs), any of the following \emph{nonlinear}
processes may also appear among the limiting stochastic trends generated
by the CKSVAR. (For further discussion of regulated BMs, see \citealp{Har85book},
Ch.\ 1.)
\begin{defn}
\label{def:nonlinearBMs}Let $W$ be a linear BM initialised from
some $W(0)\in\reals^{p}$. Suppose $p=1$; the scalar process $V$
is said to be a 
\begin{enumerate}[itemsep=2pt,topsep=3pt]
\item \emph{censored} \emph{BM }(\emph{from below}), if $V(\lambda)=\max\{W(\lambda),0\}$
\item \emph{regulated} \emph{BM} (\emph{from below}), if $V(\lambda)=W(\lambda)+\sup_{\lambda^{\prime}\leq\lambda}[-W(\lambda^{\prime})]_{+}$.
\end{enumerate}
If $V$ is as in (i) or (ii), then $-V$ is respectively censored
or regulated \emph{from above}. Suppose now that $p\geq1$, and let
$G:\reals\setmap\reals^{p\times p}$ be a map that depends only on
the sign of its argument, and is such that: (a) $h\defeq e_{1}^{\trans}G(+1)=\mu e_{1}^{\trans}G(-1)$
for some $\mu>0$; and (b) $w\elmap G(h^{\trans}w)w$ is continuous.
Then the $p$-dimensional process $V$ is said to be a
\begin{enumerate}[resume, resume*]
\item \emph{kinked} \emph{BM}, if $V(\lambda)=G[h^{\trans}W(\lambda)]W(\lambda)=G[V_{1}(\lambda)]W(\lambda).$
\end{enumerate}
\end{defn}
\begin{figure}
\begin{adjustwidth}{-2cm}{-2cm}
\begin{centering}
\begin{tabular}{cc}
\includegraphics[viewport=230bp 165bp 612bp 430bp,clip,scale=0.6]{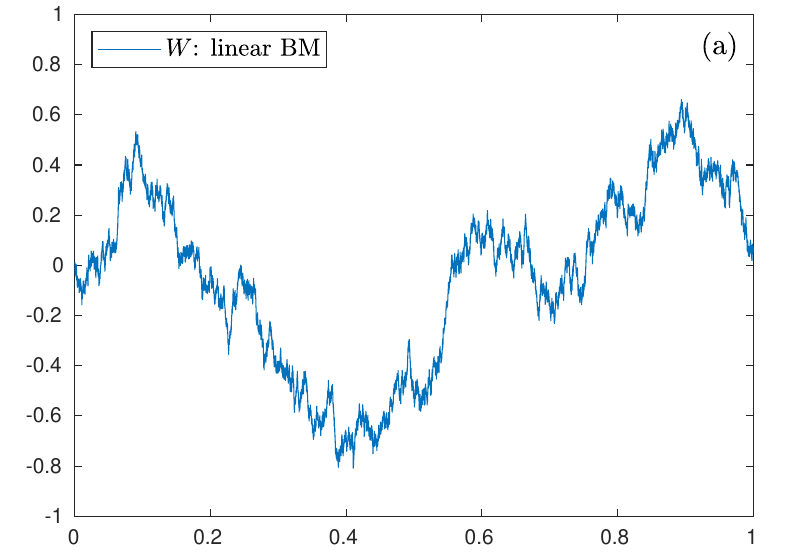} & \includegraphics[viewport=230bp 165bp 612bp 430bp,clip,scale=0.6]{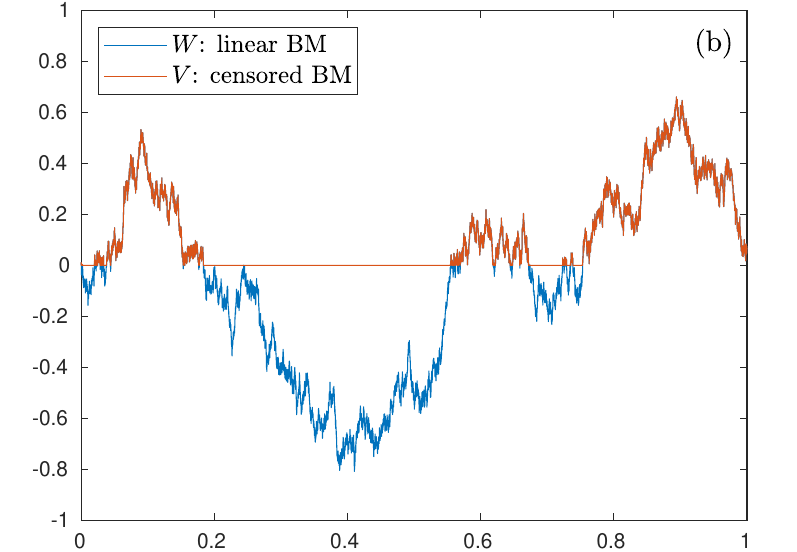}\tabularnewline
\includegraphics[viewport=230bp 165bp 612bp 430bp,clip,scale=0.6]{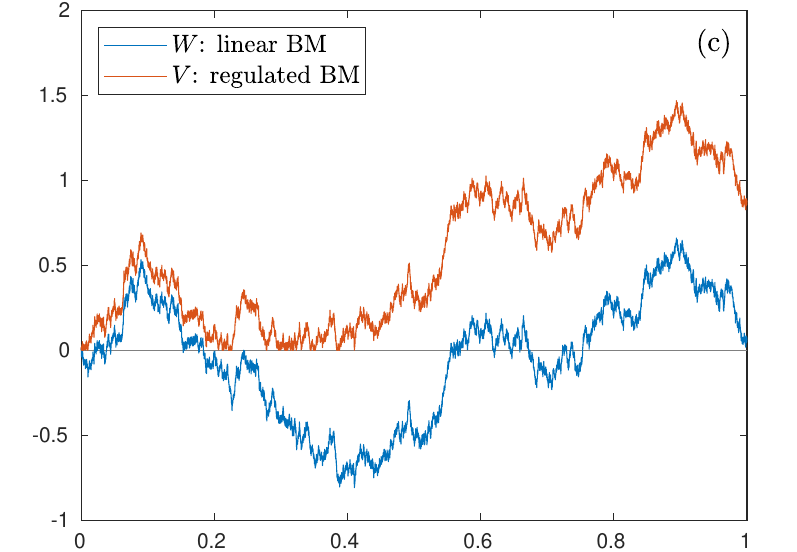} & \includegraphics[viewport=230bp 165bp 612bp 430bp,clip,scale=0.6]{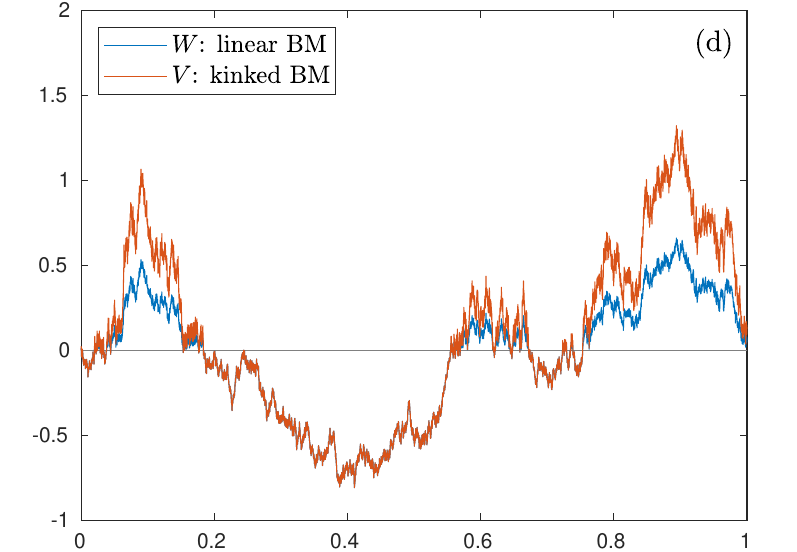}\tabularnewline
\end{tabular}
\par\end{centering}
\caption{Linear Brownian motion and derived nonlinear processes}
\label{fig:uni-bm}

\end{adjustwidth}
\end{figure}

\begin{rem}
Trajectories of these processes (denoted by $V$) in the univariate
case ($p=1$), together with the realisation of the standard Brownian
motion $W$ used to construct them, are plotted in \figref{uni-bm}.
While both censored and regulated BMs are constrained to be positive,
it is evident from panels~(b) and (c) that there are important differences
between them. For the former, the censoring does not feed back into
the underlying dynamics, and $V$ spends long stretches at zero (while
$W$ is negative); whereas for the latter, the $\sup_{\lambda^{\prime}\leq\lambda}[-W(\lambda^{\prime})]_{+}$
term continually reflects $V$ away from zero, so that $V$ spends
relatively little time near zero. The kinked BM in panel~(d) is constructed
as 
\[
V(\lambda)=\sigma_{-}W(\lambda)\indic\{W(\lambda)<0\}+\sigma_{+}W(\lambda)\indic\{W(\lambda)\geq0\}
\]
with $\sigma_{-}=1$ and $\sigma_{+}=2$; hence it tracks $W$ exactly
when $W(\lambda)<0$, but doubles the scale of $W$ when $W(\lambda)>0$.
In other respects, in the univariate case, the trajectory of a kinked
BM more closely resembles that of a linear BM than do either censored
or regulated BMs. In the multivariate case, kinked BMs are linear
combinations of the elements of $W$, with weights that depend on
the sign of $V_{1}$. Thus when plotted individually they appear similarly
to panel~(d), but if $G(\pm1)$ is rank deficient, then there will
also be certain linear combinations of the elements of $V$ that
will be zero, with those combinations depending on the sign of $V_{1}$.
\end{rem}
For $p=2$, various possible shapes of $\ctspc$ are illustrated graphically
in Figures~\ref{fig:ctspc1}--\ref{fig:ctspc3}, along with matching
example trajectories for $(y_{t},x_{t})$. In all figures, $(y_{t},x_{t})$
are generated by a canonical CKSVAR with $k=1$, $c=0$, $u_{t}\distiid N[0,I_{2}]$,
and $y_{0}=x_{0}=0$; the specification of the model is completed
by specifying $\pi^{+}$, $\pi^{-}$ and $\Pi^{x}$ in \eqref{vecm},
the values of which are given in each panel. The associated cointegrating
vectors (on the half spaces $\mathscr{Z}^{+}$ and $\mathscr{Z}^{-}$)
are the vectors orthogonal of $\ctspc^{+}$ and $\ctspc^{-}$, while
the form of the limiting processes $Z$ is suggested by the shape
of the set $\ctspc$, where it concentrates. The main qualitative
features of the three cases, as enumerated in \tblref{PI} above,
are as follows.

\subsubsection{Case~\casecens{}: regulated cointegration}

\begin{figure}
\begin{adjustwidth}{-2cm}{-2cm}
\begin{centering}
\begin{tabular}{cc}
\includegraphics[viewport=20bp 35bp 530bp 365bp,clip,scale=0.5]{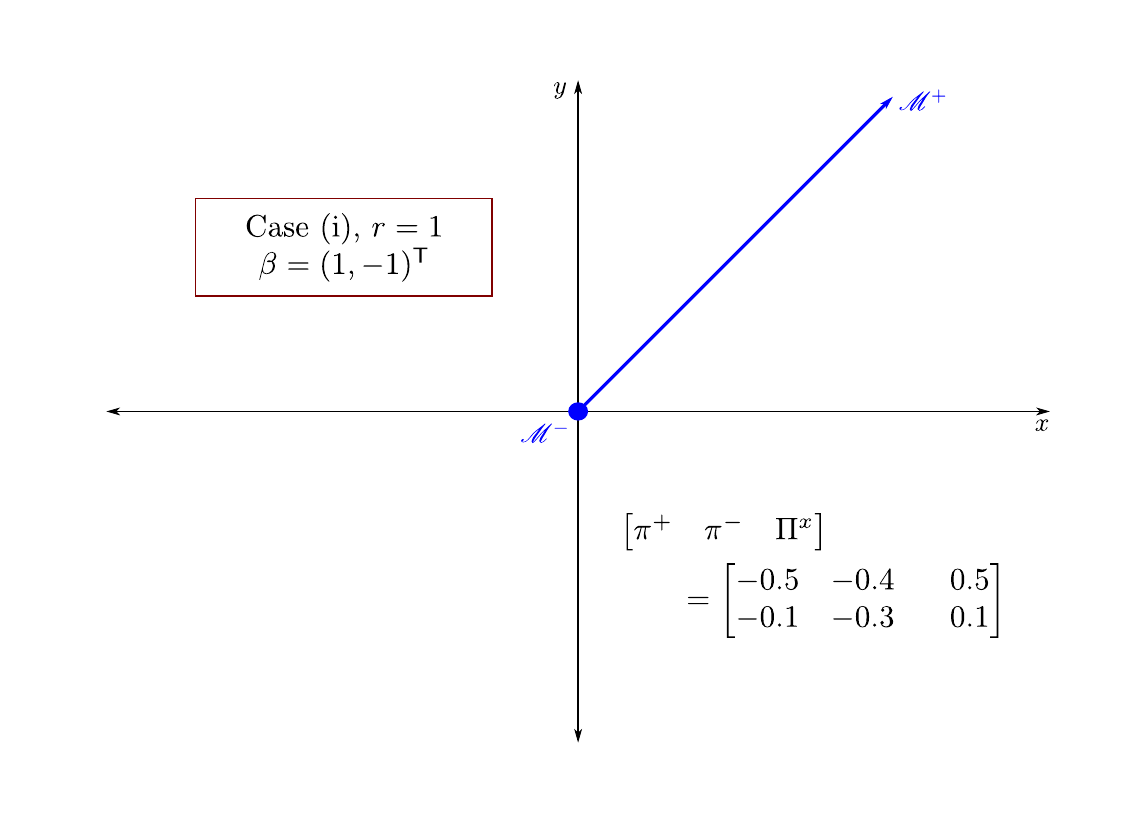} & \includegraphics[viewport=230bp 165bp 612bp 430bp,clip,scale=0.6]{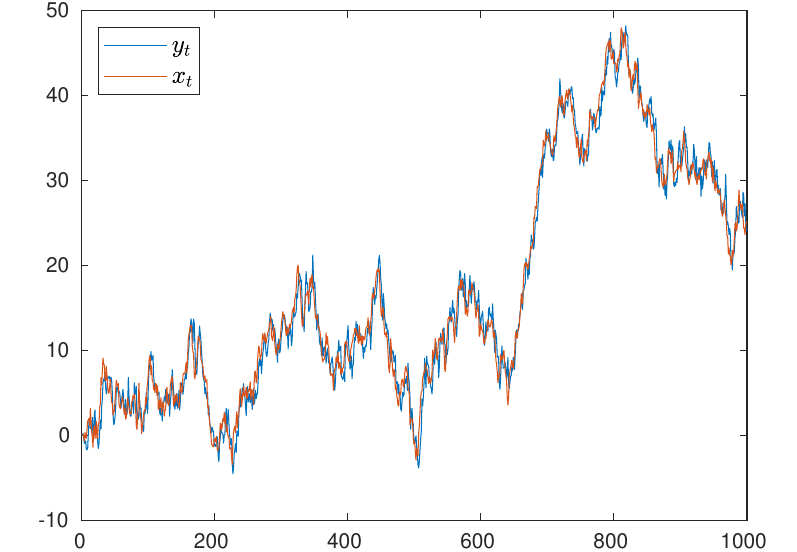}\tabularnewline
\end{tabular}
\par\end{centering}
\caption{Case~\casecens{} configuration of $\protect\ctspc$ and trajectories
of $(y_{t},x_{t})$}
\label{fig:ctspc1}

\end{adjustwidth}
\end{figure}

The distinguishing characteristic of this case is that the common
trends are restricted to the region where $y\geq0$, so that $Y$
will always be a Brownian motion regulated (from below) at zero, even
though $y_{t}$ itself may take negative values. In the case that
$q=1$, as e.g.\ when $p=2$ and $r=1$ as depicted in \figref{ctspc1},
$X$ will also be a regulated process, i.e.\ we have cointegration
where both processes share a common \emph{regulated} stochastic trend.
(When $q\geq2$, $X$ will also depend on $q-1$ additional linear
BMs.) Accordingly, a model configured as in case~\casecens{} would
be most appropriate when $y_{t}$ appears to wander randomly above
a threshold, and makes only brief sojourns below that threshold.

\saveexamplex{}

\exname{\ref*{exa:infldrift}}
\begin{example}[trending inflation target; ctd]
 Recall that this model (with $r_{t}^{\ast}=r^{\ast}=0$) is described
by:
\begin{align}
\pi_{t}-\abv{\pi}_{t} & =\theta(i_{t}^{+}+\mu i_{t}^{-}-\abv{\pi}_{t})+\varepsilon_{t}\tag{\ref{eq:is-pc}b}\label{eq:is-pc21b}\\
i_{t} & =\abv{\pi}_{t}+\gamma(\pi_{t}-\abv{\pi}_{t})\tag{\ref{eq:taylor}b}\label{eq:taylor21b}\\
\abv{\pi}_{t} & =\abv{\pi}_{t-1}+\delta(\pi_{t-1}-\abv{\pi}_{t-1})+\eta_{t},\tag{\ref{eq:infltarget}}\nonumber 
\end{align}
with the associated CKSVAR for $(i_{t},\pi_{t})$ being as in \eqref{cksvar-infldrift}
above. In this first-order model, $\Pi^{\pm}=-\Phi^{\pm}(1)=\Phi_{1}^{\pm}-\Phi_{0}^{\pm}$,
and thus it follows from \eqref{cksvar-infldrift} that
\begin{align}
\Pi^{+} & =\begin{bmatrix}\delta & -\delta\\
-\varphi_{1} & \varphi_{1}
\end{bmatrix}=\begin{bmatrix}\delta\\
-\varphi_{1}
\end{bmatrix}\begin{bmatrix}1 & -1\end{bmatrix} & \Pi^{-} & =\begin{bmatrix}\delta & -\delta\\
-\varphi_{\mu} & \varphi_{1}
\end{bmatrix}.\label{eq:PIpmexample}
\end{align}
Suppose $\delta<0$, which in the context of \ref{eq:infltarget}
implies that past deviations from target indeed feed back into the
central bank's inflation target. Then unless $\mu=1$ (in which case
the model is linear), $\varphi_{\mu}\neq\varphi_{1}$ and $\Pi^{-}$
has full rank. It follows that $\ctspc^{-}=\{0\}$, whereas since
$\rank\Pi^{+}=1$, $\ctspc^{+}$ is the `half' subspace orthogonal
to $\beta\defeq(1,-1)^{\trans}$, as depicted in \figref{ctspc1}.
It follows (via \thmref{cocens} below) that $i_{t}$ and $\pi_{t}$
are cointegrated, with cointegrating vector $\beta$, when $i_{t}$
is positive; but $I^{\ast}(0)$ when $i_{t}$ is negative; their common
limiting stochastic trend is a regulated BM.

For the economics underlying this, note that \eqref{is-pc21b} and
\eqref{taylor21b} imply that when the solution to the model has $i_{t}>0$,
i.e.\ when the ZLB is not binding, monetary policy is able to fully
achieve its objectives, in the sense that inflation is stabilised
to within an i.i.d.\ error of its target, as 
\[
\pi_{t}-\abv{\pi}_{t}=(1-\gamma\theta)^{-1}\varepsilon_{t}
\]
As a consequence, \eqref{infltarget} entails that $\abv{\pi}_{t}$
has a stochastic trend, which is inherited by both $\abv{\pi}_{t}$
and $i_{t}^{+}$ -- in the latter case, because the equilibrium rate
of interest implied by \eqref{taylor21b} is equal to $\abv{\pi}_{t}$.
Both $i_{t}^{+}$ and $\pi_{t}$ thus put the same loading on the
common trend, whence $(1,-1)^{\trans}$ is the cointegrating vector
when $i_{t}>0$. On the other hand, when the solution to the model
has $i_{t}<0$, i.e.\ when the ZLB binds, the lesser effectiveness
of unconventional monetary policy ($\mu<1$) entails that policy is
too contractionary, and so $\pi_{t}$ begins to drift below $\abv{\pi}_{t}$.
However, via \eqref{infltarget} this discrepancy \emph{raises} the
inflation target, and thereby raises the nominal rate of interest
required to achieve target inflation. This feedback actually renders
$\abv{\pi}_{t}$ as $I^{\ast}(0)$, and hence also $\pi_{t}$ and
$i_{t}^{-}$. In a relatively short time, the solution to the model
entails $i_{t}>0$ again, and thus the economy tends to spend relatively
little time in the vicinity of the ZLB.

It will be observed that the qualitative behaviour described above
is wholly contingent on $\mu<1$, i.e.\ on the ZLB as actually constraining
the conduct of monetary policy. This manifests itself, quantitatively,
as $\rank\Pi^{+}+1=\rank\Pi^{-}=2$ when $\mu<1$, as opposed to $\rank\Pi^{+}=\rank\Pi^{-}=2$
when $\mu=1$. Thus, in this setting, a test for $\Pi^{+}$ having
reduced rank would amount to a test of the null that unconventional
monetary policy is less effective than conventional policy, against
the alternative that it is equally effective.
\end{example}
\restoreexamplex{}

\subsubsection{Case~\caseclas{}: kinked cointegration}

\begin{figure}[t]
\begin{adjustwidth}{-2cm}{-2cm}
\begin{centering}
\begin{tabular}{cc}
\includegraphics[viewport=20bp 35bp 530bp 365bp,clip,scale=0.5]{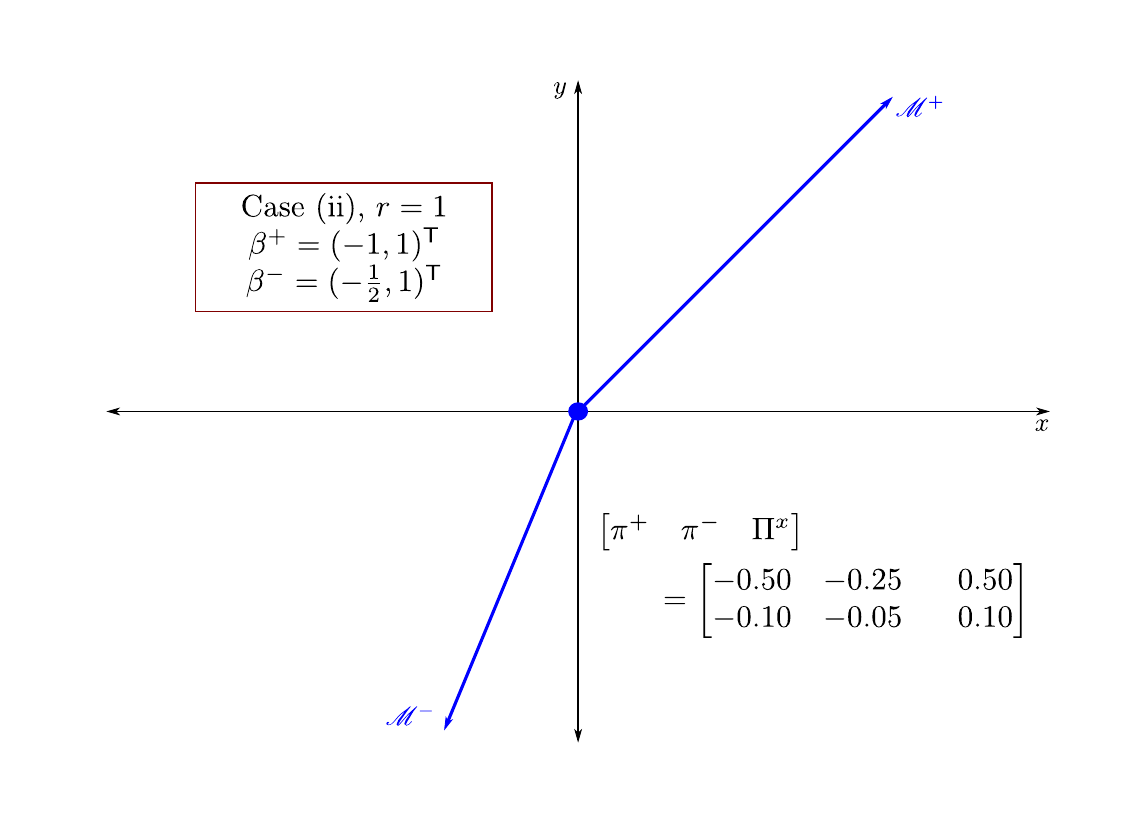} & \includegraphics[viewport=230bp 165bp 612bp 430bp,clip,scale=0.6]{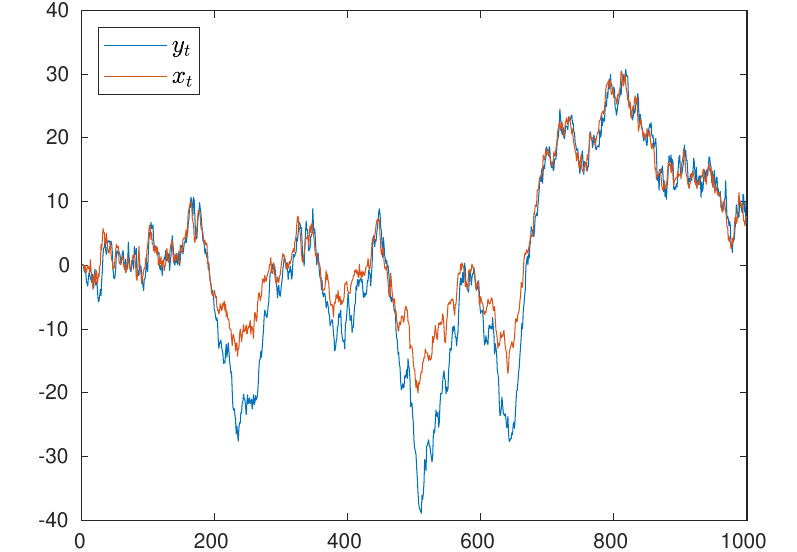}\tabularnewline
\includegraphics[viewport=20bp 35bp 530bp 365bp,clip,scale=0.5]{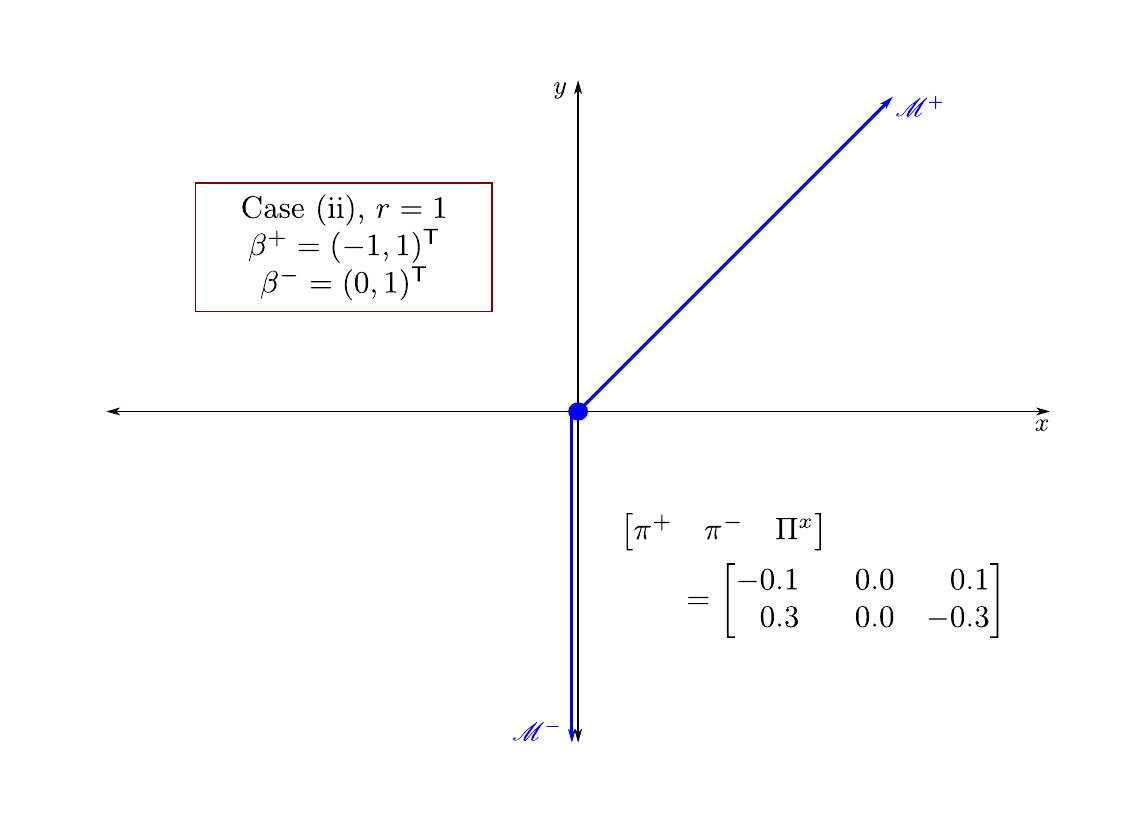} & \includegraphics[viewport=230bp 165bp 612bp 430bp,clip,scale=0.6]{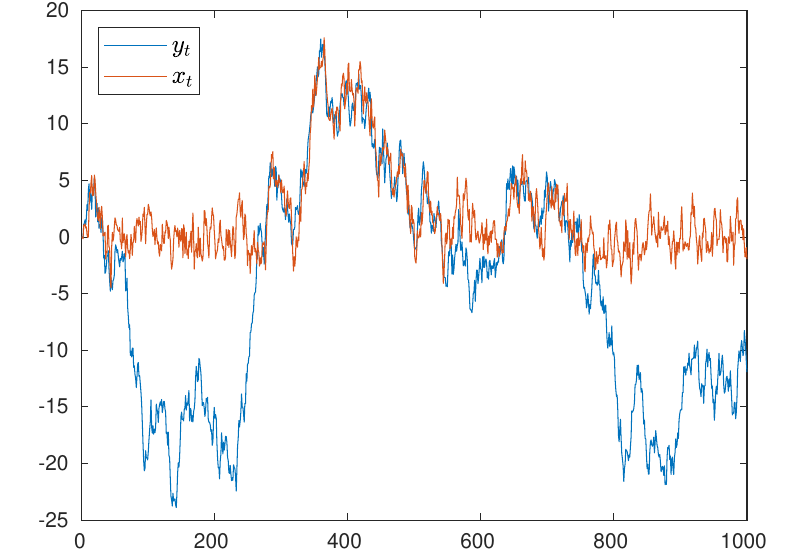}\tabularnewline
\end{tabular}
\par\end{centering}
\caption{Case~\caseclas{} configurations of $\protect\ctspc$ and trajectories
of $(y_{t},x_{t})$}
\label{fig:ctspc2}

\end{adjustwidth}
\end{figure}

This case is perhaps more reminiscent of linear cointegration, which
it accommodates as a special case. Here $\ctspc^{+}$ and $\ctspc^{-}$
trace out `half' subspaces of the same dimension $q$, but they
need not be parallel, giving rise to a kink in $\ctspc$ at the origin,
with $\ctspc$ itself being a linear cone. In general, $(Y,X)$ will
follow a kinked Brownian motion driven by $p$ linear BMs, whose loadings,
and the associated $r$ cointegrating relations, vary with the sign
of $Y$. In the top panel of \figref{ctspc2}, $y_{t}$ and $x_{t}$
are cointegrated, but with distinct cointegrating vectors $\beta^{+}=(1,-1)^{\trans}$
or $\beta^{-}=(\tfrac{1}{2},-1)^{\trans}$ applying on $\mathscr{Z}^{+}$
or $\mathscr{Z}^{-}$, i.e.\ when $y_{t}$ is positive or negative.
We have a kind of `threshold cointegration', with the movement of
$y_{t}$ across zero causing the model to switch between distinct
cointegrating spaces (cf.\ \citealp{SC04ET}). This switch is reflected
in the trajectories of $y_{t}$ and $x_{t}$: when $y_{t}\geq0$,
the two series move together approximately one-for-one; whereas if
$y_{t}\leq0$, $y_{t}$ changes by two units for every one-unit change
in $x_{t}$, in the long run. Relative to case~\casecens{}, the
trajectories of $\{y_{t}\}$ will now much more closely resemble those
of a linear unit root process, and $\{y_{t}\}$ will accordingly tend
to spend long stretches in both the positive and negative regions,
with no tendency to revert to one or the other.

The bottom panel of \figref{ctspc2} presents an important special
case where the cointegrating relationships are such that $\{x_{t}\}$
behaves like an $I^{\ast}(0)$ process when $y_{t}<0$, but cointegrates
with $y_{t}$ when $y_{t}>0$; the limiting process $X$ will thus
be a censored Brownian motion.

\saveexamplex{}

\exname{\ref*{exa:natrate}}
\begin{example}[trending natural rate; ctd]
\label{exa:monetary-coint} This model (with $\abv{\pi}_{t}=\abv{\pi}=0$)
may be rendered as
\begin{align}
\pi_{t} & =\theta(i_{t}^{+}+\mu i_{t}^{-}-r_{t}^{\ast})+\varepsilon_{t}\tag{\ref{eq:is-pc}a}\label{eq:is-pc21a}\\
i_{t} & =r_{t}^{\ast}+\gamma\pi_{t}\tag{\ref{eq:taylor}a}\label{eq:taylor21a}\\
r_{t}^{\ast} & =r_{t-1}^{\ast}+\eta_{t}.\tag{\ref{eq:natural-rate}}\label{eq:natrate21a}
\end{align}
Then similarly to \eqref{PIpmexample}, we have from the corresponding
CKSVAR for $(i_{t},\pi_{t})$, as given in \eqref{cksvar-natrate}
above, that
\begin{align*}
\Pi^{+} & =\begin{bmatrix}0\\
-(1-\theta\gamma)
\end{bmatrix}\begin{bmatrix}0 & 1\end{bmatrix} & \Pi^{-} & =\begin{bmatrix}0\\
-(1-\theta\gamma)
\end{bmatrix}\begin{bmatrix}\gamma^{-1}\tau_{\mu} & 1\end{bmatrix},
\end{align*}
where $\gamma^{-1}\tau_{\mu}=\theta(1-\mu)(1-\theta\gamma)^{-1}\leq0$,
with strict inequality unless $\mu=1$. Thus $\rank\Pi^{+}=\rank\Pi^{-}=1$,
and (by \thmref{coclas} below), there is a common stochastic trend
present in both regimes -- but when $i_{t}>0$ this loads only on
$i_{t}$, and the cointegrating vector is $\beta^{+}=(0,1)^{\trans}$.
If unconventional policy is as effective as conventional policy ($\mu=1$),
this holds also when $i_{t}<0$; otherwise the trend is shared by
both $i_{t}$ and $\pi_{t}$, and their cointegrating vector in the
negative regime is $\beta^{-}=(\gamma^{-1}\tau_{\mu},1)^{\trans}$.
Qualitatively, the behaviour of the series is as plotted for $(y_{t},x_{t})$
in the second row of \figref{ctspc2} (with $y_{t}=-i_{t}$ and $x_{t}=-\pi_{t}$,
i.e.\ with the signs reversed); in the limit, $n^{-1/2}\pi_{\smlfloor{n\lambda}}$
will be a censored BM (from above).

To account for this in economic terms, recall that \eqref{is-pc21a}
and \eqref{taylor21a} imply that when $i_{t}>0$, the central bank
is able to fully stabilise inflation, in the sense that $\pi_{t}=\abv{\pi}_{t}+(1-\gamma\theta)^{-1}\varepsilon_{t}=(1-\gamma\theta)^{-1}\varepsilon_{t}$,
since $\abv{\pi}_{t}$ is constant and normalised to zero. Thus $\pi_{t}\indic\{i_{t}>0\}\sim I^{\ast}(0)$;
whereas, $i_{t}^{+}\sim I^{\ast}(1)$, since it inherits the stochastic
trend in the natural rate of interest. However, if the model solution
requires $i_{t}<0$, then the ZLB constraint inhibits the operation
of monetary policy (if $\mu<1$), and inflation begins to drift away
from its target. That drift corresponds to the stochastic trend in
$r_{t}^{\ast}$, which is thus present in both $i_{t}^{-}$ and $\pi_{t}$,
and hence these series cointegrate.

By contrast, if monetary policy is not effectively constrained by
the ZLB, then $\pi_{t}$ would remain $I^{\ast}(0)$, irrespective
of the sign of $i_{t}$. Thus the long run behaviour of $\pi_{t}$
-- whether it is $I^{\ast}(0)$, or whether it follows a stochastic
trend (and so is $I^{\ast}(1)$) when interest rates are at the ZLB
-- here provides identifying information on the relative effectiveness
of unconventional monetary policy, i.e.\ on whether the ZLB is ever
a truly binding constraint on the central bank, just as it did in
\exaref{infldrift}.
\end{example}
\restoreexamplex{}

\saveexamplex{}

\exname{\ref*{exa:infldrift}}
\begin{example}[trending inflation target; ctd]
 Suppose that $\delta=0$ in \eqref{infltarget}, so the model is
now described by \eqref{is-pc21b}, \eqref{taylor21b} and
\begin{align*}
\abv{\pi}_{t} & =\abv{\pi}_{t-1}+\eta_{t}.
\end{align*}
The inflation target thus remains time-varying, but there is no longer
any feedback from past failures to hit that target. Then \eqref{PIpmexample}
simplifies to
\begin{align*}
\Pi^{+} & =\begin{bmatrix}0\\
-\varphi_{1}
\end{bmatrix}\begin{bmatrix}1 & -1\end{bmatrix} & \Pi^{-} & =\begin{bmatrix}0\\
-\varphi_{1}
\end{bmatrix}\begin{bmatrix}\varphi_{1}^{-1}\varphi_{\mu} & -1\end{bmatrix},
\end{align*}
where $\varphi_{1}^{-1}\varphi_{\mu}=1+\frac{\theta(\gamma-1)(1-\mu)}{1-\theta\gamma}\in(0,1]$.
Thus $i_{t}$ and $\pi_{t}$ are $I^{\ast}(1)$ and cointegrated,
but with cointegrating relations $\beta^{+}=(1,-1)^{\trans}$ and
$\beta^{-}=(\varphi_{1}^{-1}\varphi_{\mu},-1)^{\trans}$ that depend
on the sign of $i_{t}$, unless $\mu=1$. Even though $i_{t}^{-}$
is unobserved, we can still distinguish between the cases $\mu=1$
and $\mu<1$ on the basis that, when $\mu<1$, the long-run variance
of $\pi_{t}$ will differ depending on whether $i_{t}>0$ or $i_{t}=0$,
as is evident in the behaviour of $x_{t}$ in the top panel of \figref{ctspc2}.\footnote{\thmref{coclas} below implies that $n^{-1/2}\pi_{\smlfloor{n\lambda}}$
converges weakly to $\sigma_{\varepsilon}(\gamma-1)[1+(\varphi_{1}^{-1}\varphi_{\mu}-1)\indic\{W(\lambda)<0\}]W(\lambda)$,
for $\sigma_{\varepsilon}^{2}\defeq\expect\varepsilon_{t}^{2}$ and
$W$ a standard BM.}
\end{example}

\restoreexamplex{}

\subsubsection{Case~\casestat{}: linear cointegration in a nonlinear VECM}

\begin{figure}
\begin{adjustwidth}{-2cm}{-2cm}
\begin{centering}
\begin{tabular}{cc}
\includegraphics[viewport=20bp 35bp 530bp 365bp,clip,scale=0.5]{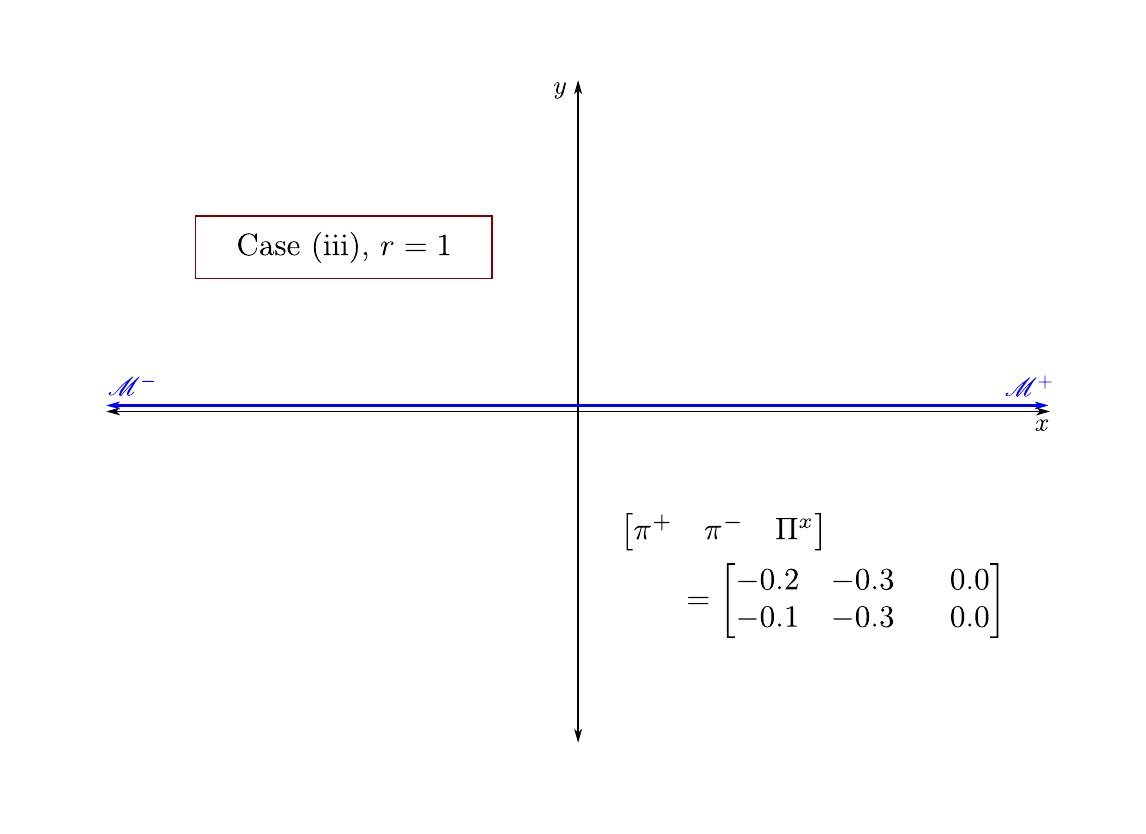} & \includegraphics[viewport=230bp 165bp 612bp 430bp,clip,scale=0.6]{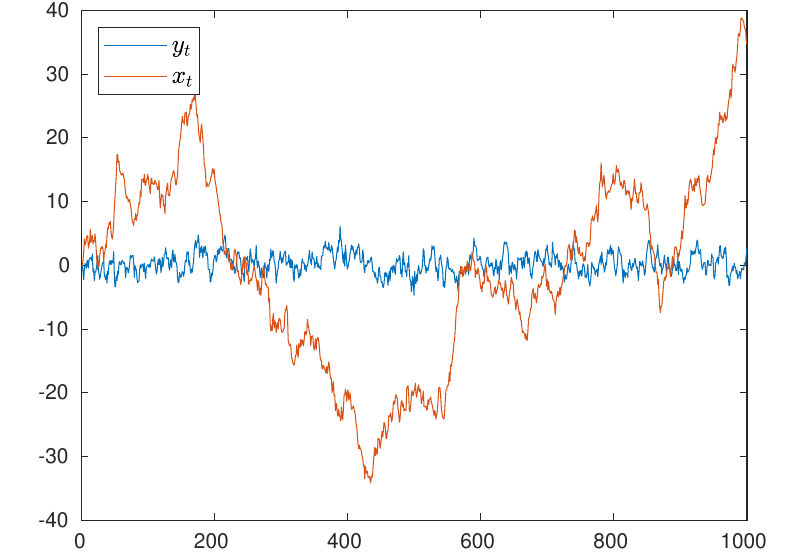}\tabularnewline
\end{tabular}
\par\end{centering}
\caption{Case~\casestat{} configuration of $\protect\ctspc$ and trajectories
of $(y_{t},x_{t})$}
\label{fig:ctspc3}

\end{adjustwidth}
\end{figure}

This case, which is depicted in the \figref{ctspc3}, entails that
no trends are present in $\{y_{t}\}$, which is in fact a stationary
process. The common trends are loaded entirely on $\{x_{t}\}$, and
the cointegrating relationships between the elements of $x_{t}$ are
unaffected by the sign of $y_{t}$, exactly as in literature on regime-switching
nonlinear VECM models (see the text following \eqref{nlvecm} above).
For this reason, a model configured as in case~\casestat{} is appropriate
only when the univariate behaviour of $\{y_{t}\}$ appears stationary,
in which case there will be frequent switches of `regime', particularly
if the mean of $y_{t}$ is near zero.

\section{Representation theory and asymptotics}

\label{sec:representation}

We now proceed to develop analogues of the Granger--Johansen representation
theorem for cases~\casecens{}--\casestat{}. These results are
of interest in their own right, because by characterising the processes
generated by the model, under alternative configurations of $\Pi^{\pm}$,
they delimit the classes of observed time series that the model might
be fruitfully applied to. Indeed, they indicate precisely the restrictions
on $\Pi^{\pm}$ that might be appropriate for specific applications,
according to whether $\{y_{t}\}$ is observed to either: wander randomly
above a threshold, but spend only brief periods below it (case~\casecens{});
wander randomly on both sides of a threshold (case~\caseclas{});
or behave like a stationary process (case~\casestat{}). Beyond such
guidance, our results also lay the groundwork for the development
of the asymptotics of likelihood-based estimators of the CKSVAR model
in the presence of unit roots, by establishing the asymptotic behaviour
of the (standardised) processes generated by the model.

To facilitate the exposition, we shall initially suppose the data
to be generated by a canonical CKSVAR, i.e.\ \assref{dgp-canon}
holds, and that there are no deterministic trends in cases~\casecens{}
and \caseclas{}, i.e.\ \assref{det} holds. Theorems~\ref{thm:cocens}--\ref{thm:costat}
below are stated under these assumptions: the minor modifications
required when \assref{dgp-canon} is replaced by \assref{dgp} are
given in the subsequent remarks, while a general treatment of deterministic
trends follows in \subsecref{deterministic}. To avoid the need to
specify how some of the quantities below should be defined when $k=1$,
we shall treat this as a special case of the model with $k=2$, in
which $\phi_{2}^{+}=\phi_{2}^{-}=0$ and $\Phi_{2}^{x}=0$; and thus
henceforth $k\geq2$, unless otherwise stated. Because of the overlap
between the arguments used to analyse the CKSVAR in each case, it
will be occasionally necessary to redefine objects that have already
appeared in the discussion of another case. While we have endeavoured
to keep such notational conflicts to a minimum, and explicitly indicated
wherever these arise, the reader is advised to treat each of the following
three subsections, and the accompanying subsections of \appref{cointegration}
where the proofs of the theorems appear, as independent of each other.

Our proofs make repeated use of companion form representations of
the VECM \eqref{vecm}, formulated slightly differently for each of
cases~\casecens{}--\casestat{}. In this respect, the nearest analogue
to our arguments, in the setting of a \emph{linear} VAR, is provided
by \citet{Han05EctJ}.
\begin{notation*}
For $A\in\reals^{m\times n}$ having full column rank, $A_{\perp}\in\reals^{m\times(m-n)}$
denotes a full column rank matrix such that $A_{\perp}^{\trans}A=0$;
we refer to $A_{\perp}$ (which is unique only up to its column span)
as `the orthocomplement of $A$'. (Note that it is \emph{not} implied
that the columns of $A_{\perp}$ should be orthogonal vectors; any
further normalisation of $A_{\perp}$ will be noted in the text if
required.) $A_{i:j}$ denotes the submatrix formed from rows $\{i,i+1,\ldots,j\}$
of $A$.
\end{notation*}

\subsection{Case~\casecens{}: regulated cointegration}

\label{subsec:censored}

Recalling \tblref{PI}, we have the familiar factorisation
\[
\Pi^{+}=\alpha^{+}\beta^{+\trans},
\]
where $\alpha^{+},\beta^{+}\in\reals^{p\times r}$ have rank $r$.
As we show below, $\beta^{+}$ spans the cointegrating space on $\mathscr{Z}^{+}=\reals_{+}\times\reals^{p-1}$
(recall \defref{coint}), so that $(y_{t}^{+},x_{t})\sim I^{\ast}(1)$
but $\beta^{+\trans}(y_{t}^{+},x_{t})\sim I^{\ast}(0)$. Moreover,
since $y_{t}^{-}\sim I^{\ast}(0)$, it follows that $\beta^{+\trans}(y_{t},x_{t})\sim I^{\ast}(0)$,
so that the columns of $\beta^{+}$ are globally cointegrating vectors.

In a linear cointegrated VAR, no assumptions additional to \assref{co}
and \assref{coerr} are needed, but the nonlinearity of the CKSVAR
prevents our assumptions on the roots of $\Phi^{\pm}(\lambda)$ from
being sufficient to ensure that the short-memory components (i.e.\ the
equilibrium errors and the first differences) are indeed $I^{\ast}(0)$.
For this reason, two further regularity conditions are required. To
state these, let
\begin{equation}
\b F_{\delta}\defeq\begin{bmatrix}I_{p(k-1)+r}+\b{\beta}^{+\trans}\b{\alpha}^{+} & \b{\beta}_{1:p}^{+\trans}(\phi_{1}^{-}-\phi_{1}^{+})\delta & \b{\beta}_{1:p}^{+\trans}(\varphi^{-}-\varphi^{+})\\
e_{1}^{\trans}\b{\alpha}^{+} & [1+e_{1}^{\trans}(\phi_{1}^{-}-\phi_{1}^{+})]\delta & e_{1}^{\trans}(\varphi^{-}-\varphi^{+})\\
0 & \delta & 0\\
0 & 0 & D
\end{bmatrix}\label{eq:Fkgeneral}
\end{equation}
where $\varphi^{\pm}\defeq[\phi_{2}^{\pm},\ldots,\phi_{k}^{\pm}]$,
$D\defeq[I_{k-2},0_{(k-2)\times1}]$; $\b{\alpha}^{+},\b{\beta}^{+}\in\reals^{kp\times[p(k-1)+r]}$
with
\begin{align}
\b{\alpha}^{+} & \defeq\begin{bmatrix}\alpha^{+} & \Gamma_{1}^{+} & \cdots & \Gamma_{k-1}^{+}\\
 & I_{p}\\
 &  & \ddots\\
 &  &  & I_{p}
\end{bmatrix}, & \b{\beta}^{+\trans} & \defeq\begin{bmatrix}\beta^{+\trans}\\
I_{p} & -I_{p}\\
 & \ddots & \ddots\\
 &  & I_{p} & -I_{p}
\end{bmatrix};\label{eq:betaplus}
\end{align}
and $\b{\beta}_{1:p}^{+}$ denotes the first $p$ rows of $\b{\beta}^{+}$.\footnote{To avoid any ambiguity, the definition of $\b{\beta}^{+\trans}$ in
\eqref{betaplus} should be read `row-wise', with the `\raisebox{-2pt}{$\smash{\ddots}$}'
signifying that successive rows of the matrix are formed by replicating
the relevant block (i.e.\ that to the upper left of the `\raisebox{-2pt}{$\smash{\ddots}$}'),
shifting it to the right by the width of the block; all other entries
are zeros. Thus for $m\in\{1,\ldots,k-1\}$, rows $r+mp+1$ to $r+mp$
of $\b{\beta}^{+\trans}$ are given by $[0_{p\times(m-1)p},I_{p},-I_{p},0_{p\times(k-m-1)}]$.
The definitions given in e.g.\ \eqref{alpabeta} below should be
interpreted similarly.} We also define
\begin{equation}
P_{\beta_{\perp}^{+}}\defeq\beta_{\perp}^{+}[\alpha_{\perp}^{+\trans}\Gamma^{+}(1)\beta_{\perp}^{+}]^{-1}\alpha_{\perp}^{+\trans}\label{eq:Pbetaplus}
\end{equation}
for $\alpha_{\perp}^{+}$ and $\beta_{\perp}^{+}$ the orthocomplements
of $\alpha^{+}$ and $\beta^{+}$. Let $\rho(M)$ denote the spectral
radius of $M\in\reals^{m\times m}$, and for $\mathcal{A}\subset\reals^{m\times m}$
a bounded collection of matrices, let
\[
\rho_{\jsr}(\mathcal{A})\defeq\limsup_{t\goesto\infty}\sup_{B\in\mathcal{A}^{t}}\rho(B)^{1/t}
\]
denote its joint spectral radius (JSR; e.g.\ \citealp{Jungers09},
Defn.\ 1.1), for $\mathcal{A}^{t}\defeq\{\prod_{s=1}^{t}M_{s}\mid M_{s}\in\mathcal{A}\}$
the set of $t$-fold products of matrices in ${\cal A}$.\footnote{For a further discussion of the JSR, and references to the literature
on methods for numerically approximating it, see \citet[Sec.~\secstabilityabstract{}]{DMW23stat}.} Let $z_{t}\defeq(y_{t},x_{t}^{\trans})^{\trans}$.

\assumpname{CO\casecens{}}
\begin{assumption}
\label{ass:cocens}~
\begin{enumerate}[label=\ass{\arabic*.}, ref=\ass{.\arabic*},itemsep=1pt,topsep=2pt]
\item \label{enu:cocens:rank}$r^{+}=\rank\Pi^{x}=r$ and $r^{-}=r+1$,
for some $r\in\{0,1,\ldots,p-1\}$.
\item \label{enu:cocens:jsr}$\rho_{\jsr}(\{\b F_{0},\b F_{1}\})<1$.
\item \label{enu:cocens:regcoef} $\regcoef_{1}<0$, where $\regcoef_{1}$
denotes the first element of $\regcoef\defeq P_{\beta_{\perp}^{+}}\pi^{-}$.
\item \label{enu:cocens:init}%
\begin{enumerate}[label=\ass{\alph*.}, ref=\ass{.\alph*}, leftmargin=0.40cm]
\item $\beta^{+\trans}z_{t}$, $y_{t}^{-}$, and $\Delta z_{t}$ have uniformly
bounded $2+\delta_{u}$ moments, for $t\in\{-k+1,\ldots,0\}$.
\item $n^{-1/2}z_{0}\inprob{\cal Z}_{0}=[\begin{smallmatrix}{\cal Y}_{0}\\
{\cal X}_{0}
\end{smallmatrix}]\in\ctspc$, where ${\cal Z}_{0}$ is non-random.
\end{enumerate}
\end{enumerate}
\end{assumption}
\needspace{3\baselineskip}
\begin{thm}
\label{thm:cocens}Suppose \assref{dgp-canon}, \assref{coerr}, \assref{co}
and \assref{cocens} hold. Then:
\begin{enumerate}[itemsep=2pt,topsep=3pt]
\item \label{enu:thm:cocens:i0}$\xi_{t}^{+}\defeq\beta^{+\trans}z_{t}\sim I^{\ast}(0)$,
$\Delta z_{t}\sim I^{\ast}(0)$, and $y_{t}^{-}\sim I^{\ast}(0)$;
\end{enumerate}
and if additionally \assref{det} holds:
\begin{enumerate}[resume]
\item \label{enu:thm:cocens:wklim}{for $U_{0}(\lambda)\defeq\Gamma^{+}(1){\cal Z}_{0}+U(\lambda)$,
jointly with \eqref{Unwkc},
\[
\begin{bmatrix}Y_{n}(\lambda)\\
X_{n}(\lambda)
\end{bmatrix}\wkc\begin{bmatrix}Y(\lambda)\\
X(\lambda)
\end{bmatrix}=P_{\beta_{\perp}^{+}}U_{0}(\lambda)+\regcoef_{1}^{-1}\regcoef\sup_{\lambda^{\prime}\leq\lambda}[-e_{1}^{\trans}P_{\beta_{\perp}^{+}}U_{0}(\lambda^{\prime})]_{+}
\]
on $D[0,1]$, where in particular $Y(\lambda)=e_{1}^{\trans}P_{\beta_{\perp}^{+}}U_{0}(\lambda)+\sup_{\lambda^{\prime}\leq\lambda}[-e_{1}^{\trans}P_{\beta_{\perp}^{+}}U_{0}(\lambda^{\prime})]_{+}$.}
\end{enumerate}
\end{thm}
\begin{rem}
\label{rem:cocens}\subremark{}\label{enu:rem:cocens:canon} If \assref{dgp-canon}
is replaced by \assref{dgp}, then the theorem continues to hold as
stated, except that \assref{cocens}\ref{enu:cocens:jsr} should be
replaced by
\begin{enumerate}[leftmargin=2cm, label=\ass{CO(i).\arabic*$^\prime$}, start=2]
\item \label{enu:cocens:jsr-struct}$\rho_{\jsr}(\{\tilde{\b F}_{0},\tilde{\b F}_{1}\})<1$;
\end{enumerate}
where the tildes refer to the parameters of the canonical CKSVAR derived
from the structural form via \propref{canonical}. (This obtains because
the derived canonical variables then themselves follow a canonical
CKSVAR that satisfies the conditions of the theorem; see \suppref{cointstruct}.)

\subremark{} The contrast with a linear cointegrated VAR is marked.
$Y$ is now a regulated Brownian motion, which also enters into other
components of $X$. Indeed, as noted in \subsecref{common-trends}
above, some components of $X$ may themselves be regulated BMs.

\subremark{}\label{enu:rem:cocens:jsr} Part~\ref{enu:thm:cocens:i0}
of the theorem is proved by obtaining a nonlinear VAR representation
for $\xi_{t}^{+}$, $\Delta z_{t}$ and $y_{t}^{-}$, whose companion
form can be expressed in terms of the matrices $\{\b F_{\delta}\mid\delta\in[0,1]\}$
(\lemref{statcomp}). Since the parameters of that VAR depend on $y_{t}^{+}$,
these processes cannot be stationary, but \assref{cocens}\enuref{cocens:jsr}
ensures that the system is sufficiently `constrained' that they
will be $I^{\ast}(0)$. A necessary but not sufficient condition for
\assref{cocens}\enuref{cocens:jsr} is that $\b F_{0}$ and $\b F_{1}$
have all their eigenvalues inside the unit circle. (This is implied
by \assref{co}: see \suppref{cocens:verification-of-remark}.)

\subremark{} It will be seen from the proof that part~\enuref{thm:cocens:wklim}
holds if \assref{cocens}\enuref{cocens:jsr} is replaced by \emph{any}
condition sufficient to ensure $(\xi_{t}^{+},\Delta z_{t},y_{t}^{-})\sim I^{\ast}(0)$.
There may thus be some scope for relaxing this assumption, which takes
essentially a worst-case approach to the behaviour of the nonlinear
VAR governing the evolution of these processes. However, even in the
more tractable univariate set-up of \citet{BD22}, it is far from
obvious what this condition (their Assumption~A4) might be replaced
by.

\subremark{}\label{enu:rem:kappa} In deriving the weak limit of
$n^{-1/2}y_{\smlfloor{n\lambda}}$, a key step is to obtain a univariate
representation for $y_{t}^{+}$ as a regulated process. The proof
of \thmref{cocens} shows that it is possible to write $y_{t}^{+}-\regcoef_{1}y_{t}^{-}=w_{t}$
for a certain series $\{w_{t}\}$: the role of \assref{cocens}\enuref{cocens:regcoef}
is to ensure that this equation is solved uniquely by taking $y_{t}^{+}=[w_{t}]_{+}$
and $y_{t}^{-}=\regcoef_{1}^{-1}[w_{t}]_{-}$. It is possible that
this condition may be redundant: as shown in \suppref{cocens:verification-of-remark},
it is implied by our other assumptions if either $k=1$ or $p=1$.
\end{rem}
\saveexamplex{}

\exname{\ref*{exa:dyntobit}}
\begin{example}[univariate; ctd]
 In the univariate ($p=1$) model \eqref{univariate-case}, case~\casecens{}
with $r=0$ corresponds to a model in which $c=0$, $\sum_{i=1}^{k}\phi_{i}^{+}=1$
and $\phi^{-}(\lambda)$ has all its roots outside the unit circle,
so that 
\begin{equation}
y_{t}=\sum_{i=1}^{k}(\phi_{i}^{+}y_{t-i}^{+}+\phi_{i}^{-}y_{t-i}^{-})+u_{t}=(1+\pi^{+})y_{t-1}^{+}+\sum_{i=1}^{k-1}\gamma_{i}^{+}\Delta y_{t-i}^{+}+\sum_{i=1}^{k}\phi_{i}^{-}y_{t-i}^{-}+u_{t}\label{eq:ur-and-stat}
\end{equation}
may be loosely regarded as an autoregressive model with a unit root
regime (since $\pi^{+}=0$) and a stationary regime (though the model
technically has $2^{k}$ distinct autoregressive regimes). \thmref{cocens}
implies that $y_{t}^{-}\sim I^{\ast}(0)$, and
\begin{equation}
Y_{n}(\lambda)\wkc Y(\lambda)=K(\lambda)-\sup_{\lambda^{\prime}\leq\lambda}[-K(\lambda^{\prime})]_{+}\label{eq:univariate-regulated}
\end{equation}
on $D[0,1]$, where $K(\lambda)=\gamma^{+}(1)^{-1}U(\lambda)$, for
$\gamma^{+}(1)=1-\sum_{i=1}^{k-1}\gamma_{i}^{+}$. $Y$ is thus a
regulated Brownian motion. \eqref{univariate-regulated} extends the
results of \citet[Thm.~3.1]{LLS11Bern} and \citet{GLY13JoE} from
a first- to a higher-order autoregressive setting. It also agrees
with the limit theory developed in \citet[Thm.~3.2]{BD22}, when their
censored dynamic Tobit model (in which $\phi_{i}^{-}=0$ for all $i\in\{1,\ldots,k\}$)
is specialised to one with an exact unit root and no intercept.
\end{example}
\restoreexamplex{}

\subsection{Case~\caseclas{}: kinked cointegration}

We turn next to the case in which the cointegrating rank $r$ is the
same across both the positive and negative regimes, though the cointegrating
space itself need not be. Here we also suppose that $\rank\Pi^{x}=r$,
which as discussed in \subsecref{common-trends} entails that $y_{t}\sim I^{\ast}(1)$.
Under the foregoing, we must have $\pi^{\pm}\in\spn\Pi^{x}$, and
so
\begin{equation}
\Pi^{\pm}=\Pi^{x}\begin{bmatrix}\theta^{\pm} & I_{p-1}\end{bmatrix}=\alpha\begin{bmatrix}\beta_{y}^{\pm} & \beta_{x}^{\trans}\end{bmatrix}\eqdef\alpha\beta^{\pm\trans},\label{eq:pikink}
\end{equation}
where $\alpha\in\reals^{p\times r}$, $\beta_{x}\in\reals^{(p-1)\times r}$
and $\beta^{\pm}\in\reals^{p\times r}$ have rank $r$, and $\theta^{\pm}\in\reals^{p-1}$
is such that $\Pi^{x}\theta^{\pm}=\pi^{\pm}$. Let $\indic^{+}(y)\defeq\indic\{y\geq0\}$
and $\indic^{-}(y)\defeq\indic\{y<0\}$, and set $\beta(y)\defeq\beta^{+}\indic^{+}(y)+\beta^{-}\indic^{-}(y)$.\footnote{There is unavoidably some arbitrariness with respect to how such objects
are defined when $y=0$, but since these only play a role in the model
when multiplied by $y$, it does not matter which convention is adopted.
Throughout the paper, we use the functions $\indic^{\pm}(y)$ to ensure
that all such definitions are mutually consistent.} Then we can define the equilibrium errors as 
\[
\xi_{t}\defeq\beta(y_{t})^{\trans}z_{t}=\indic^{+}(y_{t})\beta^{+\trans}z_{t}+\indic^{-}(y_{t})\beta^{-\trans}z_{t}.
\]
Observe how \eqref{pikink} implies that the `loadings' $\alpha$
of the equilibrium errors will be the same in both regimes, even though
the cointegrating vectors that define those errors need not be. (Case~\casestat{}
entails the opposite, with fixed cointegrating vectors but loadings
that depend on the sign of $y_{t}$: see \eqref{costatfactor} below.)

The theorem below establishes that $\xi_{t}\sim I^{\ast}(0)$, and
that $\beta^{+}$ and $\beta^{-}$ span the cointegrating spaces on
$\mathscr{Z}^{+}$ and $\mathscr{Z}^{-}$ respectively. The limiting
common trends are kinked Brownian motions given by a kind of projection
of $U$ onto $\ctspc$, defined in terms of
\begin{gather}
P_{\beta_{\perp}}(y)\defeq\beta_{\perp}(y)[\alpha_{\perp}^{\trans}\Gamma(1;y)\beta_{\perp}(y)]^{-1}\alpha_{\perp}^{\trans},\label{eq:projk1}\\
\begin{aligned}\beta_{\perp}(y) & \defeq\begin{bmatrix}1 & 0\\
-\theta(y) & \beta_{x,\perp}
\end{bmatrix}, & \qquad\Gamma(1;y) & \defeq\Gamma^{+}(1)\indic^{+}(y)+\Gamma^{-}(1)\indic^{-}(y)\end{aligned}
,\label{eq:betaperpfn}
\end{gather}
where $\theta(y)\defeq\indic^{+}(y)\theta^{+}+\indic^{-}(y)\theta^{-}$.
Such objects as $P_{\beta_{\perp}}(y)$ take only two distinct values,
depending on the sign of $y$, and we routinely use the notation $P_{\beta_{\perp}}(+1)$
and $P_{\beta_{\perp}}(-1)$ to indicate these. Similarly to case~\casecens{},
beyond our assumptions on the ranks of $\Pi^{\pm}$ and $\Pi^{x}$,
two further regularity conditions are needed to ensure that the system
is well behaved. To state these, let $\b{\alpha},\b{\beta}(y)\in\reals^{[k(p+1)-1]\times[r+(k-1)(p+1)]}$
with
\begin{align}
\b{\alpha}\defeq & \begin{bmatrix}\alpha & \Gamma_{1} & \Gamma_{2} & \cdots & \Gamma_{k-1}\\
 & I_{p+1}\\
 &  & I_{p+1}\\
 &  &  & \ddots\\
 &  &  &  & I_{p+1}
\end{bmatrix}, & \b{\beta}(y)^{\trans} & \defeq\begin{bmatrix}\beta(y)^{\trans}\\
S(y) & -I_{p+1}\\
 & I_{p+1} & -I_{p+1}\\
 &  & \ddots & \ddots\\
 &  &  & I_{p+1} & -I_{p+1}
\end{bmatrix},\label{eq:alpabeta}
\end{align}
where $\Gamma_{i}\defeq[\gamma_{i}^{+},\gamma_{i}^{-},\Gamma_{i}^{x}]$
for $i\in\{1,\ldots,k-1\}$, and
\begin{equation}
S(y)\defeq\begin{bmatrix}\indic^{+}(y) & 0\\
\indic^{-}(y) & 0\\
0 & I_{p-1}
\end{bmatrix}\label{eq:S}
\end{equation}
so that $S(y_{t})z_{t}=(y_{t}^{+},y_{t}^{-},x_{t}^{\trans})^{\trans}$,
where $z_{t}=(y_{t},x_{t}^{\trans})^{\trans}$.

\assumpname{CO\caseclas{}}
\begin{assumption}
\label{ass:coclas}~
\begin{enumerate}[label=\ass{\arabic*.}, ref=\ass{.\arabic*},itemsep=1pt,topsep=2pt]
\item \label{enu:coclas:rank}$r^{+}=r^{-}=\rank\Pi^{x}=r$, for some $r\in\{0,1,\ldots,p-1\}$.
\item \label{enu:coclas:jsr}$\rho_{\jsr}(\{I+\b{\beta}(+1)^{\trans}\b{\alpha},I+\b{\beta}(-1)^{\trans}\b{\alpha}\})<1$.
\item \label{enu:coclas:det}$\sgn\det\alpha_{\perp}^{\trans}\Gamma(1;+1)\beta_{\perp}(+1)=\sgn\det\alpha_{\perp}^{\trans}\Gamma(1;-1)\beta_{\perp}(-1)\neq0$.
\item \label{enu:coclas:init}%
\begin{enumerate}[label=\ass{\alph*.}, ref=\ass{.\alph*}, leftmargin=0.40cm]
\item $\beta(y_{t})^{\trans}z_{t}$, and $\Delta z_{t}$ have uniformly
bounded $2+\delta_{u}$ moments, for $t\in\{-k+1,\ldots,0\}$.
\item $n^{-1/2}z_{0}\inprob{\cal Z}_{0}=[\begin{smallmatrix}{\cal Y}_{0}\\
{\cal X}_{0}
\end{smallmatrix}]\in\ctspc$, where ${\cal Z}_{0}$ is non-random.
\end{enumerate}
\end{enumerate}
\end{assumption}
\begin{thm}
\label{thm:coclas}Suppose \assref{dgp-canon}, \assref{coerr}, \assref{co}
and \assref{coclas} hold, and let $\vartheta^{\trans}\defeq e_{1}^{\trans}P_{\beta_{\perp}}(+1)$
and $U_{0}(\lambda)\defeq\Gamma(1;{\cal Y}_{0}){\cal Z}_{0}+U(\lambda)$.
Then:
\begin{enumerate}[itemsep=2pt,topsep=3pt]
\item \label{enu:thm:coclas:i0}$\xi_{t}\defeq\beta(y_{t})^{\trans}z_{t}\sim I^{\ast}(0)$
and $\Delta z_{t}\sim I^{\ast}(0)$;
\end{enumerate}
and if additionally \assref{det} holds:
\begin{enumerate}[resume]
\item \label{enu:thm:coclas:wklim}on $D[0,1]$, jointly with \eqref{Unwkc},
\begin{equation}
\begin{bmatrix}Y_{n}(\lambda)\\
X_{n}(\lambda)
\end{bmatrix}\wkc\begin{bmatrix}Y(\lambda)\\
X(\lambda)
\end{bmatrix}=P_{\beta_{\perp}}[Y(\lambda)]U_{0}(\lambda)=P_{\beta_{\perp}}[\vartheta^{\trans}U_{0}(\lambda)]U_{0}(\lambda),\label{eq:coclasXY}
\end{equation}
where in particular $\sgn Y(\lambda)=\sgn\vartheta^{\trans}U_{0}(\lambda)$.
\end{enumerate}
\end{thm}
\begin{rem}
\label{rem:coclas}\subremark{}\label{enu:rem:coclas:canon} Similarly
to \remref{cocens}\ref{enu:rem:cocens:canon} above, if \assref{dgp-canon}
is replaced by \assref{dgp}, then the theorem continues to hold exactly
as stated, except that \assref{coclas}\ref{enu:coclas:jsr} should
be modified to
\begin{enumerate}[leftmargin=2cm, label=\ass{CO(ii).\arabic*$^\prime$}, start=2]
\item \label{enu:coclas:jsr-struct}$\rho_{\jsr}(\{I+\tilde{\b{\beta}}(+1)^{\trans}\tilde{\b{\alpha}},I+\tilde{\b{\beta}}(-1)^{\trans}\tilde{\b{\alpha}}\})<1$,
\end{enumerate}
where the tildes refer to the parameters of the canonical CKSVAR derived
from the structural form via \propref{canonical}. (See \suppref{cointstruct}.)

\subremark{} Even when $\beta^{+}=\beta^{-}$, such that the cointegrating
space is the same in both the positive and negative regimes, $(Y,X)$
will generally be a kinked Brownian motion because of the residual
dependence of $P_{\beta_{\perp}}(y)=\beta_{\perp}[\alpha_{\perp}^{\trans}\Gamma(1;y)\beta_{\perp}]^{-1}\alpha_{\perp}^{\trans}$
on $y$ via $\Gamma(1;y)$. Indeed, in the univariate model \eqref{univariate-case}
under case~(ii) with $r=0$,
\[
\Delta y_{t}=\sum_{i=1}^{k-1}(\gamma_{i}^{+}\Delta y_{t-i}^{+}+\gamma_{i}^{-}\Delta y_{t-i}^{-})+u_{t},
\]
\thmref{coclas} entails $Y(\lambda)=\gamma[1;U_{0}(\lambda)]^{-1}U_{0}(\lambda)$
is a kinked Brownian motion, whose variance depends on the sign of
$U_{0}(\lambda)$.

\subremark{} \assref{coclas}\enuref{coclas:jsr} plays an analogous
role to \assref{cocens}\enuref{cocens:jsr} above, ensuring that
the nonlinear VAR representation (see \lemref{nlvarcase2}) obtained
for the short-memory components $(\xi_{t},\Delta z_{t})$ is sufficiently
well-behaved that these processes are $I^{\ast}(0)$. Part~\ref{enu:thm:coclas:wklim}
would continue to hold if \assref{coclas}\enuref{coclas:jsr} were
replaced by any other condition sufficient for $(\xi_{t},\Delta z_{t})\sim I^{\ast}(0)$.
A necessary condition for \assref{coclas}\enuref{coclas:jsr} is
that the eigenvalues of $I+\b{\beta}(\pm1)^{\trans}\b{\alpha}$ should
lie strictly inside the unit circle, which is implied by \assref{co}
(see \lemref{lincoint}).

\subremark{} The map $(y,u)\elmap P_{\beta_{\perp}}(y)u$ is not
in general continuous, a fact that could interfere with the convergence
in \eqref{coclasXY}. \assref{coclas}\enuref{coclas:det} helps to
ensure this map is continuous on a sufficiently large domain to permit
\eqref{coclasXY} to follow via an application of the continuous mapping
theorem (CMT), with $Y$ and $X$ having continuous paths.

\subremark{}\label{enu:rem:coclas:detsign} Given $\alpha,\beta\in\reals^{p\times r}$
with full column rank, their orthocomplements $\alpha_{\perp},\beta_{\perp}\in\reals^{p\times q}$
are unique only up to their column span. Since in a linearly cointegrated
system \assref{co} implies that $\alpha_{\perp}^{\trans}\Gamma(1)\beta_{\perp}$
has nonzero determinant (\lemref{lincoint}), we may normalise $\alpha_{\perp}$
and/or $\beta_{\perp}$ so that $\det\alpha_{\perp}^{\trans}\Gamma(1)\beta_{\perp}=1$.
It should therefore be emphasised that \ref{ass:coclas}\ref{enu:coclas:det}
applies when $\beta_{\perp}(+1)$ and $\beta_{\perp}(-1)$ are related
via \eqref{betaperpfn}, so we are not entirely free to choose $\beta_{\perp}(\pm1)$
such that the signs of these determinants can be brought into agreement.
\end{rem}
\begin{example}
\label{exa:abh}To illustrate how \thmref{coclas} may be applied
to derive the long-run behaviour of a structural model -- even one
in which the observables do not follow a (nonlinear) VAR -- here
we consider the model of \citet{ABH23mimeo}. They regard a vector
of observable series $w_{t}$ (the inflation rate, GDP per capita,
and the nominal interest rate) as fluctuating in a stationary manner
around their long-run components $\bar{w}_{t}$, as per
\[
A(L)(w_{t}-\bar{w}_{t})=\err_{t},
\]
where $A(L)$ is the lag polynomial of a stationary VAR. The long-run
behaviour of $w_{t}$ will thus be governed by that of $\bar{w}_{t}$,
which consists of trend inflation $\bar{\pi}_{t}$, potential output
$\bar{y}_{t}$, and the trend nominal interest rate $\bar{\imath}_{t}$.
The first two components are assumed, together with the trend growth
rate $g_{t}$ of potential output, to evolve according to the nonlinear
VAR\begin{subequations}\label{eq:abh1}
\begin{align}
\bar{\pi}_{t} & =\bar{\pi}_{t-1}+u_{t}^{\pi}\label{eq:abhinfl}\\
\Delta\bar{y}_{t}-\delta(\bar{\pi}_{t}-\tau) & =-\delta(\bar{\pi}_{t-1}-\tau)+g_{t-1}+u_{t}^{y}\label{eq:abhdy}\\
g_{t} & =g_{t-1}+u_{t}^{g},\label{eq:abhg}
\end{align}
\end{subequations}where $\delta(\cdot)$ is a piecewise linear function
of the form $\delta(x)=\delta^{+}[x]_{-}+\delta^{-}[x]_{+}$, which
captures their `long-run Phillips curve'; see their equations~(5)--(9).
(Since we can trivially rewrite the model in terms of $\bar{\pi}_{t}-\tau$
rather than $\bar{\pi}_{t}$, here we also suppose that $\tau=0$,
for simplicity.) The trend in the nominal interest rate is specified
to be a linear function of $\bar{\pi}_{t}$, $g_{t}$, and an additional
stochastic trend $z_{t}$, as per their equations (10) and (11):\begin{subequations}\label{eq:abh2}
\begin{align}
\bar{\imath}_{t} & =\bar{\pi}_{t}+cg_{t}+z_{t}\\
z_{t} & =z_{t-1}+u_{t}^{z}.
\end{align}
\end{subequations}Since $\bar{\imath}_{t}$ does not enter \eqref{abh1},
we can analyse the long-run behaviour of this subsystem separately,
and thence deduce that of $\bar{\imath}_{t}$.

Because $\bar{\pi}_{t}$ enters \eqref{abh1} nonlinearly, the long-run
properties of $(\bar{\pi},\Delta\bar{y}_{t},g_{t})$ cannot be determined
on the basis of the ordinary Granger--Johansen representation theorem,
but instead requires an application of the theory developed in the
present work. To that end, we render the system \eqref{abh1} in CKSVAR
form as
\begin{multline}
\begin{bmatrix}1\\
-\delta^{+}\\
0
\end{bmatrix}\bar{\pi}_{t}^{+}+\begin{bmatrix}1\\
-\delta^{-}\\
0
\end{bmatrix}\bar{\pi}_{t}^{-}+\begin{bmatrix}0 & 0\\
1 & 0\\
0 & 1
\end{bmatrix}\begin{bmatrix}\Delta\bar{y}_{t}\\
g_{t}
\end{bmatrix}\\
=\begin{bmatrix}1\\
-\delta^{+}\\
0
\end{bmatrix}\bar{\pi}_{t-1}^{+}+\begin{bmatrix}1\\
-\delta^{-}\\
0
\end{bmatrix}\bar{\pi}_{t-1}^{-}+\begin{bmatrix}0 & 0\\
0 & 1\\
0 & 1
\end{bmatrix}\begin{bmatrix}\Delta\bar{y}_{t-1}\\
g_{t-1}
\end{bmatrix}+\begin{bmatrix}u_{t}^{\pi}\\
u_{t}^{y}\\
u_{t}^{g}
\end{bmatrix}.\label{eq:ABH}
\end{multline}
Direct calculation (for further details on this and the subsequent
calculations, see \suppref{abh-calculations}) yields
\begin{align*}
\Pi^{+}=\Pi^{-} & =\begin{bmatrix}0 & 0 & 0\\
0 & -1 & 1\\
0 & 0 & 0
\end{bmatrix}=\begin{bmatrix}0\\
1\\
0
\end{bmatrix}\begin{bmatrix}0 & -1 & 1\end{bmatrix}, & \Pi^{x} & =\begin{bmatrix}0 & 0\\
-1 & 1\\
0 & 0
\end{bmatrix}
\end{align*}
so that $\rank\Pi^{+}=\rank\Pi^{-}=\rank\Pi^{x}=1$, indicating that
the system falls within the purview of case~\caseclas{} with $\alpha=(0,1,0)^{\trans}$
and $\beta^{+}=\beta^{-}=(0,-1,1)^{\trans}$. It remains to verify
the conditions of \thmref{coclas}. We have immediately that \assref{coclas}\ref{enu:coclas:rank}
holds with $r=1$. By calculating $\tilde{\Pi}^{\pm}$ for the associated
canonical form of the model, we obtain that $1+\tilde{\beta}^{+\trans}\tilde{\alpha}=1+\tilde{\beta}^{-\trans}\tilde{\alpha}=0$,
so \ref{enu:coclas:jsr-struct} is trivially satisfied. Finally, since
in this first-order model $\Gamma^{\pm}(1)=\Phi_{0}^{\pm}$, it may
be shown that $\det\alpha_{\perp}^{\trans}\Gamma^{\pm}(1)\beta_{\perp}^{\pm}=\det I_{2}=1$,
and thus \assref{coclas}\ref{enu:coclas:det} holds. Since
\[
P_{\beta_{\perp}^{\pm}}=\beta_{\perp}^{\pm}[\alpha_{\perp}^{\trans}\Gamma^{\pm}(1)\beta_{\perp}^{\pm}]^{-1}\alpha_{\perp}^{\trans}=\begin{bmatrix}1 & 0\\
0 & 1\\
0 & 1
\end{bmatrix}\begin{bmatrix}1 & 0 & 0\\
0 & 0 & 1
\end{bmatrix}
\]
is regime-invariant, it follows therefore by \thmref{coclas} that
(supposing that all processes are initialised at zero, for simplicity)
\[
n^{-1/2}\begin{bmatrix}\bar{\pi}_{\smlfloor{n\lambda}}\\
\Delta\bar{y}_{\smlfloor{n\lambda}}\\
g_{\smlfloor{n\lambda}}
\end{bmatrix}\wkc\begin{bmatrix}1 & 0\\
0 & 1\\
0 & 1
\end{bmatrix}\begin{bmatrix}U^{\pi}(\lambda)\\
U^{g}(\lambda)
\end{bmatrix},
\]
where $n^{-1/2}\sum_{t=1}^{\smlfloor{n\lambda}}(u_{t}^{\pi},u_{t}^{g})\wkc[U^{\pi}(\lambda),U^{g}(\lambda)]$.

There are thus two common stochastic trends in the subsystem \eqref{abh1},
one of which is shared between $\Delta\bar{y}_{t}$ and $g_{t}$;
there is a single (linear) cointegrating relation given by the vector
$\beta^{+}=\beta^{-}=(0,-1,1)^{\trans}$, which eliminates these trends.
It follows moreover from \eqref{abh2} that
\[
n^{-1/2}\bar{\imath}_{\smlfloor{n\lambda}}=n^{-1/2}(\bar{\pi}_{\smlfloor{n\lambda}}+cg_{\smlfloor{n\lambda}}+z_{\smlfloor{n\lambda}})\wkc U^{\pi}(\lambda)+cU^{g}(\lambda)+U^{z}(\lambda),
\]
where $n^{-1/2}\sum_{t=1}^{\smlfloor{n\lambda}}u_{t}^{z}\wkc U^{z}(\lambda)$,
and so $\bar{\imath}_{t}$ does not cointegrate with $(\bar{\pi},\Delta\bar{y}_{t},g_{t})$.
\end{example}

\subsection{Case~\casestat{}: linear cointegration in a nonlinear VECM}

\label{subsec:case:stat}

Finally, we consider the other case with a regime-invariant cointegrating
rank $r$, but in which $\rank\Pi^{x}=r-1$. Then we have the factorisation
\begin{equation}
\Pi^{\pm}=\begin{bmatrix}\pi^{\pm} & \Pi^{x}\end{bmatrix}=\begin{bmatrix}\pi^{\pm} & \alpha_{x}\end{bmatrix}\begin{bmatrix}1 & 0\\
0 & \beta_{x}^{\trans}
\end{bmatrix}\eqdef\alpha^{\pm}\beta^{\trans},\label{eq:costatfactor}
\end{equation}
where $\alpha_{x}\in\reals^{p\times(r-1)}$ and $\beta_{x}\in\reals^{(p-1)\times(r-1)}$
have rank $r-1$, and $\alpha^{\pm},\beta\in\reals^{p\times r}$ have
rank $r$. (Note that the dimensions of $\beta_{x}$ here differ from
those in case~\caseclas{}.) The model is thus one in which the positive
and negative regimes share a common cointegrating space, which contains
$e_{p,1}$. Thus we would expect $y_{t}$ to behave like a stationary
process: and indeed, the model falls within the very general framework
of \citet{Saik08ET}, whose results could be applied to establish
the stationarity and ergodicity of the equilibrium errors 
\[
\xi_{t}\defeq\beta^{\trans}z_{t}=\begin{bmatrix}y_{t}\\
\beta_{x}^{\trans}x_{t}
\end{bmatrix},
\]
and hence of $y_{t}$. Our technical contribution here is to exploit
the structure of the CKSVAR so as to permit his conditions, which
refer to the JSR of a collection of autoregressive matrices, to be
relaxed to merely requiring the stability of a certain deterministic
subsystem (for which control over the JSR is a sufficient but not
necessary condition).\footnote{It should be emphasised that we do \emph{not }claim to be relaxing
the conditions of \citet[Thm.~1]{Saik08ET} for the \emph{general}
class of regime-switching error correction models considered in that
paper. Rather, we are able to exploit the fact that our model (the
CKSVAR configured as per \eqref{costatfactor}) is a special case
of that framework to obtain weaker sufficient conditions for the ergodicity
of $\beta^{\trans}z_{t}$, in our setting.}

Our only regularity condition on the system, in this case, is a stability
condition of this kind. To present the system to which this applies,
define
\begin{align*}
\phi_{i}(y) & \defeq\phi_{i}^{+}\indic^{+}(y)+\phi_{i}^{-}\indic^{-}(y) & \Phi_{i}(y) & \defeq\begin{bmatrix}\phi_{i}(y) & \Phi_{i}^{x}\end{bmatrix}
\end{align*}
for $i\in\{1,\ldots,k\}$, $\b y_{t-1}\defeq(y_{t-1},\ldots,y_{t-k})^{\trans}$,
and recognise that the factorisation \eqref{costatfactor} applies
more generally to
\[
\Pi(\b y_{t-1})\defeq\sum_{i=1}^{k}\Phi(y_{t-i})-I_{p}=\begin{bmatrix}\sum_{i=1}^{k}\phi_{i}(y_{t-i})-e_{1} & \alpha_{x}\end{bmatrix}\beta^{\trans}\eqdef\alpha(\b y_{t-1})\beta^{\trans}.
\]
Further, define $\b{\alpha}(\b y_{t-1}),\b{\beta}\in\reals^{kp\times[r+(k-1)p]}$
as
\begin{align}
\b{\alpha}(\b y_{t-1}) & \defeq\begin{bmatrix}\alpha(\b y_{t-1}) & \Gamma_{1}(\b y_{t-1}) & \cdots & \Gamma_{k-1}(\b y_{t-1})\\
 & I_{p}\\
 &  & \ddots\\
 &  &  & I_{p}
\end{bmatrix}, & \b{\beta}^{\trans} & \defeq\begin{bmatrix}\beta^{\trans}\\
I_{p} & -I_{p}\\
 & \ddots & \ddots\\
 &  & I_{p} & -I_{p}
\end{bmatrix},\label{eq:alphbetcase3}
\end{align}
for
\begin{equation}
\Gamma_{i}(\b y_{t-1})\defeq-\sum_{j=i+1}^{k}\Phi_{j}(y_{t-j})\label{eq:gamstatinitial}
\end{equation}
all of which depend only on the \emph{signs} of the elements of $\b y_{t-1}$.
Collect the short memory components of the model as
\[
\chi_{t}\defeq\begin{bmatrix}\xi_{t}\\
\beta_{\perp}^{\trans}\Delta z_{t}
\end{bmatrix}=\begin{bmatrix}y_{t}\\
\beta_{x}^{\trans}x_{t}\\
\beta_{x,\perp}^{\trans}\Delta x_{t}
\end{bmatrix}
\]
for $\beta_{x,\perp}\in\reals^{p-1}$ having $\rank\beta_{x,\perp}=p-r$,
such that $\beta_{x,\perp}^{\trans}\beta_{x}=0$. Letting $\b{\chi}_{t}\defeq(\chi_{t}^{\trans},\ldots,\chi_{t-k+1}^{\trans})^{\trans}$,
we show in the proof of \thmref{costat} below that
\[
\chi_{t}=Bc+M[I_{p(k-1)+r}+\b{\beta}^{\trans}\b{\alpha}(\b G\b{\chi}_{t-1})]\b H\b{\chi}_{t-1}+Bu_{t}
\]
for $B\in\reals^{p\times p}$ invertible, $M\in\reals^{p\times[p(k-1)+r]}$,
$\b G\in\reals^{k\times pk}$ and $\b H\in\reals^{[p(k-1)+r]\times kp}$.
$\{\b{\chi}_{t}\}$ thus evolves according to a regime-switching VAR,
a Markov process that will be stationary and geometrically ergodic
under the conditions given below.

The first of these conditions relates to the innovations $\{u_{t}\}$:
for technical reasons, requiring that $u_{t}$ have a (conditional)
Lebesgue density that is bounded away from zero, an assumption that
is common in the literature on ergodic Markov processes, greatly facilitates
our analysis.

\needspace{3\baselineskip}

\assumpname{ERR$^{\prime}$}
\begin{assumption}
\label{ass:err}$\{u_{t}\}$ is i.i.d.\ with a Lebesgue density that
is bounded away from zero on compact subsets of $\reals^{p}$, $\expect u_{t}=0$,
and $\expect\smlnorm{u_{t}}^{m_{0}}<\infty$ for some $m_{0}\geq1$.
\end{assumption}

Our main condition on the model parameters is the following.

\assumpname{CO\casestat{}}
\begin{assumption}
\label{ass:costat}~
\begin{enumerate}[label=\ass{\arabic*.}, ref=\ass{.\arabic*}]
\item \label{enu:costat:rank}$r^{+}=r^{-}=r$ and $\rank\Pi^{x}=r-1$;
and 
\item \label{enu:costat:stable}the deterministic system:
\[
\hat{\b{\xi}}_{t}=(I_{p(k-1)+r}+\b{\beta}^{\trans}\b{\alpha}(\b G\hat{\b{\xi}}_{t-1}))\hat{\b{\xi}}_{t-1}
\]
is \emph{stable} in the sense that $\hat{\b{\xi}}_{t}\goesto0$ for
every initialisation $\hat{\b{\xi}}_{0}\in\reals^{p(k-1)+r}$.
\end{enumerate}
\end{assumption}

Finally, to state our main result for case~\casestat{}, we recall
the following (cf.\ \citealp{Lieb2005}, p.~671; \citealp{MS08JTSA},
pp.~460f.).
\begin{defn}
Let $\{w_{t}\}_{t\in\naturals_{0}}$ be a Markov chain taking values
in $\reals^{d_{w}}$, with $m$-step transition kernel $P^{m}(w,A)\defeq\Prob\{w_{t+m}\in A\mid w_{t}=w\}$,
and $\mathcal{Q}:\reals^{d_{w}}\setmap\reals_{+}$. We say that $\{w_{t}\}$
is $\mathcal{Q}$-\emph{geometrically ergodic}, with stationary distribution
$\pi$, if $\int_{\reals^{d_{w}}}\mathcal{Q}(w)\pi(\deriv w)<\infty$,
and there exist $a,b>0$ and $\gamma\in(0,1)$ such that
\[
\sup_{B\in\Borel}\smlabs{P^{m}(w,B)-\pi(B)}\leq(a+b\mathcal{Q}(w))\gamma^{m}
\]
for all $w\in\reals^{d_{w}}$, where $\Borel$ denotes the Borel sigma-field
on $\reals^{d_{w}}$.
\end{defn}

If $\{w_{t}\}$ is $\mathcal{Q}$-geometrically ergodic, it will be
stationary if given a stationary initialisation, i.e.\ if $w_{0}$
also has distribution $\pi$; moreover, it will have geometrically
decaying $\beta$-mixing coefficients. For these and further properties,
and a discussion of how this concept relates to other notions of ergodicity
used in the literature, see \citet[pp.~671--73]{Lieb2005}.
\begin{thm}
\label{thm:costat}Suppose \assref{dgp-canon}, \assref{err}, \assref{co}
and \assref{costat} hold. Then $\{\b{\chi}_{t}\}$ is $\mathcal{Q}$-geometrically
ergodic, for $\mathcal{Q}(\b{\chi})\defeq1+\smlnorm{\b{\chi}}^{m_{0}}$.
\end{thm}
\begin{rem}
\label{rem:costat}\subremark{} If \assref{dgp-canon} is replaced
by \assref{dgp}, then \assref{costat}\ref{enu:costat:stable} should
be replaced by
\begin{enumerate}[leftmargin=2cm, label=\ass{CO(iii).\arabic*$^\prime$}, start=2]
\item \emph{\label{enu:costat:stable-struct}the following deterministic
system is stable:}
\[
\hat{\b{\xi}}_{t}=(I_{p(k-1)+r}+\tilde{\b{\beta}}^{\trans}\tilde{\b{\alpha}}(\b G\hat{\b{\xi}}_{t-1}))\hat{\b{\xi}}_{t-1},
\]
\end{enumerate}
where the tildes refer to the parameters of the canonical CKSVAR derived
from the structural form via \propref{canonical}. The theorem then
delivers the ${\cal Q}$-geometric ergodicity of $\tilde{\b{\chi}}_{t}=(\tilde{\chi}_{t}^{\trans},\ldots,\tilde{\chi}_{t-k+1}^{\trans})^{\trans}$,
where
\[
\tilde{\chi}_{t}\defeq\begin{bmatrix}\tilde{\beta}^{\trans}\tilde{z}_{t}\\
\tilde{\beta}_{\perp}^{\trans}\Delta\tilde{z}_{t}
\end{bmatrix}
\]
is formed from the canonical parameters and variables. Since, as shown
in \suppref{cointstruct}, we can write $\beta^{\trans}z_{t}=(y_{t},(\beta_{x}^{\trans}x_{t})^{\trans})^{\trans}$
and $\Delta x_{t}$ as measurable (indeed, Lipschitz continuous) functions
of $\tilde{\beta}^{\trans}\tilde{z}_{t}$ and $(\tilde{\chi}_{t},\tilde{\chi}_{t-1})$
respectively, these processes will inherit the stationarity and geometric
$\beta$-mixing properties of $\b{\chi}_{t}$ that are a corollary
of $\mathcal{Q}$-geometric ergodicity.

\subremark{} Since $\{\Delta x_{t}\}$ is geometrically $\beta$-mixing,
and $x_{t}=\sum_{s=1}^{t}\Delta x_{s}+x_{0}$, the preceding result
may be used as a starting point for the derivation of the asymptotics
of $n^{-1/2}x_{\smlfloor{n\lambda}}$, which should converge to a
multivariate Brownian motion (possibly after linear detrending). However,
as discussed in \citet[pp.~307f.]{Saik08ET}, determining the rank
of the long-run variance of $\{\Delta x_{t}\}$, and hence the rank
of the limiting process, is non-trivial. While its rank is bounded
above by $q=p-r$, it need not be equal to $q$; guaranteeing the
latter is likely to require further conditions on the model parameters.
We leave this for future work.
\end{rem}

\subsection{Deterministic trends}

\label{subsec:deterministic}

To simplify the exposition of our results for cases~\casecens{}
and \caseclas{} above, we have so far maintained \assref{det}, which
prevents the model from generating any common \emph{deterministic}
trends. Since this is likely to be restrictive in applications --
there being many macroeconomic series that exhibit \emph{both} stochastic
and deterministic trends -- it is important to clarify that this
simplifying assumption may be almost entirely dispensed with.

The presence of deterministic trends complicates the analysis of cases~\casecens{}
and \caseclas{}, because these deterministic trends will generally
dominate any common stochastic trends, such that the limit theory
developed in Theorems~\ref{thm:cocens} and \ref{thm:coclas} no
longer applies directly to the standardised process $Z_{n}(\lambda)\defeq n^{-1/2}z_{\smlfloor{n\lambda}}$.
Nonetheless, as explained below, these results continue to provide
an accurate description of the asymptotics of this process \emph{upon
linear detrending}. (Since we did not maintain \assref{det} in the
context of case~\casestat{}, we have nothing to say about that case
in this section.)

We therefore now contemplate a model in which \assref{det} is relaxed,
i.e.\ in which it is no longer required that $c\in\spn\Pi^{+}\intsect\Pi^{-}$.
In a linear cointegrated VAR, an unrestricted intercept $c\in\reals^{p}$
permits the model to impart deterministic trends to all elements of
$z_{t}=(y_{t},x_{t}^{\trans})^{\trans}$, but in such a way that these
deterministic trends are eliminated by the cointegrating relations
(simultaneously with the common stochastic trends; see \citealp{Joh95},
Sec.~5.7). This remains true in the CKSVAR, with the caveat that
if a deterministic trend is imparted to $y_{t}$, then this dominant
drift component will asymptotically push $y_{t}$ so far into either
the positive or negative region (depending on the sign of the drift)
as to render the threshold nonlinearity at zero irrelevant -- and
so the series would come to be adequately described by a linear VAR.

Since we would only apply the CKSVAR when $\{y_{t}\}$ spent an appreciable
portion of the sample on both sides of zero, it is appropriate to
work within an asymptotic framework in which both regions continue
to be visited with nonvanishing probability as $T\goesto\infty$.\footnote{Here we maintain that the location of the threshold delimiting the
regimes is normalised to zero; of course the argument continues to
apply if that threshold were instead fixed at some other finite level.} The cleanest way to ensure this is to restrict $c$ such that a deterministic
trend may be imparted to $x_{t}$, but not to $y_{t}$. In the context
of cases~\casecens{} and \caseclas{}, essentially the same condition
is required, but is denoted slightly differently so as to be consistent
with the notation used in the analysis of these two cases. (Recall
the definitions of $P_{\beta_{\perp}^{+}}$ and $P_{\beta_{\perp}}(+1)$
appearing in \eqref{Pbetaplus} and \eqref{projk1} above.)

\assumpname{DET$^{\prime}$}
\begin{assumption}
\label{ass:detprime}In case~\casecens{}, $e_{1}^{\trans}P_{\beta_{\perp}^{+}}c=0$;
in case~\caseclas{}, $e_{1}^{\trans}P_{\beta_{\perp}}(+1)c=0$.
\end{assumption}

To state our results in terms of the structural form of the CKSVAR,
we shall say that `\assref{cocens}\ass{$^\prime$} holds' whenever
\assref{cocens} holds with \assref{cocens}\ref{enu:cocens:jsr}
replaced by \ref{enu:cocens:jsr-struct}, and mutatis mutandis for
\assref{coclas}\ass{$^\prime$} (see part~(i) of Remarks~\ref{rem:cocens}
and \ref{rem:coclas} above). To deal with the possible presence of
a deterministic trend in $x_{t}$, in case~\casecens{} we define
the linearly detrended processes
\[
\begin{bmatrix}y_{t}^{d}\\
x_{t}^{d}
\end{bmatrix}=z_{t}^{d}\defeq z_{t}-(P_{\beta_{\perp}^{+}}c)t
\]
noting that $y_{t}^{d}=y_{t}$ under \assref{detprime}; in case~\caseclas{},
$z_{t}^{d}$ is defined similarly with $P_{\beta_{\perp}^{+}}$ replaced
by $P_{\beta_{\perp}}(+1)$. Define $X_{n}^{d}(\lambda)\defeq n^{-1/2}x_{\smlfloor{n\lambda}}^{d}$
to be the standardised process corresponding to $x_{t}^{d}$. The
\emph{only} modification that then needs to be made to the conclusions
of Theorems~\ref{thm:cocens} and \ref{thm:coclas} above is that
$X_{n}$ should be replaced by $X_{n}^{d}$.
\begin{thm}
\label{thm:codet}Suppose that \assref{dgp}, \assref{coerr}, \assref{detprime}
and \assref{co} hold. If additionally \assref{cocens}\ass{$^\prime$}
(respectively \assref{coclas}\ass{$^\prime$}) holds, then the conclusions
of \thmref{cocens} (respectively \thmref{coclas}) hold with with
$X_{n}^{d}$ in place of $X_{n}$.
\end{thm}

\begin{rem}
\subremark{} The preceding illustrates how common stochastic and
deterministic trends may be simultaneously accomodated within a CKSVAR.
In particular, because the deterministic trends enter $z_{t}$ through
$(P_{\beta_{\perp}^{+}}c)t$ (or $[P_{\beta_{\perp}}(+1)c]t$), it
remains true that $\beta^{+\trans}z_{t}\sim I^{\ast}(0)$ in case~\casecens{}
and $\beta(y_{t})^{\trans}z_{t}\sim I^{\ast}(0)$ in case \caseclas{},
since these transformations of $z_{t}$ now eliminate both its common
stochastic \emph{and} deterministic trends. These objects thus continue
to describe the (nonlinear) cointegrating relations between the elements
of $z_{t}$.

\subremark{} It would be possible to accommodate a deterministic
trend in $y_{t}$ within our asymptotic framework, provided that this
component is sufficiently small that it does not overwhelm the stochastic
trend component in $y_{t}$, such that the vicinity of $y_{t}=0$
continues to be visited with non-negligible probability in the limit.
Mathematically, this could be engineered by relaxing \assref{detprime}
so as to permit $e_{1}^{\trans}P_{\beta_{\perp}^{+}}c$ (or $e_{1}^{\trans}P_{\beta_{\perp}}(+1)c$)
to be nonzero but of order $n^{-1/2}$: in which case the conclusions
of Theorem~\ref{thm:codet} would remain unaltered, except that the
limiting processes would now also incorporate a drift term. (See \citealp{BD22},
for an illustration of the form that this takes in the univariate
form of case~\casecens{}.) Thus even if \assref{detprime} is \emph{not}
imposed on the model, so that $c$ is unrestricted, it is still possible
to interpret $\beta^{+}$ and $\beta(y)$ in the manner suggested
above, in cases~\casecens{} and \caseclas{} respectively.
\end{rem}

\section{Conclusion}

\label{sec:conclusion}

The CKSVAR provides a flexible yet tractable framework in which to
structurally model vector time series subject to an occasionally binding
constraint, such as the zero lower bound on interest rates, and more
general threshold nonlinearities. Nonetheless, even that seemingly
limited amount of nonlinearity radically changes the properties of
the model relative to a linear VAR. When unit autoregressive roots
are introduced into the model, it is able to accommodate varieties
of long-run behaviour that cannot be generated within a linear VAR,
such as nonlinear common stochastic trends (censored, regulated and
kinked Brownian motions) and cointegrating relationships that may
be regime-dependent. This is not merely a theoretical curiosity but
rather something that, as our examples illustrate, allows the long-run
properties of the model to carry useful identifying information on
structural parameters, as might pertain e.g.\ to the relative effectiveness
of unconventional monetary policy.

Our results provide a complete characterisation of the forms of nonlinear
cointegration (between processes $\ast$-integrated of order one)
generated by the CKSVAR. In deriving these, we have given the first
treatment of how nonlinear cointegration, in the profound sense of
nonlinear common stochastic trends and nonlinear cointegrating relations,
may be systematically generated within a (nonlinear) VAR, and thus
the first extension of the Granger--Johansen representation theorem
to a nonlinearly cointegrated setting. The special structure of the
CKSVAR makes this problem peculiarly tractable, while being flexible
enough to generate interesting departures from linear cointegration.
Our results indicate how progress may now be made in the analysis
of more general nonlinear VARs with unit roots, while our representation
theory provides the foundations for inference on cointegrating relations
in the CKSVAR. Our findings with respect to these problems will be
reported elsewhere; some initial results regarding inference in this
setting, in the univariate case, are given in \citet{BD22} and \citet{DJ25}.

{\singlespacing

\bibliographystyle{ecta}
\bibliography{cksvar}

}

\appendix

\section{Auxiliary lemmas}

\label{app:prelims}

We here collect three elementary lemmas, whose proofs appear in \suppref{auxiliaryproofs}
for completeness. The third records, for convenience, some properties
of the linear cointegrated VAR which are closely related to Lemmas~A.1
and A.2 in \citet{Han05EctJ}.
\begin{lem}
\label{lem:stochbound}Let $\{v_{t}\}$, $\{A_{t}\}$, $\{B_{t}\}$
and $\{c_{t}\}$ be random sequences, respectively taking values in
$\reals^{d_{v}}$, $\reals^{d_{w}\times d_{w}}$, $\reals^{d_{w}\times d_{v}}$
and $\reals^{d_{w}}$, where $t\in\naturals$. Suppose $\{w_{t}\}$
satisfies 
\[
w_{t}=c_{t}+A_{t}w_{t-1}+B_{t}v_{t}
\]
for some given (random) $w_{0}$, and:
\begin{enumerate}[itemsep=2pt,topsep=3pt]
\item $A_{t}\in\mset A$, $B_{t}\in\mset B$ and $c_{t}\in{\cal C}$ for
all $t\in\naturals$, where $\mset A$, $\mset B$ and ${\cal C}$
are bounded subsets of $\reals^{d_{w}\times d_{w}}$, $\reals^{d_{w}\times d_{v}}$
and $\reals^{d_{w}}$ respectively, and $\rho_{\jsr}(\mset A)<1$;
\item $m_{0}\geq1$ is such that $\smlnorm{w_{0}}_{m_{0}}+\sup_{t\in\naturals}\smlnorm{v_{t}}_{m_{0}}<\infty$.
\end{enumerate}
Then for each $\gamma\in(\rho_{\jsr}(\mset A),1)$, there exists a
$C<\infty$ such that for all $t\in\naturals$,
\[
\smlnorm{w_{t}}\leq C\left[\sum_{s=0}^{t-1}\gamma^{s}(1+\smlnorm{v_{t-s}})+\gamma^{t}\smlnorm{w_{0}}\right],
\]
and for all $n\in\naturals$ and $m\in[1,m_{0}]$,
\[
\max_{1\leq t\leq n}\smlnorm{w_{t}}_{m}\leq C\left(1+\smlnorm{w_{0}}_{m}+\max_{1\leq t\leq n}\smlnorm{v_{t}}_{m}\right).
\]
\end{lem}
\needspace{4\baselineskip}
\begin{lem}
\label{lem:rk1perturb}Suppose $c\in\reals^{m}$, $d,v\in\reals^{n}$,
and $A,B_{1},B_{2}\in\reals^{m\times n}$ have full column rank and
are such that $B_{1}-B_{2}=cd^{\trans}$ and $\det(A^{\trans}B_{1})\cdot\det(A^{\trans}B_{2})>0$.
Then
\begin{enumerate}
\item there exists a $\mu>0$ such that $d^{\trans}(A^{\trans}B_{1})^{-1}=\mu d^{\trans}(A^{\trans}B_{2})^{-1}$;
and
\item if $d^{\trans}(A^{\trans}B_{1})^{-1}v=0$, then $(A^{\trans}B_{1})^{-1}v=(A^{\trans}B_{2})^{-1}v$.
\end{enumerate}
\end{lem}
\begin{lem}
\label{lem:lincoint}Suppose $\Phi(\lambda)\defeq I_{p}-\sum_{i=1}^{k}\Phi_{i}\lambda^{i}$
and $\Pi\defeq-\Phi(1)$ satisfy \assref{co} (i.e.\ without the
`$\pm$' superscripts). Define
\[
\b{\Pi}\defeq\begin{bmatrix}\Pi+\Gamma_{1} & -\Gamma_{1}+\Gamma_{2} & \cdots & -\Gamma_{k-1}\\
I_{p} & -I_{p}\\
 & \ddots & \ddots\\
 &  & I_{p} & -I_{p}
\end{bmatrix},
\]
and let $\Gamma(\lambda)\defeq I_{p}-\sum_{i=1}^{k-1}\Gamma_{i}\lambda^{i}$
be such that $\Phi(\lambda)=\Phi(1)\lambda+\Gamma(\lambda)(1-\lambda)$.\footnote{As per the footnote to \eqref{betaplus}, the definition of $\b{\Pi}$
here should be read `row-wise', so that for $m\in\{2,\ldots,k\}$,
rows $(m-1)p+1$ to $mp$ of $\b{\Pi}$ are given by $[0_{p\times(m-2)p},I_{p},-I_{p},0_{p\times(k-m)p}]$.
The matrix definitions appearing subsequently in these appendices
should be interpreted similarly.} Then:
\begin{enumerate}[itemsep=2pt,topsep=3pt]
\item there exist $\alpha,\beta\in\reals^{p\times r}$ with rank $r$,
such that $\Pi=\alpha\beta^{\trans}$ and $\b{\Pi}=\b{\alpha}\b{\beta}^{\trans}$,
where
\begin{align*}
\b{\alpha} & \defeq\begin{bmatrix}\alpha & \Gamma_{1} & \cdots & \Gamma_{k-1}\\
 & I_{p}\\
 &  & \ddots\\
 &  &  & I_{p}
\end{bmatrix} & \b{\beta}^{\trans} & \defeq\begin{bmatrix}\beta^{\trans}\\
I_{p} & -I_{p}\\
 & \ddots & \ddots\\
 &  & I_{p} & -I_{p}
\end{bmatrix};
\end{align*}
\item $I_{p(k-1)+r}+\b{\beta}^{\trans}\b{\alpha}$ has its eigenvalues strictly
inside the unit circle, and $\b{\alpha}_{\perp}^{\trans}\b{\beta}_{\perp}=\alpha_{\perp}^{\trans}\Gamma(1)\beta_{\perp}$
and $\b{\beta}^{\trans}\b{\alpha}$ are full rank, where
\begin{align}
\b{\alpha}_{\perp}^{\trans} & \defeq\alpha_{\perp}^{\trans}\begin{bmatrix}I_{p} & -\Gamma_{1} & \cdots & -\Gamma_{k-1}\end{bmatrix} & \b{\beta}_{\perp}^{\trans} & \defeq\beta_{\perp}^{\trans}\begin{bmatrix}I_{p} & I_{p} & \cdots & I_{p}\end{bmatrix};\label{eq:abperp}
\end{align}
\item $\b P_{\b{\beta}_{\perp}}\defeq\b{\beta}_{\perp}[\b{\alpha}_{\perp}^{\trans}\b{\beta}_{\perp}]^{-1}\b{\alpha}_{\perp}^{\trans}$
and $\b P_{\b{\alpha}}\defeq\b{\alpha}[\b{\beta}^{\trans}\b{\alpha}]^{-1}\b{\beta}^{\trans}$
are complementary projections, and the upper left $p\times p$ block
of $\b P_{\b{\beta}_{\perp}}$ is $P_{\beta_{\perp}}\defeq\beta_{\perp}[\alpha_{\perp}^{\trans}\Gamma(1)\beta_{\perp}]^{-1}\alpha_{\perp}^{\trans}$.
\end{enumerate}
\end{lem}

\section{Proofs of theorems}

\label{app:cointegration}

\subsection{Proof of \thmref{cocens}}

Defining
\begin{align}
v_{t} & \defeq u_{t}+(\pi^{-}-\pi^{+})y_{t-1}^{-}+\sum_{i=1}^{k-1}(\gamma_{i}^{-}-\gamma_{i}^{+})\Delta y_{t-i}^{-}=u_{t}+\sum_{i=1}^{k}(\phi_{i}^{-}-\phi_{i}^{+})y_{t-i}^{-}\label{eq:vdef}
\end{align}
and recalling $z_{t}=(y_{t},x_{t}^{\trans})^{\trans}$, we may rewrite
the model \eqref{vecm} as
\begin{align*}
\Delta z_{t} & =c+\Pi^{+}z_{t-1}+\sum_{i=1}^{k-1}\Gamma_{i}^{+}\Delta z_{t-i}+v_{t}
\end{align*}
where $\Pi^{+}=[\pi^{+},\Pi^{x}]$ and $\Gamma_{i}^{+}=[\gamma_{i}^{+},\Gamma^{x}]$.
Conformably defining
\begin{align*}
\b{\Pi}^{+} & \defeq\begin{bmatrix}\Pi^{+}+\Gamma_{1}^{+} & -\Gamma_{1}^{+}+\Gamma_{2}^{+} & \cdots & -\Gamma_{k-1}^{+}\\
I_{p} & -I_{p}\\
 & \ddots & \ddots\\
 &  & I_{p} & -I_{p}
\end{bmatrix} & \b c & \defeq\begin{bmatrix}c\\
0_{p}\\
\vdots\\
0_{p}
\end{bmatrix} & \b v_{t} & \defeq\begin{bmatrix}v_{t}\\
0_{p}\\
\vdots\\
0_{p}
\end{bmatrix} & \b z_{t} & \defeq\begin{bmatrix}z_{t}\\
z_{t-1}\\
\vdots\\
z_{t-k+1}
\end{bmatrix}
\end{align*}
so that $\b{\Pi}^{+}=\b{\alpha}^{+}\b{\beta}^{+\trans}$, where $\b{\alpha}^{+}$
and $\b{\beta}^{+}$ are as in \eqref{betaplus}, we render the system
in companion form as
\begin{equation}
\Delta\b z_{t}=\b c+\b{\Pi}^{+}\b z_{t-1}+\b v_{t}.\label{eq:poscompanion}
\end{equation}

The next lemma provides a (nonlinear) VAR representation for the short-memory
components, which comprise: (a) the equilibrium errors (using the
cointegrating vectors $\beta^{+}$ from the positive regime) and (lagged)
differences,
\[
\b{\xi}_{t}^{+}\defeq\b{\beta}^{+\trans}\b z_{t}=(\xi_{t}^{+\trans},\Delta z_{t}^{\trans},\ldots,\Delta z_{t-k+2}^{\trans})^{\trans}
\]
where $\xi_{t}^{+}\defeq\beta^{+\trans}z_{t}$; (b) the lagged levels
$\b y_{t-1}^{-}\defeq(y_{t-1}^{-},\ldots,y_{t-k+1}^{-})^{\trans}$;
and (c) an auxiliary series $\pseudy_{t}$. Collect these in $\b{\zeta}_{t}\in\reals^{p(k-1)+k+r}$,
and conformably define
\begin{align*}
\b{\zeta}_{t} & =\begin{bmatrix}\b{\xi}_{t}^{+}\\
\pseudy_{t}\\
\b y_{t-1}^{-}
\end{bmatrix} & \b{\varepsilon}_{t} & \defeq\begin{bmatrix}\b{\beta}_{1:p}^{+\trans}u_{t}\\
e_{1}^{\trans}u_{t}\\
0_{k-1}
\end{bmatrix} & \b{\mu} & \defeq\begin{bmatrix}\b{\beta}_{1:p}^{+\trans}c\\
e_{1}^{\trans}c\\
0_{k-1}
\end{bmatrix}
\end{align*}
where $\b{\beta}_{1:p}^{+}$ denotes the first $p$ rows of $\b{\beta}^{+}$.
Recall the definition of $\b F_{\delta}$ given in \eqref{Fkgeneral}
above.
\begin{lem}
\label{lem:statcomp}Suppose \assref{co} and \assref{cocens}\enuref{cocens:rank}
hold. Set $\delta_{0}\defeq1$ and $\abv y_{0}\defeq y_{0}^{-}$.
Then $\{\b{\zeta}_{t}\}$ follows
\begin{align}
\b{\zeta}_{t} & =\b{\mu}+\b F_{\delta_{t-1}}\b{\zeta}_{t-1}+\b{\varepsilon}_{t}\label{eq:srcomp}\\
\delta_{t} & =\begin{cases}
(y_{t-1}^{+}+\pseudy_{t})/\pseudy_{t} & \text{if }y_{t-1}^{+}+\pseudy_{t}<0,\\
0 & \text{otherwise};
\end{cases}\label{eq:delcens}
\end{align}
which implies, in particular, that $y_{t}^{-}=\delta_{t}\pseudy_{t}$
for all $t\in\naturals$.
\end{lem}
\begin{proof}
Premultiply \eqref{poscompanion} by $\b{\beta}^{+\trans}$ and then
use \eqref{vdef} to obtain
\begin{align}
\b{\xi}_{t}^{+} & =\b{\beta}_{1:p}^{+\trans}c+(I+\b{\beta}^{+\trans}\b{\alpha}^{+})\b{\xi}_{t-1}^{+}+\b{\beta}_{1:p}^{+\trans}v_{t}\nonumber \\
 & =\b{\beta}_{1:p}^{+\trans}c+(I+\b{\beta}^{+\trans}\b{\alpha}^{+})\b{\xi}_{t-1}^{+}\nonumber \\
 & \qquad\qquad+\b{\beta}_{1:p}^{+\trans}(\phi_{1}^{-}-\phi_{1}^{+})y_{t-1}^{-}+\b{\beta}_{1:p}^{+\trans}\sum_{i=2}^{k}(\phi_{i}^{-}-\phi_{i}^{+})y_{t-i}^{-}+\b{\beta}_{1:p}^{+\trans}u_{t}.\label{eq:xirep}
\end{align}
To obtain the law of motion for $y_{t}^{-}$, we take the first equation
in \eqref{poscompanion}
\[
\Delta y_{t}=e_{1}^{\trans}c+\b e_{1}^{\trans}\b{\alpha}^{+}\b{\xi}_{t-1}^{+}+e_{1}^{\trans}v_{t}
\]
(where $e_{1}$ and $\b e_{1}$ denote the first column of $I_{p}$
and $I_{kp}$ respectively), which using \eqref{vdef} can be rearranged
as
\begin{align}
y_{t}-y_{t-1}^{+} & =e_{1}^{\trans}c+[1+e_{1}^{\trans}(\phi_{1}^{-}-\phi_{1}^{+})]y_{t-1}^{-}\label{eq:pseudy}\\
 & \qquad\qquad+\sum_{i=2}^{k}e_{1}^{\trans}(\phi_{i}^{-}-\phi_{i}^{+})y_{t-i}^{-}+\b e_{1}^{\trans}\b{\alpha}^{+}\b{\xi}_{t-1}^{+}+e_{1}^{\trans}u_{t}.\nonumber 
\end{align}
It follows that if we define $\pseudy_{t}$ to equal the r.h.s., then
$y_{t}=y_{t-1}^{+}+\pseudy_{t}$, and so
\[
y_{t}^{-}=[y_{t}]_{-}=[y_{t-1}^{+}+\pseudy_{t}]_{-}=\delta_{t}\pseudy_{t}
\]
where $\delta_{t}\in[0,1]$ is defined in \eqref{delcens}. (Observe
that if $y_{t-1}^{+}+\pseudy_{t}<0$, then $\abv y_{t}<-y_{t-1}^{+}\leq0$,
so this object is well defined.) Making the substitution $y_{t-1}^{-}=\delta_{t-1}\pseudy_{t-1}$
on the r.h.s.\ of \eqref{xirep} and \eqref{pseudy}, we see that
the trajectories of $\{\b{\xi}_{t}^{+}\}$ and $\{y_{t}^{-}\}$ for
$t\geq1$ (conditional on the initial values $\{z_{t}\}_{t=-k+1}^{0}$)
will be reproduced exactly by\begin{subequations}\label{eq:coclas:sm}
\begin{align}
\b{\xi}_{t}^{+} & =\b{\beta}_{1:p}^{+\trans}c+(I+\b{\beta}^{+\trans}\b{\alpha}^{+})\b{\xi}_{t-1}^{+}\\
 & \qquad\qquad\qquad+\b{\beta}_{1:p}^{+\trans}(\phi_{1}^{-}-\phi_{1}^{+})\delta_{t-1}\pseudy_{t-1}+\b{\beta}_{1:p}^{+\trans}\sum_{i=2}^{k}(\phi_{i}^{-}-\phi_{i}^{+})y_{t-i}^{-}+\b{\beta}_{1:p}^{+\trans}u_{t},\nonumber \\
\pseudy_{t} & =e_{1}^{\trans}c+[1+e_{1}^{\trans}(\phi_{1}^{-}-\phi_{1}^{+})]\delta_{t-1}\pseudy_{t-1}\nonumber \\
 & \qquad\qquad\qquad+\sum_{i=2}^{k}e_{1}^{\trans}(\phi_{i}^{-}-\phi_{i}^{+})y_{t-i}^{-}+\b e_{1}^{\trans}\b{\alpha}^{+}\b{\xi}_{t-1}^{+}+e_{1}^{\trans}u_{t},\\
y_{t-1}^{-} & =\delta_{t-1}\pseudy_{t-1},
\end{align}
\end{subequations}with $\delta_{0}\defeq1$ and $\abv y_{0}\defeq y_{0}^{-}$.
Finally, we note that \eqref{srcomp} provides the companion form
representation of this system.
\end{proof}

By \citet[Prop.~1.8]{Jungers09}, \assref{cocens}\enuref{cocens:jsr}
implies that $\rho_{\jsr}(\{\b F_{\delta}\mid\delta\in[0,1]\})=\rho_{\jsr}(\{\b F_{0},\b F_{1}\})<1$.
In view of \assref{coerr} and \assref{cocens}\enuref{cocens:init},
applying \lemref{stochbound} to the representation \eqref{srcomp}
then yields that $\sup_{t\in\naturals}\smlnorm{\b{\zeta}_{t}}_{2+\delta_{u}}<\infty$,
and hence $\b{\zeta}_{t}\sim I^{\ast}(0)$. In particular, $\b{\xi}_{t}^{+}\sim I^{\ast}(0)$
and $y_{t}^{-}\sim I^{\ast}(0)$, which gives part~\enuref{thm:cocens:i0}
of the theorem.

For the purposes of proving part~\enuref{thm:cocens:wklim}, we shall
for the moment we replace \assref{det} by the weaker condition 
\begin{equation}
e_{1}^{\trans}P_{\beta_{\perp}^{+}}c=0,\label{eq:detprime-case1}
\end{equation}
for $P_{\beta_{\perp}^{+}}$ as defined in \eqref{Pbetaplus}, noting
only at the very end of the proof how the result simplifies when \assref{det}
holds. As discussed in \subsecref{deterministic}, the condition \eqref{detprime-case1}
(which appears there as \assref{detprime}) permits the model to impart
a deterministic trend to $x_{t}$, but not to $y_{t}$.

Our first step is to extract the common trend components. Define the
$q\times kp$ matrices
\begin{align*}
\b{\alpha}_{\perp}^{+\trans} & \defeq\alpha_{\perp}^{+\trans}\begin{bmatrix}I_{p} & -\Gamma_{1}^{+} & \cdots & -\Gamma_{k-1}^{+}\end{bmatrix} & \b{\beta}_{\perp}^{+\trans} & \defeq\beta_{\perp}^{+\trans}\begin{bmatrix}I_{p} & I_{p} & \cdots & I_{p}\end{bmatrix},
\end{align*}
which are orthogonal to $\b{\alpha}^{+}$ and $\b{\beta}^{+}$, and
the complementary projections 
\begin{align*}
\b P_{\b{\beta}_{\perp}^{+}} & \defeq\b{\beta}_{\perp}^{+}[\b{\alpha}_{\perp}^{+\trans}\b{\beta}_{\perp}^{+}]^{-1}\b{\alpha}_{\perp}^{+\trans} & \b P_{\b{\alpha}^{+}} & \defeq\b{\alpha}^{+}[\b{\beta}^{+\trans}\b{\alpha}^{+}]^{-1}\b{\beta}^{+\trans}
\end{align*}
(see \lemref{lincoint}). Premutiplying \eqref{poscompanion} by $\b{\alpha}_{\perp}^{+\trans}$
and cumulating yields
\begin{align*}
\b{\alpha}_{\perp}^{+\trans}\b z_{t} & =\b{\alpha}_{\perp}^{+\trans}\b z_{0}+\b{\alpha}_{\perp}^{+\trans}\sum_{s=1}^{t}\b v_{s}+(\b{\alpha}_{\perp}^{+\trans}\b c)t.
\end{align*}
Hence

\[
\b z_{t}=(\b P_{\b{\beta}_{\perp}^{+}}+\b P_{\b{\alpha}^{+}})\b z_{t}=\b P_{\b{\beta}_{\perp}^{+}}\left(\b z_{0}+\sum_{s=1}^{t}\b v_{s}+\b ct\right)+\b z_{\xi,t}
\]
where $\b z_{\xi,t}\defeq\b P_{\b{\alpha}^{+}}\b z_{t}=\b{\alpha}^{+}[\b{\beta}^{+\trans}\b{\alpha}^{+}]^{-1}\b{\xi}_{t}^{+}\sim I^{\ast}(0)$.

Since the top left $p\times p$ block of $\b P_{\b{\beta}_{\perp}^{+}}$
is $P_{\beta_{\perp}^{+}}=\beta_{\perp}^{+}[\alpha_{\perp}^{+\trans}\Gamma^{+}(1)\beta_{\perp}^{+}]^{-1}\alpha_{\perp}^{+\trans}$
(\lemref{lincoint}), the preceding implies
\begin{align}
z_{t} & =[\b P_{\b{\beta}_{\perp}^{+}}\b z_{0}]_{1:p}+P_{\beta_{\perp}^{+}}\sum_{s=1}^{t}v_{s}+(P_{\beta_{\perp}^{+}}c)t+z_{\xi,t}\label{eq:ztsmldecmp}
\end{align}
where $z_{\xi,t}\defeq[\b z_{\xi,t}]_{1:p}$. Under \assref{cocens}\enuref{cocens:init},
$n^{-1/2}z_{t}\inprob{\cal Z}_{0}$ for each $t\in\{-k+1,\ldots,0\}$;
hence
\begin{equation}
n^{-1/2}[\b P_{\b{\beta}_{\perp}^{+}}\b z_{0}]_{1:p}\inprob\beta_{\perp}^{+}[\alpha_{\perp}^{+\trans}\Gamma^{+}(1)\beta_{\perp}^{+}]^{-1}\alpha_{\perp}^{+\trans}\Gamma^{+}(1){\cal Z}_{0}={\cal Z}_{0}\label{eq:cocens-init}
\end{equation}
where the final equality follows since ${\cal Z}_{0}\in\ctspc\subset\spn\beta_{\perp}^{+}$.
Now recalling \eqref{vdef},
\begin{align*}
\sum_{s=1}^{t}v_{s} & =\sum_{s=1}^{t}u_{s}+(\pi^{-}-\pi^{+})\sum_{s=0}^{t-1}y_{s}^{-}+\sum_{i=1}^{k-1}(\gamma_{i}^{-}-\gamma_{i}^{+})[y_{t-i}^{-}-y_{-i}^{-}],
\end{align*}
we have from \eqref{ztsmldecmp} that
\[
z_{t}=[\b P_{\b{\beta}_{\perp}^{+}}\b z_{0}]_{1:p}+(P_{\beta_{\perp}^{+}}c)t+P_{\beta_{\perp}^{+}}\sum_{s=1}^{t}u_{s}+\regcoef\sum_{s=0}^{t-1}y_{s}^{-}+P_{\beta_{\perp}^{+}}\sum_{i=1}^{k-1}(\gamma_{i}^{-}-\gamma_{i}^{+})[y_{t-i}^{-}-y_{-i}^{-}]+z_{\xi,t},
\]
where $\regcoef\defeq P_{\beta_{\perp}^{+}}(\pi^{-}-\pi^{+})=P_{\beta_{\perp}^{+}}\pi^{-}$.
Using that $z_{t}=[\begin{smallmatrix}y_{t}\\
x_{t}
\end{smallmatrix}]=[\begin{smallmatrix}y_{t}^{+}\\
x_{t}
\end{smallmatrix}]+e_{1}y_{t}^{-}$, and defining
\[
\eta_{t}\defeq-(e_{1}+\regcoef)y_{t}^{-}+P_{\beta_{\perp}^{+}}\sum_{i=1}^{k-1}(\gamma_{i}^{-}-\gamma_{i}^{+})[y_{t-i}^{-}-y_{-i}^{-}]+z_{\xi,t},
\]
which is $I^{\ast}(0)$ by part~\enuref{thm:cocens:i0} of the theorem,
we may rewrite the preceding as
\begin{align}
\begin{bmatrix}y_{t}^{+}\\
x_{t}
\end{bmatrix} & =[\b P_{\b{\beta}_{\perp}^{+}}\b z_{0}]_{1:p}+(P_{\beta_{\perp}^{+}}c)t+\regcoef\sum_{s=0}^{t}y_{s}^{-}+P_{\beta_{\perp}^{+}}\sum_{s=1}^{t}u_{s}+\eta_{t}\label{eq:reginit}
\end{align}

Under \eqref{detprime-case1}, the first equation in \eqref{reginit}
is
\begin{align}
y_{t}^{+} & =e_{1}^{\trans}[\b P_{\b{\beta}_{\perp}^{+}}\b z_{0}]_{1:p}+\regcoef_{1}\sum_{s=0}^{t}y_{s}^{-}+e_{1}^{\trans}P_{\beta_{\perp}^{+}}\sum_{s=1}^{t}u_{s}+\eta_{1,t}\label{eq:reginit1}
\end{align}
where $\eta_{1,t}=e_{1}^{\trans}\eta_{t}$. Taking first differences
yields
\[
y_{t}^{+}-\regcoef_{1}y_{t}^{-}=y_{t-1}^{+}+e_{1}^{\trans}P_{\beta_{\perp}^{+}}u_{t}+\Delta\eta_{1,t}.
\]
Since $\regcoef_{1}<0$ by \assref{cocens}\enuref{cocens:regcoef},
only one of $y_{t}^{+}$ and $\regcoef y_{t}^{-}$ can be nonzero,
and must have opposite signs; hence
\[
y_{t}^{+}=[y_{t-1}^{+}+e_{1}^{\trans}P_{\beta_{\perp}^{+}}u_{t}+\Delta\eta_{1,t}]_{+}.
\]
As noted above $\eta_{t}\sim I^{\ast}(0)$, while $n^{-1/2}y_{0}\inprob{\cal Y}_{0}$
by \assref{cocens}\enuref{cocens:init}. Hence by \assref{coerr},
\begin{align*}
 & n^{-1/2}y_{0}+n^{-1/2}\sum_{t=1}^{\smlfloor{n\lambda}}(e_{1}^{\trans}P_{\beta_{\perp}^{+}}u_{t}+\Delta\eta_{t})\\
 & \qquad\qquad={\cal Y}_{0}+n^{-1/2}e_{1}^{\trans}P_{\beta_{\perp}^{+}}\sum_{t=1}^{\smlfloor{n\lambda}}u_{t}+n^{-1/2}(\eta_{1,\smlfloor{n\lambda}}-\eta_{1,0})+o_{p}(1)\\
 & \qquad\qquad\wkc{\cal Y}_{0}+e_{1}^{\trans}P_{\beta_{\perp}^{+}}U(\lambda)=e_{1}^{\trans}P_{\beta_{\perp}^{+}}[\Gamma^{+}(1){\cal Z}_{0}+U(\lambda)]=e_{1}^{\trans}P_{\beta_{\perp}^{+}}U_{0}(\lambda).
\end{align*}
on $D[0,1]$, where the penultimate equality follows from $P_{\beta_{\perp}^{+}}\Gamma^{+}(1){\cal Z}_{0}={\cal Z}_{0}$,
as per \eqref{cocens-init} above. Therefore by \citet[Lem.~A.1]{BD22},
and $y_{t}^{-}\sim I^{\ast}(0)$, 
\begin{align}
n^{-1/2}y_{\smlfloor{n\lambda}} & =n^{-1/2}y_{\smlfloor{n\lambda}}^{+}+o_{p}(1)\nonumber \\
 & \wkc e_{1}^{\trans}P_{\beta_{\perp}^{+}}U_{0}(\lambda)+\sup_{\lambda^{\prime}\leq\lambda}[-e_{1}^{\trans}P_{\beta_{\perp}^{+}}U_{0}(\lambda^{\prime})]_{+}=Y(\lambda).\label{eq:ycenslim}
\end{align}

Since \eqref{detprime-case1} permits a deterministic trend to be
present in $x_{t}$, as evident from \eqref{reginit} above, we consider
the linearly detrended process
\[
x_{t}^{d}\defeq x_{t}-(E_{-1}^{\trans}P_{\beta_{\perp}^{+}}c)t
\]
where $E_{-1}\in\reals^{p\times(p-1)}$ collects the final $p-1$
columns of $I_{p}$. To obtain the weak limit of $n^{-1/2}x_{\smlfloor{n\lambda}}^{d},$
we substitute
\[
\sum_{s=0}^{t}y_{s}^{-}=\regcoef_{1}^{-1}y_{t}^{+}-\regcoef_{1}^{-1}e_{1}^{\trans}\left([\b P_{\b{\beta}_{\perp}^{+}}\b z_{0}]_{1:p}+P_{\beta_{\perp}^{+}}\sum_{s=1}^{t}u_{s}+\eta_{t}\right)
\]
from \eqref{reginit1} into the final $p-1$ equations from \eqref{reginit},
to obtain
\begin{align*}
x_{t}^{d} & =x_{t}-(E_{-1}^{\trans}P_{\beta_{\perp}^{+}}c)t\\
 & =E_{-1}^{\trans}\left\{ [\b P_{\b{\beta}_{\perp}^{+}}\b z_{0}]_{1:p}+\regcoef\sum_{s=0}^{t}y_{s}^{-}+P_{\beta_{\perp}^{+}}\sum_{s=1}^{t}u_{s}+\eta_{t}\right\} \\
 & =E_{-1}^{\trans}\left\{ (I_{p}-\regcoef_{1}^{-1}\regcoef e_{1}^{\trans})\left[[\b P_{\b{\beta}_{\perp}^{+}}\b z_{0}]_{1:p}+P_{\beta_{\perp}^{+}}\sum_{s=1}^{t}u_{s}+\eta_{t}\right]+\regcoef_{1}^{-1}\regcoef y_{t}^{+}\right\} .
\end{align*}
for $E_{-1}$ the final $p-1$ columns of $I_{p}$. Hence, by $\eta_{t}\sim I^{\ast}(0)$,
\assref{coerr}, \eqref{cocens-init} and \eqref{ycenslim},
\begin{align}
n^{-1/2}x_{\smlfloor{n\lambda}}^{d} & =E_{-1}^{\trans}\left\{ (I_{p}-\regcoef_{1}^{-1}\regcoef e_{1}^{\trans})\left[{\cal Z}_{0}+P_{\beta_{\perp}^{+}}n^{-1/2}\sum_{s=1}^{\smlfloor{n\lambda}}u_{s}\right]+\regcoef_{1}^{-1}\regcoef n^{-1/2}y_{\smlfloor{n\lambda}}^{+}\right\} +o_{p}(1)\nonumber \\
 & \wkc E_{-1}^{\trans}\left\{ (I_{p}-\regcoef_{1}^{-1}\regcoef e_{1}^{\trans})P_{\beta_{\perp}^{+}}\left[\Gamma^{+}(1){\cal Z}_{0}+U(\lambda)\right]+\regcoef_{1}^{-1}\regcoef Y(\lambda)\right\} \nonumber \\
 & =E_{-1}^{\trans}\left\{ P_{\beta_{\perp}^{+}}U_{0}(\lambda)+\regcoef_{1}^{-1}\regcoef\sup_{\lambda^{\prime}\leq\lambda}[-e_{1}^{\trans}P_{\beta_{\perp}^{+}}U_{0}(\lambda^{\prime})]_{+}\right\} \label{eq:xcenslim}
\end{align}
on $D[0,1]$.

Finally, we suppose that \assref{det} holds in place of \eqref{detprime-case1}.
In this case, $P_{\beta_{\perp}^{+}}c=0$ and therefore $x_{t}^{d}=x_{t}$
for all $t\in\{1,\ldots,T\}$. Thus part~\ref{enu:thm:cocens:wklim}
of the theorem follows from \eqref{ycenslim} and \eqref{xcenslim}.\hfill\qedsymbol{}

\subsection{Proof of \thmref{coclas}}

Because $\{y_{t}\}$ now behaves like an integrated process in both
the positive and negative regimes, we have to take quite a different
approach from that utilised in case~\casecens{}. Recall $z_{t}^{\trans}=(y_{t},x_{t}^{\trans})$,
but now also define $z_{t}^{\ast\trans}\defeq(y_{t}^{+},y_{t}^{-},x_{t}^{\trans})$,
so that \eqref{vecm} may be written as
\begin{align*}
\Delta z_{t}=\begin{bmatrix}\Delta y_{t}\\
\Delta x_{t}
\end{bmatrix} & =c+\begin{bmatrix}\pi^{+} & \pi^{-} & \Pi^{x}\end{bmatrix}\begin{bmatrix}y_{t-1}^{+}\\
y_{t-1}^{-}\\
x_{t-1}
\end{bmatrix}+\sum_{i=1}^{k-1}\begin{bmatrix}\gamma_{i}^{+} & \gamma_{i}^{-} & \Gamma_{i}^{x}\end{bmatrix}\begin{bmatrix}\Delta y_{t-i}^{+}\\
\Delta y_{t-i}^{-}\\
\Delta x_{t-i}
\end{bmatrix}+u_{t}\\
 & =c+\Pi(y_{t-1})z_{t-1}+\sum_{i=1}^{k-1}\Gamma_{i}\Delta z_{t-i}^{\ast}+u_{t},
\end{align*}
where $\Pi(y)\defeq\Pi^{+}\indic^{+}(y)+\Pi^{-}\indic^{-}(y)$. Taking
$\b z_{t}^{\trans}\defeq(z_{t}^{\trans},z_{t-1}^{\ast\trans},\ldots,z_{t-k+1}^{\ast\trans})$
-- the first $p$ elements of which are $z_{t}$, not $z_{t}^{\ast}$
-- as the state vector, and conformably defining
\begin{align*}
\b{\Pi}(y) & \defeq\begin{bmatrix}\Pi(y)+\Gamma_{1}S(y) & -\Gamma_{1}+\Gamma_{2} & -\Gamma_{2}+\Gamma_{3} & \cdots & -\Gamma_{k-1}\\
S(y) & -I_{p+1}\\
 & I_{p+1} & -I_{p+1}\\
 &  & \ddots & \ddots\\
 &  &  & I_{p+1} & -I_{p+1}
\end{bmatrix} &  & \begin{aligned}\b c & \defeq\begin{bmatrix}c\\
0_{(p+1)(k-1)}
\end{bmatrix}\\
\\\b u_{t} & \defeq\begin{bmatrix}u_{t}\\
0_{(p+1)(k-1)}
\end{bmatrix}
\end{aligned}
\end{align*}
so that $\b{\Pi}(y)=\b{\alpha}\b{\beta}(y)^{\trans}$, where $\b{\alpha}$,
$\b{\beta}(y)$ and $S(y)$ are as defined in \eqref{alpabeta}--\eqref{S},
we can rewrite the system in companion form as
\begin{equation}
\Delta\b z_{t}=\b c+\b{\Pi}(y_{t-1})\b z_{t-1}+\b u_{t}.\label{eq:vecmcoclas}
\end{equation}
Note the definitions of $\b z_{t}$, $\b c$ and $\b u_{t}$ here
differ from those in case~\casecens{}.

Our next result is the counterpart of \lemref{statcomp} for the present
case, in that it provides a nonlinear VAR representation for the short-memory
components of the model: the equilibrium errors $\xi_{t}\defeq\beta(y_{t})^{\trans}z_{t}$
and the differences $\Delta z_{t}^{\ast}$, as collected in
\[
\b{\xi}_{t}\defeq\b{\beta}(y_{t})^{\trans}\b z_{t}=(\xi_{t}^{\trans},\Delta z_{t}^{\ast\trans},\ldots,\Delta z_{t-k+2}^{\ast\trans})^{\trans}.
\]

\begin{lem}
\label{lem:nlvarcase2}Suppose that \assref{dgp-canon}, \assref{co}
and \assref{coclas}\enuref{cocens:rank} hold. Let
\begin{equation}
\delta_{t}\defeq\begin{cases}
0 & \text{if }y_{t-1}\cdot y_{t}\geq0,\\
y_{t}/\Delta y_{t} & \text{if }y_{t-1}\cdot y_{t}<0
\end{cases}\label{eq:delocens}
\end{equation}
for $t\in\naturals$, which takes values only in $[0,1]$, and define
$\abv{\b{\beta}}_{t}\defeq(1-\delta_{t})\b{\beta}(y_{t-1})+\delta_{t}\b{\beta}(y_{t})$.
Then $\{\b{\xi}_{t}\}_{t\in\naturals}$ satisfies
\begin{equation}
\b{\xi}_{t}=\abv{\b{\beta}}_{t}^{\trans}\b c+[I_{r+(k-1)(p+1)}+\abv{\b{\beta}}_{t}^{\trans}\b{\alpha}]\b{\xi}_{t-1}+\abv{\b{\beta}}_{t}^{\trans}\b u_{t}.\label{eq:xicase2}
\end{equation}
\end{lem}
\begin{proof}
Using the definition of $\b{\xi}_{t}$, we may write \eqref{vecmcoclas}
as
\begin{equation}
\Delta\b z_{t}=\b c+\b{\alpha}\b{\xi}_{t-1}+\b u_{t}.\label{eq:delzxi}
\end{equation}
Let $\Delta\b{\beta}(y_{t})\defeq\b{\beta}(y_{t})-\b{\beta}(y_{t-1})$,
and note that $[\Delta\b{\beta}(y_{t})]^{\trans}\b z_{t}=[\Delta\b{\beta}(y_{t})]_{1:1}^{\trans}y_{t}$,
for $[\Delta\b{\beta}(y_{t})]_{1:1}$ the first row of $\Delta\b{\beta}(y_{t})$,
since all its other rows are zero. Thus with the aid of \eqref{delzxi}
we obtain
\begin{align}
\b{\xi}_{t} & =\b{\beta}(y_{t})^{\trans}\b z_{t}=\b{\beta}(y_{t-1})^{\trans}(\b z_{t-1}+\Delta\b z_{t})+[\Delta\b{\beta}(y_{t})]^{\trans}\b z_{t}\nonumber \\
 & =\b{\xi}_{t-1}+\b{\beta}(y_{t-1})^{\trans}\Delta\b z_{t}+[\Delta\b{\beta}(y_{t})]^{\trans}\b z_{t}\nonumber \\
 & =\b{\beta}(y_{t-1})^{\trans}\b c+[I+\b{\beta}(y_{t-1})^{\trans}\b{\alpha}]\b{\xi}_{t-1}+\b{\beta}(y_{t-1})^{\trans}\b u_{t}+[\Delta\b{\beta}(y_{t})]_{1:1}^{\trans}y_{t}.\label{eq:xirepcase2}
\end{align}
By construction, for $\delta_{t}$ as defined in \eqref{delocens},
we have that $y_{t}=\delta_{t}\Delta y_{t}$ whenever $y_{t-1}$ and
$y_{t}$ have opposite signs. Since this is also the only case in
which $\Delta\b{\beta}(y_{t})\neq0$, it follows that
\begin{align*}
[\Delta\b{\beta}(y_{t})]_{1:1}^{\trans}y_{t}=\delta_{t}[\Delta\b{\beta}(y_{t})]_{1:1}^{\trans}\Delta y_{t} & =\delta_{t}[\Delta\b{\beta}(y_{t})]^{\trans}\Delta\b z_{t}\\
 & =\delta_{t}[\b{\beta}(y_{t})-\b{\beta}(y_{t-1})]^{\trans}(\b c+\b{\alpha}\b{\xi}_{t-1}+\b u_{t}).
\end{align*}
Substituting this into \eqref{xirepcase2} and recalling the definition
of $\abv{\b{\beta}}_{t}$ thus yields \eqref{xicase2}.

Finally, observe that $\delta_{t}$ is nonzero only if $y_{t-1}$
and $y_{t}$ have opposite signs: in which case $\smlabs{y_{t}}\leq\smlabs{\Delta y_{t}}$,
and so $\smlabs{\delta_{t}}\leq1$. Since $y_{t}$ and $\Delta y_{t}$
must also have the same sign in this case, we deduce that $\delta_{t}\in[0,1]$,
as claimed.
\end{proof}
Since $\abv{\b{\beta}}_{t}$ is always a convex combination of $\b{\beta}(-1)$
and $\b{\beta}(+1)$, the autoregressive matrices in the representation
\eqref{xicase2} are contained in the set
\[
\mathcal{A}\defeq\{I+((1-\delta)\b{\beta}(+1)+\delta\b{\beta}(-1))^{\trans}\b{\alpha}\mid\delta\in[0,1]\}.
\]
By \citet[Prop.~1.8]{Jungers09}, $\rho_{\jsr}(\mathcal{A})=\rho_{\jsr}(\{I+\b{\beta}(+1)^{\trans}\b{\alpha},I+\b{\beta}(-1)^{\trans}\b{\alpha}\})$,
which is strictly less than unity by \assref{coclas}\enuref{coclas:jsr}.
That $\b{\xi}_{t}\sim I^{\ast}(0)$ now follows from \assref{coerr},
\assref{coclas}\enuref{coclas:init} and \lemref{stochbound}. In
particular, $\xi_{t}=\beta(y_{t})^{\trans}z_{t}$, $\Delta y_{t}^{+}$,
$\Delta y_{t}^{-}$ and $\Delta x_{t}$ are all $I^{\ast}(0)$, and
hence so too is
\[
\Delta y_{t}=(y_{t}-y_{t-1})=(y_{t}^{+}-y_{t-1}^{+})+(y_{t}^{-}-y_{t-1}^{-})=\Delta y_{t}^{+}+\Delta y_{t}^{-},
\]
as claimed in part~\enuref{thm:coclas:i0} of the theorem.

For the purposes of proving part~\enuref{thm:coclas:wklim} of the
theorem, we shall initially suppose (as we did in the proof of \thmref{cocens})
that \assref{det} is replaced by the weaker condition 
\begin{equation}
e_{1}^{\trans}P_{\beta_{\perp}}(+1)c=0,\label{eq:detprime-case2}
\end{equation}
for $P_{\beta_{\perp}}(+1)$ as defined in \eqref{projk1}, noting
only at the very end of the proof how the result simplifies when \assref{det}
holds.

Our next step is to extract the common trend components. Recall the
definitions of $\beta_{\perp}(y)$ and $\Gamma(1;y)$ in \eqref{betaperpfn},
and note that we may take
\begin{align*}
\b{\alpha}_{\perp}^{\trans} & =\alpha_{\perp}^{\trans}\begin{bmatrix}I_{p} & -\Gamma_{1} & \cdots & -\Gamma_{k-1}\end{bmatrix}, & \b{\beta}_{\perp}(y)^{\trans} & \defeq\beta_{\perp}(y)^{\trans}\begin{bmatrix}I_{p} & S(y)^{\trans} & \cdots & S(y)^{\trans}\end{bmatrix},
\end{align*}
as the orthocomplements of $\b{\alpha}$ and $\b{\beta}(y)$ respectively.
Observe that
\begin{equation}
\b{\alpha}_{\perp}^{\trans}\b{\beta}_{\perp}(y)=\alpha_{\perp}^{\trans}\left[I_{p}-\sum_{i=1}^{k-1}\Gamma_{i}S(y)\right]\beta_{\perp}(y)=\alpha_{\perp}^{\trans}\Gamma(1;y)\beta_{\perp}(y)\label{eq:aperpbperp}
\end{equation}
where the r.h.s.\ is invertible for each $y\in\reals$ by \assref{coclas}\ref{enu:coclas:det}.
By Lemma~A.1 in \citet{Han05EctJ}, $\b{\beta}(y)^{\trans}\b{\alpha}$
is therefore also invertible, whence it follows that the complementary
projections 
\begin{align*}
\b P_{\b{\beta}_{\perp}}(y) & \defeq\b{\beta}_{\perp}(y)[\b{\alpha}_{\perp}^{\trans}\b{\beta}_{\perp}(y)]^{-1}\b{\alpha}_{\perp}^{\trans} & \b P_{\b{\alpha}}(y) & \defeq\b{\alpha}[\b{\beta}(y)^{\trans}\b{\alpha}]^{-1}\b{\beta}(y)^{\trans}
\end{align*}
are well defined for each $y\in\reals$.

Now premultiplying \eqref{vecmcoclas} by $\b{\alpha}_{\perp}^{\trans}$
and cumulating yields
\[
\b{\alpha}_{\perp}^{\trans}\b z_{t}=\b{\alpha}_{\perp}^{\trans}\b z_{0}+\b{\alpha}_{\perp}^{\trans}\sum_{s=1}^{t}\b u_{s}+(\b{\alpha}_{\perp}^{\trans}\b c)t.
\]
Hence
\begin{equation}
\b z_{t}=[\b P_{\b{\beta}_{\perp}}(y_{t})+\b P_{\b{\alpha}}(y_{t})]\b z_{t}=\b P_{\b{\beta}_{\perp}}(y_{t})\left(\b z_{0}+\sum_{s=1}^{t}\b u_{s}+\b ct\right)+\b{\alpha}(\b{\beta}(y_{t})^{\trans}\b{\alpha})^{-1}\b{\xi}_{t}.\label{eq:z-costat-proj}
\end{equation}
In view of \eqref{aperpbperp} and the definitions of $\b{\alpha}_{\perp}$
and $\b{\beta}_{\perp}(y)$, the upper left $p\times p$ block of
$\b P_{\b{\beta}_{\perp}}(y)$ is $P_{\beta_{\perp}}(y)$, where the
latter is as defined in \eqref{projk1}. Under \assref{coclas}\enuref{coclas:init},
$n^{-1/2}z_{t}\inprob{\cal Z}_{0}$ for all $t\in\{-k+1,\ldots,0\}$,
and so
\[
n^{-1/2}\b{\alpha}_{\perp}^{\trans}\b z_{0}=n^{-1/2}\alpha_{\perp}^{\trans}\left[z_{0}-\sum_{i=1}^{k-1}\Gamma_{i}z_{-i}^{\ast}\right]\inprob\alpha_{\perp}^{\trans}\Gamma(1;{\cal Y}_{0}){\cal Z}_{0}.
\]
Hence, using the result of part~\ref{enu:thm:coclas:i0}, the first
$p$ rows of \eqref{z-costat-proj} give
\begin{equation}
\begin{bmatrix}y_{t}\\
x_{t}
\end{bmatrix}=P_{\beta_{\perp}}(y_{t})\left[n^{1/2}\Gamma(1;{\cal Y}_{0}){\cal Z}_{0}+\sum_{s=1}^{t}u_{s}+ct\right]+o_{p}(n^{1/2})\label{eq:beforewkc}
\end{equation}
uniformly in $t\in\{1,\ldots,n\}$. To determine the weak limit of
the preceding upon standardisation (and possibly linear detrending),
we first provide an auxiliary result on the mapping 
\begin{equation}
g(y,u)\defeq P_{\beta_{\perp}}(y)u.\label{eq:gdef}
\end{equation}

\begin{lem}
\label{lem:gcty}Suppose \assref{co} and \assref{coclas} hold. Then
\begin{enumerate}[itemsep=2pt,topsep=3pt]
\item taking $\vartheta^{\trans}\defeq e_{1}^{\trans}P_{\beta_{\perp}}(+1)\neq0$,
there exists a $\mu>0$ such that
\[
e_{1}^{\trans}P_{\beta_{\perp}}(y)=[\indic^{+}(y)+\mu\indic^{-}(y)]\vartheta^{\trans}\eqdef h(y)\vartheta^{\trans};
\]
\item if $u\in\reals^{p}$ is such that $\vartheta^{\trans}u=0$, then $P_{\beta_{\perp}}(+1)u=P_{\beta_{\perp}}(-1)u$;
and 
\item $g$ is Lipschitz continuous at every point in $D_{g}\defeq\{(y,u)\in\reals^{p+1}\mid y=h(y)\vartheta^{\trans}u\}$,
in the sense that for every $(y,u)\in D_{g}$ and $(y^{\prime},u^{\prime})\in\reals\times\reals^{p}$,
\begin{equation}
\smlabs{g(y,u)-g(y^{\prime},u^{\prime})}\leq C(\smlabs{y-y^{\prime}}+\smlnorm{u-u^{\prime}}).\label{eq:gLip}
\end{equation}
\end{enumerate}
\end{lem}
\begin{proof}
\textbf{(i).} Recall from \eqref{pikink} that $\beta^{\pm\trans}=[\beta_{y}^{\pm},\beta_{x}^{\trans}]=\beta_{x}^{\trans}[\theta^{\pm},I_{p-1}]$
for some $\theta^{\pm}\in\reals^{p-1}$, and take
\begin{equation}
\beta_{\perp}^{\pm}=\begin{bmatrix}1 & 0\\
-\theta^{\pm} & \beta_{x,\perp}
\end{bmatrix}=\beta_{\perp}(\pm1)\label{eq:betaperp}
\end{equation}
as per \eqref{betaperpfn}, so that $e_{1}^{\trans}\beta_{\perp}^{\pm}=e_{1}^{\trans}$;
it follows that $e_{1}^{\trans}P_{\beta_{\perp}}(+1)\neq0$. Since
\[
\Gamma^{\pm}(1)\beta_{\perp}^{\pm}=\begin{bmatrix}\gamma^{\pm}(1) & \Gamma^{x}(1)\end{bmatrix}\begin{bmatrix}1 & 0\\
-\theta^{\pm} & \beta_{x,\perp}
\end{bmatrix}=\begin{bmatrix}\gamma^{\pm}(1)-\Gamma^{x}(1)\theta^{\pm} & \Gamma^{x}(1)\beta_{x,\perp}\end{bmatrix}
\]
we have 
\[
\Gamma^{-}(1)\beta_{\perp}^{-}-\Gamma^{+}(1)\beta_{\perp}^{+}=\delta_{\gamma}e_{1}^{\trans},
\]
where $\delta_{\gamma}\defeq[\gamma^{-}(1)-\gamma^{+}(1)-\Gamma^{x}(1)(\theta^{-}-\theta^{+})]$.
Under \assref{coclas}\enuref{coclas:det}, we may apply \lemref{rk1perturb}
with $A=\alpha_{\perp}^{\trans}$, $B_{1}=\Gamma^{-}(1)\beta_{\perp}^{-}$,
$B_{2}=\Gamma^{+}(1)\beta_{\perp}^{+}$ and $d=e_{1}$, to obtain
\[
e_{1}^{\trans}[\alpha_{\perp}^{\trans}\Gamma^{-}(1)\beta_{\perp}^{-}]^{-1}=\mu e_{1}^{\trans}[\alpha_{\perp}^{\trans}\Gamma^{+}(1)\beta_{\perp}^{+}]^{-1}
\]
for some $\mu>0$. Hence, for $\vartheta^{\trans}=e_{1}^{\trans}P_{\beta_{\perp}}(+1)=e_{1}^{\trans}[\alpha_{\perp}^{\trans}\Gamma^{+}(1)\beta_{\perp}^{+}]^{-1}\alpha_{\perp}^{\trans}$
\begin{align*}
e_{1}^{\trans}P_{\beta_{\perp}}(y) & =e_{1}^{\trans}\beta_{\perp}(y)[\alpha_{\perp}^{\trans}\Gamma(1;y)\beta_{\perp}(y)]^{-1}\alpha_{\perp}^{\trans}=e_{1}^{\trans}[\alpha_{\perp}^{\trans}\Gamma(1;y)\beta_{\perp}(y)]^{-1}\alpha_{\perp}^{\trans}\\
 & =[\indic^{+}(y)+\mu\indic^{-}(y)]\vartheta^{\trans}=h(y)\vartheta^{\trans}
\end{align*}

\textbf{(ii).} Let $u\in\reals^{p}$ be such that $\vartheta^{\trans}u=0$.
Then
\[
0=\mu\vartheta^{\trans}u=e_{1}^{\trans}[\alpha_{\perp}^{\trans}\Gamma^{-}(1)\beta_{\perp}^{-}]^{-1}\alpha_{\perp}^{\trans}u=d^{\trans}(A^{\trans}B_{1})^{-1}v
\]
where $v\defeq\alpha_{\perp}^{\trans}u$, and $A$, $B_{1}$, $B_{2}$
and $d$ are as above. Hence by \lemref{rk1perturb},
\begin{align*}
0 & =[(A^{\trans}B_{2})^{-1}-(A^{\trans}B_{1})^{-1}]v=[(\alpha_{\perp}^{\trans}\Gamma^{+}(1)\beta_{\perp}^{+})^{-1}-(\alpha_{\perp}^{\trans}\Gamma^{-}(1)\beta_{\perp}^{-})^{-1}]\alpha_{\perp}^{\trans}u.
\end{align*}
Finally, noting that \eqref{betaperp} implies $\beta_{\perp}^{+}-\beta_{\perp}^{-}=\delta_{\beta}e_{1}^{\trans}$
for $\delta_{\beta}\defeq(0,(\theta^{-}-\theta^{+})^{\trans})^{\trans}$,
we have
\begin{align*}
[P_{\beta_{\perp}}(+1)-P_{\beta_{\perp}}(-1)]u & =\beta_{\perp}^{+}[(\alpha_{\perp}^{\trans}\Gamma^{+}(1)\beta_{\perp}^{+})^{-1}-(\alpha_{\perp}^{\trans}\Gamma^{-}(1)\beta_{\perp}^{-})^{-1}]\alpha_{\perp}u\\
 & \qquad\qquad-(\beta_{\perp}^{-}-\beta_{\perp}^{+})(\alpha_{\perp}^{\trans}\Gamma^{-}(1)\beta_{\perp}^{-})^{-1}\alpha_{\perp}u\\
 & =\delta_{\beta}e_{1}^{\trans}(\alpha_{\perp}^{\trans}\Gamma^{-}(1)\beta_{\perp}^{-})^{-1}\alpha_{\perp}^{\trans}u\\
 & =_{(1)}\delta_{\beta}e_{1}^{\trans}P_{\beta_{\perp}}(-1)u=_{(2)}\mu\delta_{\beta}\vartheta^{\trans}u=_{(3)}0,
\end{align*}
where $=_{(1)}$ holds by $e_{1}^{\trans}\beta_{\perp}^{-}=e_{1}^{\trans}$,
$=_{(2)}$ by the result of part~(i), and $=_{(3)}$ by our choice
of $u$.

\textbf{(iii).} Now let $(y,u)\in D_{g}$ and $(y^{\prime},u^{\prime})\in\reals\times\reals^{p}$.
Suppose first that $h(y)=h(y^{\prime})$: then
\[
\smlnorm{g(y,u)-g(y^{\prime},u^{\prime})}\leq\max\{\smlnorm{P_{\beta_{\perp}}(+1)},\smlnorm{P_{\beta_{\perp}}(-1)}\}\smlnorm{u-u^{\prime}}\leq C\smlnorm{u-u^{\prime}}.
\]
Suppose next that $h(y)\neq h(y^{\prime})$; then either: (a) $y\geq0$
and $y^{\prime}<0$; or (b) $y<0$ and $y^{\prime}\geq0$. Suppose
(a) holds. Then $\vartheta^{\trans}u\geq0$, and so
\[
\smlabs{y-y^{\prime}}=y-y^{\prime}=h(y)\vartheta^{\trans}u-y^{\prime}\geq h(y)\vartheta^{\trans}u\geq0
\]
whence
\begin{equation}
\smlabs{\vartheta^{\trans}u}\leq C_{0}\smlabs{y-y^{\prime}}\label{eq:vartheta-u}
\end{equation}
where $C_{0}\defeq\max\{1,\mu^{-1}\}\in(0,\infty)$. (If (b) holds
instead, then \eqref{vartheta-u} follows by an analogous argument.)
Now let $u_{\vartheta_{\perp}}\defeq u-\vartheta(\vartheta^{\trans}\vartheta)^{-1}\vartheta^{\trans}u$.
Then $\vartheta^{\trans}u_{\vartheta_{\perp}}=0$, and so
\begin{align}
g(y,u)-g(y^{\prime},u^{\prime}) & =P_{\beta_{\perp}}(+1)(u-u_{\vartheta_{\perp}}+u_{\vartheta_{\perp}})-P_{\beta_{\perp}}(-1)u^{\prime}\nonumber \\
 & =P_{\beta_{\perp}}(+1)(u-u_{\vartheta_{\perp}})+P_{\beta_{\perp}}(-1)(u_{\vartheta_{\perp}}-u^{\prime}),\label{eq:gdiff}
\end{align}
where the second equality follows by the result of part~(ii). By
\eqref{vartheta-u},
\begin{equation}
\smlnorm{u-u_{\vartheta_{\perp}}}\leq\smlnorm{\vartheta(\vartheta^{\trans}\vartheta)^{-1}}\smlabs{\vartheta^{\trans}u}\leq C_{1}\smlabs{y-y^{\prime}}\label{eq:u-utheta}
\end{equation}
where $C_{1}\defeq\smlnorm{\vartheta(\vartheta^{\trans}\vartheta)^{-1}}C_{0}$.
Further,
\begin{equation}
\smlnorm{u_{\vartheta_{\perp}}-u^{\prime}}\leq\smlnorm{u_{\vartheta_{\perp}}-u}+\smlnorm{u^{\prime}-u}\leq C_{1}\smlabs{y-y^{\prime}}+\smlnorm{u^{\prime}-u}.\label{eq:uprime-utheta}
\end{equation}
Thus \eqref{gLip} follows from \eqref{gdiff}--\eqref{uprime-utheta}
\end{proof}

Returning now to \eqref{beforewkc}, we note that \eqref{detprime-case2}
and \lemref{gcty} imply that $P_{\beta_{\perp}}(y)c=P_{\beta_{\perp}}(+1)c$
for all $y\in\reals$. Therefore by defining the detrended process
\[
x_{t}^{d}\defeq x_{t}-[E_{-1}^{\trans}P_{\beta_{\perp}}(+1)c]t,
\]
for $E_{-1}$ the final $p-1$ columns of $I_{p}$, we may rewrite
\eqref{beforewkc} as
\[
n^{-1/2}\begin{bmatrix}y_{t}\\
x_{t}^{d}
\end{bmatrix}=P_{\beta_{\perp}}(y_{t})\left[\Gamma(1;{\cal Y}_{0}){\cal Z}_{0}+n^{-1/2}\sum_{s=1}^{t}u_{s}\right]+o_{p}(1)
\]
uniformly in $t\in\{1,\ldots,n\}$, and so for $\lambda\in[0,1]$,
\begin{equation}
\begin{bmatrix}Y_{n}(\lambda)\\
X_{n}^{d}(\lambda)
\end{bmatrix}=P_{\beta_{\perp}}[Y_{n}(\lambda)]U_{n,0}(\lambda)+\abv o_{p}(1),\label{eq:wkcoclas}
\end{equation}
where $X_{n}^{d}(\lambda)\defeq n^{-1/2}x_{\smlfloor{n\lambda}}^{d}$,
$U_{n,0}(\lambda)\defeq\Gamma(1;{\cal Y}_{0}){\cal Z}_{0}+U_{n}(\lambda)$,
$\abv o_{p}(1)$ denotes a term that converges in probability to zero
uniformly over $\lambda\in[0,1]$, and we have used the fact that
$P_{\beta_{\perp}}(y)$ depends only on the sign of $y$.

It remains to determine the weak limit of the r.h.s.\ of \eqref{wkcoclas}.
By \lemref{gcty}, $e_{1}^{\trans}P_{\beta_{\perp}}(y)=h(y)\vartheta^{\trans}$,
and the first equation of \eqref{wkcoclas} becomes
\[
Y_{n}(\lambda)=h[Y_{n}(\lambda)]\vartheta^{\trans}U_{n,0}(\lambda)+\abv o_{p}(1).
\]
Now let $f(y)\defeq h(y)^{-1}y$. This is Lipschitz, and because $h(y)>0$
is bounded away from zero and infinity, it has an inverse $f^{-1}(y^{\ast})=h(y^{\ast})y^{\ast}$
that is also Lipschitz. Thus
\[
Y_{n}^{\ast}(\lambda)\defeq f[Y_{n}(\lambda)]=\vartheta^{\trans}U_{n,0}(\lambda)+h[Y_{n}(\lambda)]^{-1}\abv o_{p}(1)\wkc\vartheta^{\trans}U_{0}(\lambda)
\]
on $D[0,1]$, since $\sup_{\lambda\in[0,1]}h[Y_{n}(\lambda)]^{-1}\leq\max\{1,\mu^{-1}\}$.
Hence, by the CMT,
\begin{multline}
Y_{n}(\lambda)=f^{-1}[Y_{n}^{\ast}(\lambda)]\wkc f^{-1}[\vartheta^{\trans}U_{0}(\lambda)]\\
=h[\vartheta^{\trans}U_{0}(\lambda)]\vartheta^{\trans}U_{0}(\lambda)=e_{1}^{\trans}P_{\beta_{\perp}}[\vartheta^{\trans}U_{0}(\lambda)]U_{0}(\lambda)=Y(\lambda)\label{eq:Ylim}
\end{multline}
on $D[0,1]$. Now consider the remaining $p-1$ equations in \eqref{wkcoclas}.
Recalling the definition of $g$ in \eqref{gdef} above, we have
\[
X_{n}^{d}(\lambda)=E_{-1}g[Y_{n}(\lambda),U_{n,0}(\lambda)]+\abv o_{p}(1).
\]
By \eqref{Ylim}, $(Y_{n},U_{n,0})\wkc(Y,U_{0})$, which have continuous
paths and concentrate in the set
\[
\abv D\defeq\{(y,u)\in D_{\reals^{p+1}}[0,1]\mid[y(\lambda),u(\lambda)]\in D_{g}\text{ for all }\lambda\in[0,1]\}.
\]
By \lemref{gcty}, if $(y_{n},u_{n})\in D_{\reals^{p+1}}[0,1]$ converge
uniformly to $(y,u)\in\abv D$, then $g[y_{n}(\lambda),u_{n}(\lambda)]\goesto g[y(\lambda),u(\lambda)]$
uniformly over $\lambda\in[0,1]$. Hence by the CMT,
\begin{align}
X_{n}^{d}(\lambda)\wkc E_{-1}g[Y(\lambda),U_{0}(\lambda)] & =E_{-1}P_{\beta_{\perp}}[Y(\lambda)]U_{0}(\lambda)=E_{-1}P_{\beta_{\perp}}[\vartheta^{\trans}U_{0}(\lambda)]U_{0}(\lambda)\label{eq:Xlim}
\end{align}
on $D_{\reals^{p-1}}[0,1]$.

Finally, we suppose that \assref{det} holds in place of \eqref{detprime-case2}.
In this case, $P_{\beta_{\perp}}(+1)c=0$ and therefore $x_{t}^{d}=x_{t}$
for all $t\in\{1,\ldots,T\}$. Thus part~\enuref{thm:coclas:wklim}
of the theorem follows from \eqref{Ylim} and \eqref{Xlim}.\hfill\qedsymbol{}

\subsection{Proof of \thmref{costat}}

\label{app:proof-of-costat}

We first restate a simplified version of Lemma~\lemergodic{} from
\citet{DMW23stat}.
\begin{lem}
\label{lem:ergodic}Let $\{w_{t}\}_{t\in\naturals}$ be a random sequence
in $\reals^{p}$, and $\b w_{t-1}\defeq(w_{t-1}^{\trans},\ldots,w_{t-k}^{\trans})^{\trans}$.
Suppose
\begin{equation}
w_{t}=\mu(\b w_{t-1})+\b A(\b w_{t-1})\b w_{t-1}+u_{t}\label{eq:wsystem}
\end{equation}
where $\b A:\reals^{kp}\setmap\reals^{p}$ is such that $\b w\elmap\b A(\b w)\b w$
is continuous and homogeneous of degree one, $\mu(\b w)=o(\smlnorm{\b w})$
as $\smlnorm{\b w}\goesto\infty$, and that $\{u_{t}\}$ satisfies
\assref{err}. If the associated deterministic system
\[
\hat{w}_{t}=\b A(\hat{\b w}_{t-1})\hat{\b w}_{t-1},
\]
is stable (in the sense that $\hat{w}_{t}\goesto0$ for every initialisation
$\hat{\b w}_{0}\in\reals^{kp}$), then $\{w_{t}\}_{t\in\naturals}$
is $\mathcal{Q}$-geometrically ergodic, for $\mathcal{Q}(w)\defeq(1+\smlnorm w^{m_{0}})$,
with a stationary distribution that is absolutely continuous with
respect to Lebesgue measure, and which has finite $m_{0}$th moment.
\end{lem}

To prove \thmref{costat}, we need to put the model into a VECM form
different again from that used in the analysis of the preceding cases.
Recalling $z_{t}^{\trans}=(y_{t},x_{t}^{\trans})$, the canonical
CKSVAR \eqref{canon-var} may be written as
\[
z_{t}=c+\sum_{i=1}^{k}\Phi_{i}(y_{t-i})z_{t-i},
\]
where $\Phi_{i}(y)\defeq[\phi_{i}(y),\Phi^{x}]$ for $\phi_{i}(y)\defeq\phi_{i}^{+}\indic^{+}(y)+\phi_{i}^{-}\indic^{-}(y)$,
and hence
\[
\Delta z_{t}=c+\Pi(\b y_{t-1})z_{t-1}+\sum_{i=1}^{k-1}\Gamma_{i}(\b y_{t-1})\Delta z_{t-i}+u_{t}.
\]
where $\Pi(\b y_{t-i})\defeq\sum_{i=1}^{k}\Phi_{i}(y_{t-i})-I_{p}$,
and as per \eqref{gamstatinitial} above,
\[
\Gamma_{i}(\b y_{t-1})=-\sum_{j=i+1}^{k}\Phi_{j}(y_{t-j})=\begin{bmatrix}-\sum_{j=i+1}^{k}\phi_{j}(y_{t-j}), & -\sum_{j=i+1}^{k}\Phi_{j}^{x}\end{bmatrix}.
\]
Taking $\b z_{t}\defeq(z_{t}^{\trans},\ldots,z_{t-k+1}^{\trans})^{\trans}$,
we may write the system in companion form as
\begin{equation}
\Delta\b z_{t}=\b c+\b{\Pi}(\b y_{t-1})\b z_{t-1}+\b u_{t}\label{eq:compcase3}
\end{equation}
where for $\b{\alpha}(\b y_{t-1})$ and $\b{\beta}$ as defined in
\eqref{alphbetcase3},
\begin{align*}
\b{\Pi}(\boldsymbol{y}) & =\begin{bmatrix}\Pi(\b y)+\Gamma_{1}(\b y) & -\Gamma_{1}(\b y)+\Gamma_{2}(\b y) & \cdots & \Gamma_{k-1}(\b y)\\
I_{p} & -I_{p}\\
 & \ddots & \ddots\\
 &  & I_{p} & -I_{p}
\end{bmatrix} & \b c & \defeq\begin{bmatrix}c\\
0_{p}\\
\vdots\\
0_{p}
\end{bmatrix} & \b u_{t} & \defeq\begin{bmatrix}u_{t}\\
0_{p}\\
\vdots\\
0_{p}
\end{bmatrix}\\
 & =\b{\alpha}(\b y)\b{\beta}^{\trans}
\end{align*}
Hence for $\b{\xi}_{t}\defeq\b{\beta}^{\trans}\b z_{t}=(\xi_{t}^{\trans},\Delta z_{t}^{\trans},\ldots,\Delta z_{t-k+2}^{\trans})^{\trans}$,
where $\xi_{t}=\beta^{\trans}z_{t}$, it follows from \eqref{compcase3}
that
\begin{equation}
\b{\xi}_{t}=\b{\beta}^{\trans}\b c+(I_{p(k-1)+r}+\b{\beta}^{\trans}\b{\alpha}(\b y_{t-1}))\b{\xi}_{t-1}+\b{\beta}^{\trans}\b u_{t}.\label{eq:bxistat}
\end{equation}
To `close' \eqref{bxistat}, we need only to recognise that with
the ordering of the cointegrating vectors given in \eqref{costatfactor},
$y_{t-1}=e_{r,1}^{\trans}\xi_{t-1}$ and
\[
y_{t-i}=y_{t-1}-\sum_{j=1}^{i-1}\Delta y_{t-j}=e_{r,1}^{\trans}\xi_{t-1}-\sum_{j=1}^{i-1}e_{p,1}^{\trans}\Delta z_{t-j}
\]
so that $\b y_{t-1}$ can be extracted from $\b{\xi}_{t-1}$ via
\begin{equation}
\b y_{t-1}=\begin{bmatrix}y_{t-1}\\
y_{t-2}\\
\vdots\\
y_{t-k}
\end{bmatrix}=\begin{bmatrix}e_{r,1}^{\trans}\\
e_{r,1}^{\trans} & -e_{p,1}^{\trans}\\
\vdots & \vdots & \ddots\\
e_{r,1}^{\trans} & -e_{p,1}^{\trans} & \cdots & -e_{p,1}^{\trans}
\end{bmatrix}\begin{bmatrix}\xi_{t-1}\\
\Delta z_{t-1}\\
\vdots\\
\Delta z_{t-k+1}
\end{bmatrix}\eqdef\b G_{\xi}\b{\xi}_{t-1}.\label{eq:xitoy}
\end{equation}
Hence
\begin{align}
\b{\xi}_{t} & =\b{\beta}^{\trans}\b c+[I_{p(k-1)+r}+\b{\beta}^{\trans}\b{\alpha}(\b G_{\xi}\b{\xi}_{t-1})]\b{\xi}_{t-1}+\b{\beta}^{\trans}\b u_{t}.\label{eq:xisystem}
\end{align}

We would like to apply \lemref{ergodic} to this system, but we are
prevented from doing so because $\b{\beta}^{\trans}\b u_{t}$ does
not satisfy the required condition on the innovations. To remedy this,
we modify the state vector of the system as follows. Let $\abv{\beta}\defeq\beta(\beta^{\trans}\beta)^{-1}$,
$\abv{\beta}_{\perp}\defeq\beta_{\perp}(\beta_{\perp}^{\trans}\beta_{\perp})^{-1}$,
and recall $\chi_{t}=(\xi_{t}^{\trans},(\beta_{\perp}^{\trans}\Delta z_{t})^{\trans})^{\trans}$
and $\b{\chi}_{t}=(\chi_{t}^{\trans},\ldots,\chi_{t-k+1}^{\trans})^{\trans}$.
Since the r.h.s.\ of
\[
\Delta z_{t}=(\abv{\beta}_{\perp}\beta_{\perp}^{\trans}+\abv{\beta}\beta^{\trans})\Delta z_{t}=\abv{\beta}_{\perp}(\beta_{\perp}^{\trans}\Delta z_{t})+\abv{\beta}(\xi_{t}-\xi_{t-1})
\]
is a linear function of $\chi_{t}$ and $\chi_{t-1}$, there is a
matrix $\b H$ such that $\b{\xi}_{t}=\b H\b{\chi}_{t}$. Defining
\begin{align*}
M & \defeq\begin{bmatrix}I_{r} & 0 & 0_{r\times p(k-2)}\\
0 & \beta_{\perp}^{\trans} & 0_{q\times p(k-2)}
\end{bmatrix} & B & \defeq\begin{bmatrix}\beta^{\trans}\\
\beta_{\perp}^{\trans}
\end{bmatrix}
\end{align*}
we have $M\b{\xi}_{t}=\chi_{t}$, $M\b{\beta}^{\trans}\b c=Bc$ and
$M\b{\beta}^{\trans}\b u_{t}=Bu_{t}$, and so premultiplying \eqref{xisystem}
by $M$ yields
\begin{equation}
\chi_{t}=Bc+M[I_{p(k-1)+r}+\b{\beta}^{\trans}\b{\alpha}(\b G\b{\chi}_{t-1})]\b H\b{\chi}_{t-1}+Bu_{t}\label{eq:chisystem}
\end{equation}
where $\b G\defeq\b G_{\xi}\b H$.

We may now apply \lemref{ergodic} to \eqref{chisystem}. Since $B$
has full rank, \assref{err} is satisfied with $Bu_{t}$ taking the
place of $u_{t}$. Regarding the other conditions of the lemma, let
$\b F$ be such that $\b F\b z_{t-1}=\b y_{t-1}$. By the continuity
of the r.h.s.\ of the CKSVAR, the map
\[
\b z_{t-1}\elmap[I_{kp}+\b{\Pi}(\b F\b z_{t-1})]\b z_{t-1}
\]
which simply replicates the autoregressive part of the companion form
\eqref{compcase3} of the CKSVAR (in levels), is continuous. Hence
so too is 
\[
\b z_{t-1}\elmap[I_{p(k-1)+r}+\b{\beta}^{\trans}\b{\alpha}(\b F\b z_{t-1})]\b{\beta}^{\trans}\b z_{t-1}.
\]
$\b{\xi}_{t-1}=\b{\beta}^{\trans}\b z_{t-1}$ varies freely, and by
our arguments following \eqref{xisystem}, $\b F\b z_{t-1}$ depends
(continuously) only on the elements of $\b{\xi}_{t-1}$. Indeed, $\b F\b z_{t-1}=\b y_{t-1}=\b G_{\xi}\b{\xi}_{t-1}$
where $\b G_{\xi}$ is as defined in \eqref{xitoy}. Hence the mapping
\[
\b{\xi}_{t-1}\elmap[I_{p(k-1)+r}+\b{\beta}^{\trans}\b{\alpha}(\b G_{\xi}\b{\xi}_{t-1})]\b{\xi}_{t-1},
\]
that appears on the r.h.s.\ of \eqref{xisystem}, is also continuous.
Thus
\[
\b{\chi}_{t-1}\elmap M[I_{p(k-1)+r}+\b{\beta}^{\trans}\b{\alpha}(\b G\b{\chi}_{t-1})]\b H\b{\chi}_{t-1}
\]
is continuous by the composition of continuous maps. Since $\b{\alpha}(\b y)$
depends only on the signs of the elements of $\b y$, the preceding
is also homogeneous of degree one. Finally, since it is possible to
write $\b{\chi}_{t+1}$ as a (deterministic) function of $\xi_{t+1}$
and $\b{\xi}_{t}$, the deterministic system associated to \eqref{chisystem}
must be stable if that associated to \eqref{xisystem} is stable,
where the latter is maintained under \assref{costat}\ref{enu:costat:stable}.
The result now follows by \lemref{ergodic}.\hfill\qedsymbol{}

\subsection{Proof of \thmref{codet}}

Suppose initially that \assref{dgp-canon} holds. With the exception
of the final claim made in each, the proofs of Theorems~\ref{thm:cocens}
and \ref{thm:coclas} proceed under conditions \eqref{detprime-case1}
and \eqref{detprime-case2} respectively, which correspond to \assref{detprime}
rather than \assref{det}, and yield the claimed limit weak limits
for $X_{n}^{d}$. Thus for $z_{t}=(y_{t},x_{t}^{\trans})^{\trans}$
generated by a canonical CKSVAR, the theorem is proved. 

Now suppose that merely \assref{dgp} holds (along with \assref{cocens}\ass{$^\prime$}
or \assref{coclas}\ass{$^\prime$}). It is shown in \suppref{cointstruct}
that, in these cases, the derived canonical form (denoted by tildes)
satisfies $e_{1}^{\trans}\tilde{P}_{\beta_{\perp}^{+}}\tilde{c}=0$
in case~\casecens{}, and $e_{1}^{\trans}\tilde{P}_{\beta_{\perp}}(+1)\tilde{c}=0$
in case~\caseclas{}, i.e.\ \assref{detprime} holds in the canonical
form; the argument subsequently given in that appendix shows that
the conclusions of the theorem may then be transposed back to the
structural form.\hfill\qedsymbol{}

\section{Proofs of auxiliary lemmas}

\label{supp:auxiliaryproofs}
\begin{proof}[Proof of \lemref{stochbound}]
 Since $\rho_{\jsr}(\mset A)<1$, it follows by \citet[Prop.\ 1.4]{Jungers09}
that for any given $\gamma\in(\rho_{\jsr}(\mset A),1)$, there exists
a norm $\smlnorm{\cdot}_{\ast}$ on $\reals^{d_{w}}$ such that
\[
\smlnorm{w_{t}}_{\ast}\leq\smlnorm{c_{t}}_{\ast}+\smlnorm{A_{t}}_{\ast}\smlnorm{w_{t-1}}_{\ast}+\smlnorm{B_{t}}_{\ast}\smlnorm{\varepsilon_{t}}\leq\gamma\smlnorm{w_{t-1}}_{\ast}+C(1+\smlnorm{v_{t}}_{\ast})
\]
where $C\defeq\max\{\sup_{B\in\mset B}\smlnorm B_{\ast},\sup_{c\in{\cal C}}\smlnorm c_{\ast}\}<\infty$,
by the boundedness of $\mset B$ and ${\cal C}$. Hence by backward
substitution,
\[
\smlnorm{w_{t}}_{\ast}\leq C\sum_{s=0}^{t-1}\gamma^{s}(1+\smlnorm{v_{t-s}}_{\ast})+\gamma^{t}\smlnorm{w_{0}}_{\ast}.
\]
By the equivalence of norms on finite-dimensional spaces, there exists
a $C^{\prime}<\infty$ such that
\[
\smlnorm{w_{t}}\leq C^{\prime}\left[\sum_{s=0}^{t-1}\gamma^{s}(1+\smlnorm{v_{t-s}})+\gamma^{t}\smlnorm{w_{0}}\right],
\]
where $\smlnorm{\cdot}$ denotes the Euclidean norm on $\reals^{d_{w}}$.
Deduce that for any $m\in[1,m_{0}]$,
\begin{align*}
\smlnorm{w_{t}}_{m} & \leq C^{\prime}\left[\sum_{s=0}^{t-1}\gamma^{s}(1+\smlnorm{v_{t-s}}_{m})+\gamma^{t}\smlnorm{w_{0}}_{m}\right]\\
 & \leq C^{\prime}\left[\frac{1}{1-\gamma}\left(1+\max_{1\leq s\leq t}\smlnorm{v_{s}}_{m}\right)+\smlnorm{w_{0}}_{m}\right],
\end{align*}
where $\smlnorm{w_{t}}_{m}=\expect(\smlnorm{w_{t}}^{m})^{1/m}$, etc.,
and hence there exists a $C^{\prime\prime}<\infty$ such that
\[
\max_{1\leq t\leq n}\smlnorm{w_{t}}_{m}\leq C^{\prime\prime}\left[1+\max_{1\leq t\leq n}\smlnorm{v_{s}}_{m}+\smlnorm{w_{0}}_{m}\right].\qedhere
\]
\end{proof}
\begin{proof}[Proof of \lemref{rk1perturb}]
 Since $A^{\trans}B_{1}=A^{\trans}B_{2}+(A^{\trans}c)d^{\trans}$,
it follows by Cauchy's formula for a rank-one perturbation (\citealp{HJ13book},
(0.8.5.11)) that
\begin{align*}
\det(A^{\trans}B_{1}) & =\det(A^{\trans}B_{2})+d^{\trans}(\adj A^{\trans}B_{2})A^{\trans}c\\
 & =\det(A^{\trans}B_{2})\{1+d^{\trans}(A^{\trans}B_{2})^{-1}A^{\trans}c\}
\end{align*}
whence $\chi\defeq1+d^{\trans}(A^{\trans}B_{2})^{-1}A^{\trans}c=\det(A^{\trans}B_{2})^{-1}\det(A^{\trans}B_{1})>0$.
Hence, by the Sherman--Morrison--Woodbury formula (\citealp{HJ13book},
(0.7.4.2))
\begin{equation}
(A^{\trans}B_{1})^{-1}=(A^{\trans}B_{2})^{-1}-\frac{(A^{\trans}B_{2})^{-1}A^{\trans}cd^{\trans}(A^{\trans}B_{2})^{-1}}{1+d^{\trans}(A^{\trans}B_{2})^{-1}A^{\trans}c}.\label{eq:smw}
\end{equation}
Premultiplying the preceding by $d^{\trans}$ and rearranging yields
\begin{align*}
d^{\trans}(A^{\trans}B_{1})^{-1} & =\left\{ 1-\frac{d^{\trans}(A^{\trans}B_{2})^{-1}A^{\trans}c}{1+d^{\trans}(A^{\trans}B_{2})^{-1}A^{\trans}c}\right\} d^{\trans}(A^{\trans}B_{2})^{-1}=\chi^{-1}d^{\trans}(A^{\trans}B_{2})^{-1}
\end{align*}
and thus part~(i) holds with $\mu=\chi^{-1}$. Finally, let $v\in\reals^{n}$
be such that $d^{\trans}(A^{\trans}B_{1})^{-1}v=0$. By part~(i)
of the lemma, we must have that $d^{\trans}(A^{\trans}B_{2})^{-1}v=0$;
hence postmultiplying \eqref{smw} by $v$ immediately yields the
result of part~(ii).
\end{proof}
\begin{proof}[Proof of \lemref{lincoint}]
 The factorisation of $\Pi$ is immediate from the definition of
matrix rank; that of $\b{\Pi}$ then follows from direct calculation.
Part\ (iii) will also follow from direct calculation, once we have
verified (ii).

It remains to prove part~(ii), which is the converse to Lemma~A.2
in \citet{Han05EctJ}. We first observe that $I_{kp}+\b{\Pi}$ gives
the companion form matrix associated with the autoregressive polynomial
$\Phi(\lambda)$, in the sense that the roots of the latter coincide
with the reciprocals of the eigenvalues of the former (e.g.\ \citealp[Sec.~2.1.1]{Lut07}),
of which $q$ lie on the unit circle by assumption, and the remaining
$kp-q$ lie strictly inside. It is readily verified that $\b{\alpha}_{\perp}^{\trans}\b{\alpha}=0$
and $\b{\beta}_{\perp}^{\trans}\b{\beta}=0$, when $\b{\alpha}_{\perp}$
and $\b{\beta}_{\perp}$ have the form given in \eqref{abperp}. Moreover,
since $\beta_{\perp}$ is determined only up to its column span, we
are free to normalise it such that $\beta_{\perp}^{\trans}\beta_{\perp}=k^{-1}I_{q}$,
whence $\b{\beta}_{\perp}^{\trans}\b{\beta}_{\perp}=I_{q}$. Letting
$Q\defeq(\b{\beta}^{\trans}\b{\beta})^{-1/2}$ denote the positive
definite square root of $(\b{\beta}^{\trans}\b{\beta})^{-1}$, we
have that $[\b{\beta}Q,\b{\beta}_{\perp}]$ is a (full rank) orthogonal
matrix. Since
\[
\b{\beta}^{\trans}(I_{kp}+\b{\Pi})\b{\beta}_{\perp}=\b{\beta}^{\trans}(I_{kp}+\b{\alpha}\b{\beta}^{\trans})\b{\beta}_{\perp}=\b{\beta}^{\trans}\b{\beta}_{\perp}=0
\]
and, noting $\b{\beta}^{\trans}\b{\beta}Q=\b{\beta}^{\trans}\b{\beta}(\b{\beta}^{\trans}\b{\beta})^{-1/2}=(\b{\beta}^{\trans}\b{\beta})^{-1/2}=Q^{-1}$,
\begin{align*}
Q\b{\beta}^{\trans}(I_{kp}+\b{\Pi})\b{\beta}Q & =Q\b{\beta}^{\trans}(I_{kp}+\b{\alpha}\b{\beta}^{\trans})\b{\beta}Q\\
 & =Q(I_{kp-q}+\b{\beta}^{\trans}\b{\alpha})\b{\beta}^{\trans}\b{\beta}Q=Q(I_{kp-q}+\b{\beta}^{\trans}\b{\alpha})Q^{-1}.
\end{align*}
It follows that
\[
\begin{bmatrix}\b{\beta}_{\perp}^{\trans}\\
Q\b{\beta}^{\trans}
\end{bmatrix}[I_{kp}+\b{\Pi}]\begin{bmatrix}\b{\beta}_{\perp} & \b{\beta}Q\end{bmatrix}=\begin{bmatrix}I_{q} & \b{\beta}_{\perp}^{\trans}\b{\Pi}\b{\beta}Q\\
0 & Q(I_{kp-q}+\b{\beta}^{\trans}\b{\alpha})Q^{-1}
\end{bmatrix}.
\]
Hence the eigenvalues of $I_{kp-q}+\b{\beta}^{\trans}\b{\alpha}$
are the $kp-q$ eigenvalues of $I_{kp}+\b{\Pi}$ that are strictly
inside the unit circle. It then follows by Lemma~A.1 in \citet{Han05EctJ}
that $\b{\beta}^{\trans}\b{\alpha}$ and $\b{\alpha}_{\perp}^{\trans}\b{\beta}_{\perp}=\alpha_{\perp}^{\trans}\Gamma(1)\beta_{\perp}$
have full rank.
\end{proof}

\section{Transposition to the structural form}

\label{supp:cointstruct}

Here we verify the claims made in part~(i) of each of Remarks~\ref{rem:cocens}--\ref{rem:costat},
and complete the proof of \thmref{codet}. We suppose that \assref{dgp}
holds, so that $(y_{t},x_{t})$ are generated by the CKSVAR
\begin{equation}
\phi^{+}(L)y_{t}^{+}+\phi^{-}(L)y_{t}^{-}+\Phi^{x}(L)x_{t}=c+u_{t}.\label{eq:cksvar-structural}
\end{equation}
By \propref{canonical}, if $(\tilde{y}_{t},\tilde{x}_{t})$ are constructed
from $(y_{t},x_{t})$ via \eqref{canon-vars}, then $(\tilde{y}_{t},\tilde{x}_{t})$
follow a canonical CKSVAR, which we denote here as
\begin{equation}
\tilde{\phi}^{+}(L)\tilde{y}_{t}^{+}+\tilde{\phi}^{-}(L)\tilde{y}_{t}^{-}+\tilde{\Phi}^{x}(L)\tilde{x}_{t}=\tilde{c}+\tilde{u}_{t}.\label{eq:cksvar-canon}
\end{equation}
We use tildes to distinguish the series and parameters of the canonical
form \eqref{cksvar-structural}, from the original CKSVAR \eqref{cksvar-canon}
from which they were derived; the mapping between these is given in
\eqref{canon-vars}--\eqref{Q-canon}.

Our approach here is as follows. We show that if \assref{dgp}, \assref{co},
\assref{det} (or \assref{detprime}) and a suitably modified form
of \assref{cocens} (resp.\ \assref{coclas}, \assref{costat}) hold
for $(y_{t},x_{t})$, then \assref{dgp-canon}, \assref{co}, \assref{det}
(or \assref{detprime}) and \assref{cocens} (resp.\ \assref{coclas},
\assref{costat}) also hold for $(\tilde{y}_{t},\tilde{x}_{t})$.
Since \assref{coerr} (or \assref{err}) trivially carry over from
$u_{t}$ to $\tilde{u}_{t}$, this permits Theorems~\ref{thm:cocens}--\ref{thm:costat}
to be applied to $(\tilde{y}_{t},\tilde{x}_{t})$: their implications
for $(y_{t},x_{t})$ are then derived by inverting the mapping \eqref{canon-vars}--\eqref{Q-canon}.

Preliminary to these arguments, we define
\begin{equation}
(P^{\pm})^{-1}\defeq\begin{bmatrix}\bar{\phi}_{0,yy}^{\pm} & 0\\
\phi_{0,xy}^{\pm} & \Phi_{0,xx}
\end{bmatrix},\label{eq:Ppm}
\end{equation}
and prove the following auxiliary result.
\begin{prop}
\label{prop:canon-vecm}Suppose \assref{dgp} holds. Then
\begin{enumerate}
\item \label{enu:vecm:phitilde}$\tilde{\Phi}^{\pm}(\lambda)=Q\Phi^{\pm}(\lambda)P^{\pm}$;
\item \label{enu:vecm:roots}the roots of $\det\tilde{\Phi}^{(i)}(\lambda)$
and $\det\Phi^{(i)}(\lambda)$ are identical, for $(i)\in\{+,-\}$;
\item \label{enu:vecm:pi}$\Pi^{\pm}=Q^{-1}\tilde{\Pi}^{\pm}(P^{\pm})^{-1}$
and $\Pi^{x}=Q^{-1}\tilde{\Pi}^{x}\Phi_{0,xx}$; and
\item \label{enu:vecm:c}$c\in\spn\Pi^{+}\intsect\spn\Pi^{-}$ if and only
if $\tilde{c}=Qc\in\spn\tilde{\Pi}^{+}\intsect\tilde{\Pi}^{-}$.
\end{enumerate}
\end{prop}

\begin{proof}
By \propref{canonical}, a CKSVAR satisfying \assref{dgp} has a canonical
form, which relates to the original model via \eqref{canon-vars}--\eqref{canon-polys}.
Recall the definition of $(P^{\pm})^{-1}$ given in \eqref{Ppm}.
Since
\begin{align*}
P & =\begin{bmatrix}(\bar{\phi}_{0,yy}^{+})^{-1} & 0 & 0\\
0 & (\bar{\phi}_{0,yy}^{-})^{-1} & 0\\
-\Phi_{0,xx}^{-1}\phi_{0,xy}^{+}(\bar{\phi}_{0,yy}^{+})^{-1} & -\Phi_{0,xx}^{-1}\phi_{0,xy}^{-}(\bar{\phi}_{0,yy}^{-})^{-1} & \Phi_{0,xx}^{-1}
\end{bmatrix}
\end{align*}
and 
\begin{equation}
P^{\pm}=\begin{bmatrix}(\bar{\phi}_{0,yy}^{\pm})^{-1} & 0\\
-\Phi_{0,xx}^{-1}\phi_{0,xy}^{\pm}(\bar{\phi}_{0,yy}^{\pm})^{-1} & \Phi_{0,xx}^{-1}
\end{bmatrix},\label{eq:Ppm-canon}
\end{equation}
as may be verified using the partitioned inverse formula, \eqref{canon-polys}
may be equivalently stated as
\begin{equation}
\tilde{\Phi}^{\pm}(\lambda)=Q\Phi^{\pm}(\lambda)P^{\pm},\label{eq:phitilde}
\end{equation}
and thus \ref{enu:vecm:phitilde} and \ref{enu:vecm:roots} hold.
Since $\Pi^{\pm}=-\Phi^{\pm}(1)$ and $\tilde{\Pi}^{\pm}=-\tilde{\Phi}^{\pm}(1)$,
\ref{enu:vecm:pi} follows immediately upon taking $\lambda=1$ in
\eqref{phitilde}, and the fact that
\[
\Pi^{\pm}=Q^{-1}\begin{bmatrix}\tilde{\pi}^{\pm} & \tilde{\Pi}^{x}\end{bmatrix}(P^{\pm})^{-1}=Q^{-1}\begin{bmatrix}\tilde{\pi}^{\pm} & \tilde{\Pi}^{x}\end{bmatrix}\begin{bmatrix}\bar{\phi}_{0,yy}^{\pm} & 0\\
\phi_{0,xy}^{\pm} & \Phi_{0,xx}
\end{bmatrix}=\begin{bmatrix}\pi^{\pm} & Q^{-1}\tilde{\Pi}^{x}\Phi_{0xx}\end{bmatrix}.
\]
Finally, \ref{enu:vecm:c} follows from $\spn\tilde{\Pi}^{+}\intsect\tilde{\Pi}^{-}=\spn(Q\Pi^{+})\intsect(Q\Pi^{-})$.
\end{proof}

We suppose henceforth that $(y_{t},x_{t})$ satisfies \assref{dgp}
and \assref{co}, and for the moment also \assref{det}. Then $(\tilde{y}_{t},\tilde{x}_{t})$
satisfies \assref{dgp-canon}, \assref{co} and \assref{det} by Propositions~\ref{prop:canonical}
and \ref{prop:canon-vecm}. We now consider cases~\casecens{}--\casestat{}
in turn, before turning to the case where \assref{det} is relaxed
to \assref{detprime}.

\textbf{\casecens{}.} Suppose $(y_{t},x_{t})$ additionally satisfies
\assref{coerr} and \assref{cocens}, except with \assref{cocens}\ref{enu:cocens:jsr}
replaced by \ref{enu:cocens:jsr-struct}. We must verify that $(\tilde{y}_{t},\tilde{x}_{t})$
satisfies \assref{cocens}. $(\tilde{y}_{t},\tilde{x}_{t})$ satisfies
\assref{cocens}\ref{enu:cocens:rank} by \propref{canon-vecm}\ref{enu:vecm:pi};
\assref{cocens}\ref{enu:cocens:jsr} since $(y_{t},x_{t})$ satisfies
\ref{enu:cocens:jsr-struct}; and \assref{cocens}\ref{enu:cocens:init}
via \propref{canonical}.

It remains therefore to verify \assref{cocens}\ref{enu:cocens:regcoef}.
We note that as a further consequence of \propref{canon-vecm}\ref{enu:vecm:pi},
\begin{equation}
\tilde{\alpha}^{+}\tilde{\beta}^{+\trans}=\tilde{\Pi}^{+}=Q\Pi^{+}P^{+}=Q\alpha^{+}\beta^{+\trans}P^{+}\label{eq:alphabetaplus}
\end{equation}
so that $\tilde{\alpha}^{+}=Q\alpha^{+}$, $\tilde{\beta}^{+}=(P^{+})^{\trans}\beta^{+}$,
and we may take $\tilde{\alpha}_{\perp}^{+}=(Q^{-1})^{\trans}\alpha_{\perp}^{+}$
and $\tilde{\beta}_{\perp}^{+}=(P^{+})^{-1}\beta_{\perp}^{+}$. Further,
$\tilde{\Gamma}^{+}(\lambda)=Q\Gamma^{+}(\lambda)P^{+}$, and so
\begin{align}
\tilde{P}_{\beta_{\perp}^{+}} & =\tilde{\beta}_{\perp}^{+}[\tilde{\alpha}_{\perp}^{+\trans}\tilde{\Gamma}^{+}(1)\tilde{\beta}_{\perp}^{+}]^{-1}\tilde{\alpha}_{\perp}^{+\trans}=(P^{+})^{-1}\beta_{\perp}^{+}[\alpha_{\perp}^{+\trans}\Gamma^{+}(1)\beta_{\perp}^{+}]^{-1}\alpha_{\perp}^{+\trans}Q^{-1}.\nonumber \\
 & =(P^{+})^{-1}P_{\beta_{\perp}^{+}}Q^{-1}.\label{eq:projtilde}
\end{align}
Hence, using \propref{canon-vecm}\ref{enu:vecm:pi} again,
\begin{align*}
\tilde{\kappa}=\tilde{P}_{\beta_{\perp}^{+}}\tilde{\pi}^{-}=\tilde{P}_{\beta_{\perp}^{+}}\tilde{\Pi}^{-}e_{1} & =(P^{+})^{-1}P_{\beta_{\perp}^{+}}\Pi^{-}P^{-}e_{1},
\end{align*}
where
\begin{align*}
\Pi^{-}P^{-}e_{1} & =\begin{bmatrix}\pi^{-} & \Pi^{x}\end{bmatrix}\begin{bmatrix}(\bar{\phi}_{0,yy}^{-})^{-1} & 0\\
-\Phi_{0,xx}^{-1}\phi_{0,xy}^{-}(\bar{\phi}_{0,yy}^{-})^{-1} & \Phi_{0,xx}^{-1}
\end{bmatrix}\begin{bmatrix}1\\
0
\end{bmatrix}\\
 & =(\bar{\phi}_{0,yy}^{-})^{-1}[\pi^{-}-\Pi^{x}\Phi_{0,xx}^{-1}\phi_{0,xy}^{-}],
\end{align*}
and hence, using that $P_{\beta_{\perp}^{+}}\Pi^{x}=0$,
\begin{equation}
\tilde{\kappa}=(\bar{\phi}_{0,yy}^{-})^{-1}(P^{+})^{-1}P_{\beta_{\perp}^{+}}[\pi^{-}-\Pi^{x}\Phi_{0,xx}^{-1}\phi_{0,xy}^{-}]=(\bar{\phi}_{0,yy}^{-})^{-1}(P^{+})^{-1}\kappa\label{eq:kaptilde}
\end{equation}
where $\kappa=P_{\beta_{\perp}^{+}}\pi^{-}$. Since $e_{1}^{\trans}(P^{+})^{-1}=\bar{\phi}_{0,yy}^{+}e_{1}^{\trans}$,
it follows that
\begin{equation}
\tilde{\kappa}_{1}=e_{1}^{\trans}\tilde{\kappa}=(\bar{\phi}_{0,yy}^{-})^{-1}\bar{\phi}_{0,yy}^{+}\kappa_{1},\label{eq:kap1tilde}
\end{equation}
and so $\sgn\tilde{\kappa}_{1}=\sgn\kappa_{1}$, since $\bar{\phi}_{0,yy}^{\pm}>0$
by \assref{dgp}\enuref{dgp:wlog}. Thus $(\tilde{y}_{t},\tilde{x}_{t})$
satisfies \assref{cocens}\enuref{cocens:regcoef}, since $(y_{t},x_{t})$
does.

It follows that \thmref{cocens} applies to $(\tilde{y}_{t},\tilde{x}_{t})$.
Regarding the conclusions of that theorem, note that in this case
$\tilde{{\cal Y}}_{0}$ and ${\cal Y}_{0}$ must be non-negative.
Hence $\tilde{{\cal Z}}_{0}=(P^{+})^{-1}{\cal Z}_{0}$ by \eqref{canon-vars},
and so by \propref{canonical}
\[
\tilde{U}_{0}(\lambda)=\tilde{\Gamma}^{+}(1)\tilde{{\cal Z}}_{0}+\tilde{U}(\lambda)=Q[\Gamma^{+}(1){\cal Z}_{0}+U(\lambda)]=QU_{0}(\lambda)
\]
whence $P_{\tilde{\beta}_{\perp}^{+}}\tilde{U}_{0}(\lambda)=(P^{+})^{-1}P_{\beta_{\perp}^{+}}U_{0}(\lambda)$.
Since $\tilde{Y}=\tilde{Y}^{+}$ in this case, we have 
\begin{align}
Z(\lambda)=P^{+}\tilde{Z}(\lambda) & =P^{+}\left\{ P_{\tilde{\beta}_{\perp}^{+}}\tilde{U}_{0}(\lambda)+\tilde{\kappa}_{1}^{-1}\tilde{\kappa}\sup_{\lambda^{\prime}\leq\lambda}[-e_{1}^{\trans}P_{\tilde{\beta}_{\perp}^{+}}\tilde{U}_{0}(\lambda)]_{+}\right\} \nonumber \\
 & =_{(1)}P_{\beta_{\perp}^{+}}U_{0}(\lambda)+(\bar{\phi}_{0,yy}^{-})^{-1}\tilde{\kappa}_{1}^{-1}\kappa\sup_{\lambda^{\prime}\leq\lambda}[-e_{1}^{\trans}(P^{+})^{-1}P_{\beta_{\perp}^{+}}U_{0}(\lambda)]_{+}\nonumber \\
 & =_{(2)}P_{\beta_{\perp}^{+}}U_{0}(\lambda)+\kappa_{1}^{-1}\kappa\sup_{\lambda^{\prime}\leq\lambda}[-e_{1}^{\trans}P_{\beta_{\perp}^{+}}U_{0}(\lambda)]_{+},\label{eq:Z-struct}
\end{align}
where $=_{(1)}$ follows by \eqref{kaptilde}, and $=_{(2)}$ by \eqref{kap1tilde}
and $e_{1}^{\trans}(P^{+})^{-1}=\bar{\phi}_{0,yy}^{+}e_{1}^{\trans}$.
It follows immediately that $\beta^{+\trans}z_{t}$ and $y_{t}^{-}$
are $I^{\ast}(0)$. Since $\tilde{y}_{t}^{-}\sim I^{\ast}(0)$, it
follows that $\Delta\tilde{y}_{t}^{-}\sim I^{\ast}(0)$, and therefore
so too is $\Delta\tilde{y}_{t}^{+}=\Delta\tilde{y}_{t}-\Delta\tilde{y}_{t}^{-}$.
Hence
\begin{equation}
\begin{bmatrix}\Delta y_{t}^{+}\\
\Delta y_{t}^{-}\\
\Delta x_{t}
\end{bmatrix}=P\begin{bmatrix}\Delta\tilde{y}_{t}^{+}\\
\Delta\tilde{y}_{t}^{-}\\
\Delta\tilde{x}_{t}
\end{bmatrix}\sim I^{\ast}(0)\label{eq:Diffi0}
\end{equation}
implies that $\Delta z_{t}\sim I^{\ast}(0)$. Thus the conclusions
of the theorem hold also for $(y_{t},x_{t})$, exactly as stated.

\textbf{\caseclas{}.} Suppose $(y_{t},x_{t})$ additionally satisfies
\assref{coerr} and \assref{coclas}, except with \assref{coclas}\ref{enu:coclas:jsr}
replaced by \ref{enu:coclas:jsr-struct}. Analogously to case~\casecens{},
$(\tilde{y}_{t},\tilde{x}_{t})$ satisfies \assref{coclas}\ref{enu:coclas:rank}
by \propref{canon-vecm}\ref{enu:vecm:pi}; \assref{coclas}\ref{enu:coclas:jsr}
since $(y_{t},x_{t})$ satisfies \ref{enu:coclas:jsr-struct}; and
\assref{coclas}\ref{enu:coclas:init} via \propref{canonical}.

Regarding \assref{coclas}\ref{enu:coclas:det}, we first note that
similarly to \eqref{alphabetaplus},
\[
\tilde{\alpha}\tilde{\beta}^{\pm\trans}=\tilde{\Pi}^{\pm}=Q\Pi^{\pm}P^{\pm}=Q\alpha\beta^{\pm\trans}P^{\pm}
\]
so that $\tilde{\alpha}=Q\alpha$ and $\tilde{\beta}^{\pm}=(P^{\pm})^{\trans}\beta^{\pm}$,
and we may take $\tilde{\alpha}_{\perp}=(Q^{-1})^{\trans}\alpha_{\perp}$.
With respect to $\tilde{\beta}_{\perp}^{\pm}=\tilde{\beta}_{\perp}(\pm1)$,
for the purposes of verifying \assref{coclas}\ref{enu:coclas:det},
we need to choose this such that its first row is equal to $e_{q,1}^{\trans}$,
and its final columns are regime-invariant, as per \eqref{betaperpfn}
(see \remref{coclas}\ref{enu:rem:coclas:detsign}). This may be engineered
by defining 
\begin{align}
P(y) & \defeq P^{+}\indic^{+}(y)+P^{-}\indic^{-}(y) & \bar{\phi}_{0,yy}(y) & \defeq\bar{\phi}_{0,yy}^{+}\indic^{+}(y)+\bar{\phi}_{0,yy}^{-}P^{-}\indic^{-}(y)\label{eq:Pydef}
\end{align}
and $M(y)=\diag\{\bar{\phi}_{0,yy}^{-1}(y),I_{p-1}\}$, and then taking
\begin{align*}
\tilde{\beta}_{\perp}(y)\defeq P(y)^{-1}\beta_{\perp}(y)M(y) & =\begin{bmatrix}\bar{\phi}_{0,yy}(y) & 0\\
\phi_{0,xy}(y) & \Phi_{0,xx}
\end{bmatrix}\begin{bmatrix}1 & 0\\
-\theta(y) & \beta_{x,\perp}
\end{bmatrix}\begin{bmatrix}\bar{\phi}_{0,yy}^{-1}(y) & 0\\
0 & I_{p-1}
\end{bmatrix}\\
 & \eqdef\begin{bmatrix}1\\
-\tilde{\theta}(y) & \tilde{\beta}_{x,\perp}
\end{bmatrix}
\end{align*}
where $\tilde{\beta}_{x,\perp}=\Phi_{0,xx}\beta_{x,\perp}$. Since
$\tilde{\Gamma}(1;y)=Q\Gamma(1;y)P(y)$ by \propref{canon-vecm}\ref{enu:vecm:pi},
we have 
\[
\det\tilde{\alpha}_{\perp}^{\trans}\tilde{\Gamma}(1;y)\tilde{\beta}_{\perp}(y)=\det\alpha_{\perp}^{\trans}\Gamma(1;y)\beta_{\perp}(y)\det M(y).
\]
Noting that $\sgn\det M(y)=\sgn\bar{\phi}_{0,yy}(y)>0$, it follows
that so that $(\tilde{y}_{t},\tilde{x}_{t})$ satisfies \assref{coclas}\enuref{coclas:det},
since $(y_{t},x_{t})$ does.

Hence \thmref{coclas} applies to $(\tilde{y}_{t},\tilde{x}_{t})$.
Regarding the conclusions of that theorem, we first note that $\sgn y_{t}=\sgn\tilde{y}_{t}$,
$\sgn{\cal Y}_{0}=\sgn{\cal \tilde{Y}}_{0}$, 
\[
\tilde{U}_{0}(\lambda)=\tilde{\Gamma}(1;{\cal \tilde{Y}}_{0})\tilde{{\cal Z}}_{0}+\tilde{U}(\lambda)=Q[\Gamma(1;{\cal Y}_{0}){\cal Z}_{0}+U(\lambda)],
\]
and
\[
P_{\tilde{\beta}_{\perp}}(y)=\tilde{\beta}_{\perp}(y)[\tilde{\alpha}_{\perp}^{\trans}\tilde{\Gamma}(1;y)\tilde{\beta}_{\perp}(y)]^{-1}\tilde{\alpha}_{\perp}^{\trans}=P(y)^{-1}P_{\beta_{\perp}}(y)Q^{-1}.
\]
Hence
\begin{align}
Z(\lambda)=P[\tilde{Y}(\lambda)]\tilde{Z}(\lambda) & =P[\tilde{Y}(\lambda)]P_{\tilde{\beta}_{\perp}}[\tilde{Y}(\lambda)]\tilde{U}_{0}(\lambda)=P_{\beta_{\perp}}[Y(\lambda)]U_{0}(\lambda).
\end{align}
Further, 
\[
\beta(y_{t})^{\trans}z_{t}=\beta(y_{t})^{\trans}z_{t}=\beta(y_{t})^{\trans}P(\tilde{y}_{t})\tilde{z}_{t}=\tilde{\beta}(\tilde{y}_{t})^{\trans}\tilde{z}_{t}\sim I^{\ast}(0).
\]
Since $\tilde{y}_{t}\sim I^{\ast}(1)$, it follows that $\tilde{y}_{t}^{\pm}\sim I^{\ast}(1)$,
and hence $\Delta\tilde{y}_{t}^{\pm}\sim I^{\ast}(0)$. That $\Delta z_{t}\sim I^{\ast}(0)$
then follows as in \eqref{Diffi0}. Deduce that the conclusions of
the theorem hold also for $(y_{t},x_{t})$.

\textbf{\casestat{}.} Suppose $(y_{t},x_{t})$ additionally satisfies
\assref{err} and \assref{costat}, except with \assref{costat}\ref{enu:costat:stable}
replaced by \ref{enu:costat:stable-struct}. Since $(\tilde{y}_{t},\tilde{x}_{t})$
satisfies \assref{costat}\ref{enu:costat:rank} by \propref{canon-vecm}\ref{enu:vecm:pi},
it follows from \thmref{costat} that $\tilde{\b{\chi}}_{t}\defeq(\tilde{\chi}_{t}^{\trans},\ldots,\tilde{\chi}_{t-k+1}^{\trans})^{\trans}$
is ${\cal Q}$-geometrically ergodic.

Recalling the definitions in \eqref{Pydef} above, we may rewrite
\eqref{canon-vars} as
\begin{equation}
\begin{bmatrix}\tilde{y}_{t}\\
\tilde{x}_{t}
\end{bmatrix}=\tilde{z}_{t}=P(y_{t})^{-1}z_{t}=\begin{bmatrix}\bar{\phi}_{0,yy}(y_{t}) & 0\\
\phi_{0,xy}(y_{t}) & \Phi_{0,xx}
\end{bmatrix}\begin{bmatrix}y_{t}\\
x_{t}
\end{bmatrix}\label{eq:tildeagain}
\end{equation}
and using the fact that $\sgn y_{t}=\sgn\tilde{y}_{t}$,
\begin{equation}
\begin{bmatrix}y_{t}\\
x_{t}
\end{bmatrix}=\begin{bmatrix}\bar{\phi}_{0,yy}^{-1}(\tilde{y}_{t}) & 0\\
\bar{\phi}_{0,xy}(\tilde{y}_{t}) & \Phi_{0,xx}^{-1}
\end{bmatrix}\begin{bmatrix}\tilde{y}_{t}\\
\tilde{x}_{t}
\end{bmatrix}\label{eq:ztoztildeagain}
\end{equation}
where $\bar{\phi}_{0,xy}(y)\defeq-\Phi_{0,xx}^{-1}\phi_{0,xy}(y)\bar{\phi}_{0,yy}^{-1}(y)$.
Further, we have by \propref{canon-vecm}\ref{enu:vecm:pi} that
\begin{equation}
\tilde{\alpha}_{x}\tilde{\beta}_{x}^{\trans}=\tilde{\Pi}^{x}=Q\Pi^{x}\Phi_{0,xx}^{-1}=Q\alpha_{x}\beta_{x}^{\trans}\Phi_{0,xx}^{-1}.\label{eq:tildeabx}
\end{equation}

Hence by \eqref{tildeagain} and \eqref{tildeabx}
\begin{multline*}
\begin{bmatrix}\tilde{y}_{t}\\
\tilde{\beta}_{x}^{\trans}\tilde{x}_{t}
\end{bmatrix}=\tilde{\beta}^{\trans}\tilde{z}_{t}=\begin{bmatrix}1 & 0\\
0 & \tilde{\beta}_{x}^{\trans}
\end{bmatrix}P(y_{t})^{-1}z_{t}=\begin{bmatrix}1 & 0\\
0 & \beta_{x}^{\trans}\Phi_{0,xx}^{-1}
\end{bmatrix}\begin{bmatrix}\bar{\phi}_{0,yy}(y_{t}) & 0\\
\phi_{0,xy}(y_{t}) & \Phi_{0,xx}
\end{bmatrix}\begin{bmatrix}y_{t}\\
x_{t}
\end{bmatrix}\\
=\begin{bmatrix}\bar{\phi}_{0,yy}(y_{t}) & 0\\
\beta_{x}^{\trans}\Phi_{0,xx}^{-1}\phi_{0,xy}(y_{t}) & I_{r-1}
\end{bmatrix}\begin{bmatrix}y_{t}\\
\beta_{x}^{\trans}x_{t}
\end{bmatrix},
\end{multline*}
and so, since $\sgn y_{t}=\sgn\tilde{y}_{t}$,
\[
\beta^{\trans}z_{t}=\begin{bmatrix}y_{t}\\
\beta_{x}^{\trans}x_{t}
\end{bmatrix}=\begin{bmatrix}\bar{\phi}_{0,yy}^{-1}(\tilde{y}_{t}) & 0\\
-\beta_{x}^{\trans}\bar{\phi}_{0,xy}(\tilde{y}_{t}) & I_{r-1}
\end{bmatrix}\begin{bmatrix}\tilde{y}_{t}\\
\tilde{\beta}_{x}^{\trans}\tilde{x}_{t}
\end{bmatrix}
\]
so that the l.h.s.\ is a Lipschitz continuous function of $\tilde{\beta}^{\trans}\tilde{z}_{t}$,
as claimed. Further, we have from \eqref{ztoztildeagain} that 
\[
\Delta x_{t}=\Delta[\bar{\phi}_{0,xy}(\tilde{y}_{t})\tilde{y}_{t}]+\Phi_{0,xx}^{-1}\Delta\tilde{x}_{t}
\]
where $\tilde{y}_{t}\elmap\bar{\phi}_{0,xy}(\tilde{y}_{t})\tilde{y}_{t}$
is a Lipschitz continuous function of (the first element of) $\tilde{\chi}_{t}$,
and since
\[
\Delta\tilde{x}_{t}=(\abv{\beta}_{x,\perp}\tilde{\beta}_{x,\perp}^{\trans}+\abv{\beta}_{x}\tilde{\beta}_{x}^{\trans})\Delta\tilde{x}_{t}=\abv{\beta}_{x,\perp}(\tilde{\beta}_{x,\perp}^{\trans}\Delta\tilde{x}_{t})+\abv{\beta}_{x}\Delta(\tilde{\beta}_{x}^{\trans}\tilde{x}_{t})
\]
for $\abv{\beta}_{x,\perp}\defeq\tilde{\beta}_{x,\perp}(\tilde{\beta}_{x,\perp}^{\trans}\tilde{\beta}_{x,\perp})^{-1}$
and $\abv{\beta}_{x}\defeq\tilde{\beta}_{x}(\tilde{\beta}_{x}^{\trans}\tilde{\beta}_{x})^{-1}$,
it follows that $\Delta\tilde{x}_{t}$ is a linear function of $(\tilde{\chi}_{t},\tilde{\chi}_{t-1})$.
Thus $\Delta x_{t}$ is a Lipschitz continuous function of $(\tilde{\chi}_{t},\tilde{\chi}_{t-1})$.

\textbf{Deterministics in cases~\casecens{} and \caseclas{}.} Now
suppose that $z_{t}=(y_{t},x_{t}^{\trans})^{\trans}$ satisfies \assref{detprime}
rather than \assref{det}. It follows from \eqref{projtilde} that
$e_{1}^{\trans}\tilde{P}_{\beta_{\perp}^{+}}\tilde{c}=0$; thus \assref{detprime}
also holds for the derived canonical process $\tilde{z}_{t}=(\tilde{y}_{t},\tilde{x}_{t}^{\trans})^{\trans}$.
We then have from the proof of \thmref{codet} (under \assref{dgp-canon})
that the conclusions of part~(ii) of Theorems~\ref{thm:cocens}
and \ref{thm:coclas} apply to the detrended canonical processes
\[
\begin{bmatrix}\tilde{y}_{t}^{d}\\
\tilde{x}_{t}
\end{bmatrix}=\tilde{z}_{t}^{d}=\tilde{z}_{t}-(\tilde{P}_{\beta_{\perp}^{+}}\tilde{c})t,
\]
where $\tilde{y}_{t}^{d}=\tilde{y}_{t}$ because $e_{1}^{\trans}\tilde{P}_{\beta_{\perp}^{+}}c=0$.
Using $P(y)$ as defined in \eqref{Pydef}, and noting that $\sgn y_{t}=\sgn\tilde{y}_{t}$,
we may write 
\begin{equation}
z_{t}=P(\tilde{y}_{t})\tilde{z}_{t}\label{eq:ztotilde}
\end{equation}
or equivalently $P(y_{t})^{-1}z_{t}=\tilde{z}_{t}$, and so
\begin{align*}
\tilde{z}_{t}^{d}=\tilde{z}_{t}-(\tilde{P}_{\beta_{\perp}^{+}}\tilde{c})t & =P(y_{t})^{-1}z_{t}-[(P^{+})^{-1}P_{\beta_{\perp}^{+}}c]t\\
 & =_{(1)}P(y_{t})^{-1}[z_{t}-(P_{\beta_{\perp}^{+}}c)t]=P(y_{t})^{-1}z_{t}^{d}
\end{align*}
where $=_{(1)}$ holds because only the first column of $P(y)$ depends
on $y$, and $e_{1}^{\trans}P_{\beta_{\perp}^{+}}c=0$. Hence
\begin{equation}
z_{t}^{d}=P(\tilde{y}_{t}^{d})\tilde{z}_{t}^{d},\label{eq:zdtotilde}
\end{equation}
where $\tilde{y}_{t}=\tilde{y}_{t}^{d}$. Because the mapping between
the detrended processes in \eqref{zdtotilde} is identical to that
for the original processes in \eqref{ztotilde}, the same arguments
as given in cases~\casecens{} and \caseclas{} immediately above
will now transpose the conclusions of \thmref{codet} from $\tilde{z}_{t}^{d}$
to $z_{t}^{d}$.

\section{Verification of Remarks~\ref{rem:cocens}\enuref{rem:cocens:jsr}
and \ref{rem:cocens}\enuref{rem:kappa}}

\label{supp:cocens:verification-of-remark}

We first note the following corollary to \lemref{lincoint}. We say
a VAR is \emph{stationary} if all its autoregressive roots lie outside
the unit circle (i.e.\ irrespective of whether the series is given
a stationary initialisation).

\needspace{3\baselineskip}
\begin{lem}
\label{lem:lincointcor}Suppose that the assumptions of \lemref{lincoint}
hold, and $\{z_{t}\}$ is generated according to
\[
z_{t}=c+\Phi(L)z_{t-1}+u_{t}.
\]
Let $\xi_{t}\defeq\beta^{\trans}z_{t}$. Then $\b{\xi}_{t}\defeq(\xi_{t}^{\trans},\Delta z_{t}^{\trans},\ldots,\Delta z_{t-(k-2)}^{\trans})^{\trans}$
follows the stationary VAR given by
\[
\b{\xi}_{t}=\b{\beta}_{1:p}^{\trans}c+(I_{p(k-1)+r}+\b{\beta}^{\trans}\b{\alpha})\b{\xi}_{t-1}+\b{\beta}_{1:p}^{\trans}u_{t}.
\]
\end{lem}
\needspace{3\baselineskip}
\begin{proof}
The claim follows from part~(ii) of \lemref{lincoint}, and premultiplying
the companion form $\Delta\b z_{t}=\b c+\b{\Pi}\b z_{t-1}+\b u_{t}$
by $\b{\beta}^{\trans}$, where $\b z_{t}=(z_{t}^{\trans},z_{t-1}^{\trans},\ldots,z_{t-k+1}^{\trans})^{\trans}$,
$\b c=(c^{\trans},0_{p(k-1)}^{\trans})^{\trans}$ and $\b u_{t}=(u_{t}^{\trans},0_{p(k-1)}^{\trans})^{\trans}$.
\end{proof}
\begin{proof}[Proof of \remref{cocens}\enuref{rem:cocens:jsr}]
 Rather than working with the matrices $\b F_{0}$ and $\b F_{1}$
directly, it is easier if we consider the autoregressive systems that
they correspond to, as given in \eqref{coclas:sm}, recognising that
those systems are dynamically stable if and only if the eigenvalues
of their companion form matrices lie strictly inside the unit circle
(e.g.\ \citealp[Sec.~2.1.1]{Lut07}).

Let $t\in\naturals$ be given. We first consider $\b F_{0}$: this
is operative if $\delta_{s}=0$ for all $s\in\{t-k,\ldots,t-1\}$.
In this case, $\pseudy_{s}=y_{s}^{-}=0$ for all such $s$, and \eqref{coclas:sm}
reduces to
\[
\b{\xi}_{t}^{+}=\b{\beta}_{1:p}^{+\trans}c+(I+\b{\beta}^{+\trans}\b{\alpha}^{+})\b{\xi}_{t-1}^{+}+\b{\beta}_{1:p}^{+\trans}u_{t},
\]
By \lemref{lincoint}, the eigenvalues of $I+\b{\beta}^{+\trans}\b{\alpha}^{+}$
lie inside the unit circle; hence so too do those of $\b F_{0}$.
(Indeed, it may be verified that the only nonzero eigenvalues of $\b F_{0}$
are those of $I+\b{\beta}^{+\trans}\b{\alpha}^{+}$.)

We next consider $\b F_{1}$, which is operative if $\delta_{s}=1$
for all $s\in\{t-k,\ldots,t-1\}$. Then $y_{s}=\pseudy_{s}=y_{s}^{-}<0$
for all such $s$, and \eqref{coclas:sm} becomes\begin{subequations}\label{eq:delta1case}
\begin{align}
\b{\xi}_{t}^{+} & =\b{\beta}_{1:p}^{+\trans}c+(I+\b{\beta}^{+\trans}\b{\alpha}^{+})\b{\xi}_{t-1}^{+}\\
 & \qquad\qquad+\b{\beta}_{1:p}^{+\trans}(\phi_{1}^{-}-\phi_{1}^{+})y_{t-1}+\b{\beta}_{1:p}^{+\trans}\sum_{i=2}^{k}(\phi_{i}^{-}-\phi_{i}^{+})y_{t-i}+\b{\beta}_{1:p}^{+\trans}u_{t}\nonumber \\
y_{t} & =e_{1}^{\trans}c+\b e_{1}^{\trans}\b{\alpha}^{+}\b{\xi}_{t-1}^{+}\nonumber \\
 & \qquad\qquad+[1+e_{1}^{\trans}(\phi_{1}^{-}-\phi_{1}^{+})]y_{t-1}+\sum_{i=2}^{k}e_{1}^{\trans}(\phi_{i}^{-}-\phi_{i}^{+})y_{t-i}+e_{1}^{\trans}u_{t}.
\end{align}
\end{subequations}Now the preceding must agree with the $(\b{\xi}_{t}^{+},y_{t})$
generated by the VAR when $y_{s}<0$ for all $s\in\{t-k,\ldots,t-1\}$,
i.e.\ generated according to
\[
z_{t}=\Phi^{-}(L)z_{t-1}+u_{t}.
\]
Under \assref{co} and \assref{cocens}\enuref{cocens:rank}, the
preceding is a linear cointegrated VAR with
\[
\Pi^{-}=\begin{bmatrix}\pi^{-} & \Pi^{x}\end{bmatrix}=\Pi^{+}+\begin{bmatrix}\pi^{-}-\pi^{+} & 0\end{bmatrix}=\begin{bmatrix}\alpha^{+} & 0\\
0 & \pi^{-}-\pi^{+}
\end{bmatrix}\begin{bmatrix}\beta^{+\trans}\\
e_{1}^{\trans}
\end{bmatrix}\eqdef\alpha^{-}\beta^{-\trans},
\]
and equilibrium errors $\xi_{t}^{-}\defeq\beta^{-\trans}z_{t}=[\begin{smallmatrix}\beta^{+\trans}z_{t}\\
y_{t}
\end{smallmatrix}]=[\begin{smallmatrix}\xi_{t}^{+}\\
y_{t}
\end{smallmatrix}]$. It follows from \lemref{lincointcor} that
\[
\b{\xi}_{t}^{-}\defeq(\xi_{t}^{-\trans},\Delta z_{t}^{\trans},\ldots,\Delta z_{t-(k-2)}^{\trans})^{\trans}=(\xi_{t}^{+\trans},y_{t},\Delta z_{t}^{\trans},\ldots,\Delta z_{t-(k-2)}^{\trans})^{\trans}
\]
evolves according to a stationary VAR. Since $(\b{\xi}_{t}^{+},y_{t})$
is simply a reordering of $\b{\xi}_{t}^{-}$, it follows that the
system described by \eqref{delta1case}, i.e.\ by $\b F_{1}$, must
also be stationary. Hence $\b F_{1}$ must have all its eigenvalues
strictly inside the unit circle.
\end{proof}
\needspace{3\baselineskip}
\begin{proof}[Proof of \remref{cocens}\enuref{rem:kappa}]
 Suppose $k=1$. In this case, $\pi^{-}-\pi^{+}=\phi_{1}^{-}-\phi_{1}^{+}$,
and so by \remref{cocens}\enuref{rem:cocens:jsr}, the matrix
\[
\b F_{1}=\begin{bmatrix}I_{r}+\beta^{+\trans}\alpha^{+} & \beta^{+\trans}(\pi^{-}-\pi^{+})\\
e_{1}^{\trans}\alpha^{+} & 1+e_{1}^{\trans}(\pi^{-}-\pi^{+})
\end{bmatrix}
\]
must have all its eigenvalues inside the unit circle; hence $\det(I_{r+1}-\b F_{1})=\prod_{i=1}^{r+1}(1-\rho_{i}(\b F_{1}))>0$,
where $\rho_{i}(\b F_{1})$ denotes the $i$th eigenvalue of $\b F_{1}$.
Note that
\begin{align*}
e_{1}^{\trans}(\pi^{-}-\pi^{+})-e_{1}^{\trans}\alpha^{+}(\beta^{+\trans}\alpha^{+})^{-1}\beta^{+\trans}(\pi^{-}-\pi^{+}) & =e_{1}^{\trans}[I_{p}-\alpha^{+}(\beta^{+\trans}\alpha^{+})^{-1}\beta^{+\trans}](\pi^{-}-\pi^{+})\\
 & =e_{1}^{\trans}\beta_{\perp}^{+}(\alpha_{\perp}^{+\trans}\beta_{\perp}^{+})^{-1}\alpha_{\perp}^{+\trans}(\pi^{-}-\pi^{+})\\
 & =e_{1}^{\trans}P_{\beta_{\perp}^{+}}\pi^{-}=\regcoef_{1},
\end{align*}
where the penultimate equality holds since $\alpha_{\perp}^{+}\pi^{+}=0$,
and $\Gamma^{+}(1)=I_{p}$ when $k=1$. Hence
\[
\begin{bmatrix}I_{r} & 0\\
-e_{1}^{\trans}\alpha^{+}(\beta^{+\trans}\alpha^{+})^{-1} & 1
\end{bmatrix}(\b F_{1}-I_{r+1})=\begin{bmatrix}\beta^{+\trans}\alpha^{+} & \beta^{+\trans}(\pi^{-}-\pi^{+})\\
0 & \regcoef_{1}
\end{bmatrix},
\]
from which it follows that
\begin{align*}
(-1)^{r+1}\det(I_{r+1}-\b F_{1}) & =\det(\b F_{1}-I_{r+1})\\
 & =\regcoef_{1}\det(\beta^{+\trans}\alpha^{+})=\regcoef_{1}(-1)^{r}\det[I_{r}-(I_{r}+\beta^{+\trans}\alpha^{+})].
\end{align*}
By \lemref{lincoint}, $I_{r}+\beta^{+\trans}\alpha^{+}$ has all
its eigenvalues inside the unit circle, and so $\det[I_{r}-(I_{r}+\beta^{+\trans}\alpha^{+})]>0$.
Hence $\regcoef_{1}<0$.

Next suppose $p=1$, while allowing for general $k\in\naturals$.
Then
\[
\regcoef_{1}=\regcoef=P_{\beta_{\perp}^{+}}\pi^{-}=-\gamma^{+}(1)^{-1}\phi^{-}(1)
\]
where $\gamma^{+}(1)=1-\sum_{i=1}^{k-1}\gamma_{i}^{+}$ and $\phi^{-}(1)=1-\sum_{i=1}^{k}\phi_{i}^{-}$.
Both $\gamma^{+}(\lambda)$ and $\phi^{-}(\lambda)$ have all their
roots outside the unit circle, and hence both $\gamma^{+}(1)$ and
$\phi^{-}(1)$ are strictly positive. It follows that $\regcoef_{1}<0$
as required.
\end{proof}

\section{Detailed calculations for \exaref{abh}}

\label{supp:abh-calculations}

Here we provide some further details of the calculations underlying
\exaref{abh}. For convenience, we reproduce the system \eqref{abh1}
in CKSVAR form as
\begin{multline}
\begin{bmatrix}1\\
-\delta^{+}\\
0
\end{bmatrix}\bar{\pi}_{t}^{+}+\begin{bmatrix}1\\
-\delta^{-}\\
0
\end{bmatrix}\bar{\pi}_{t}^{-}+\begin{bmatrix}0 & 0\\
1 & 0\\
0 & 1
\end{bmatrix}\begin{bmatrix}\Delta\bar{y}_{t}\\
g_{t}
\end{bmatrix}\\
=\begin{bmatrix}1\\
-\delta^{+}\\
0
\end{bmatrix}\bar{\pi}_{t-1}^{+}+\begin{bmatrix}1\\
-\delta^{-}\\
0
\end{bmatrix}\bar{\pi}_{t-1}^{-}+\begin{bmatrix}0 & 0\\
0 & 1\\
0 & 1
\end{bmatrix}\begin{bmatrix}\Delta\bar{y}_{t-1}\\
g_{t-1}
\end{bmatrix}+\begin{bmatrix}u_{t}^{\pi}\\
u_{t}^{y}\\
u_{t}^{g}
\end{bmatrix}.\tag{\ref{eq:ABH}}\label{eq:ABH-appendix}
\end{multline}
We thus have first-order ($k=1$) CKSVAR with
\begin{align}
\Phi_{0}^{\pm} & =\begin{bmatrix}1 & 0 & 0\\
-\delta^{\pm} & 1 & 0\\
0 & 0 & 1
\end{bmatrix} & \Phi_{1}^{\pm} & =\begin{bmatrix}1 & 0 & 0\\
-\delta^{\pm} & 0 & 1\\
0 & 0 & 1
\end{bmatrix}.\label{eq:PhipmABH}
\end{align}
Hence
\begin{equation}
\Pi^{\pm}=\begin{bmatrix}1 & 0 & 0\\
-\delta^{\pm} & 0 & 1\\
0 & 0 & 1
\end{bmatrix}-\begin{bmatrix}1 & 0 & 0\\
-\delta^{\pm} & 1 & 0\\
0 & 0 & 1
\end{bmatrix}=\begin{bmatrix}0 & 0 & 0\\
0 & -1 & 1\\
0 & 0 & 0
\end{bmatrix}=\begin{bmatrix}0\\
1\\
0
\end{bmatrix}\begin{bmatrix}0 & -1 & 1\end{bmatrix}=\alpha\beta^{\pm\trans}\label{eq:PipmABH}
\end{equation}
and thus $\Pi^{\pm}$ and $\beta^{\pm}=(0,-1,1)^{\trans}$ are regime-invariant.
$\Pi^{x}$ is obtained from the final $p-1=2$ columns of $\Pi^{+}$
(or $\Pi^{-}$), and so
\[
\Pi^{x}=\begin{bmatrix}0 & 0\\
-1 & 1\\
0 & 0
\end{bmatrix}
\]
whence $\rank\Pi^{\pm}=\rank\Pi^{x}=1$, showing that \assref{coclas}\ref{enu:coclas:rank}
holds with $r=1$.

In view of \eqref{PipmABH}, we can construct orthocomplement matrices
as
\begin{align*}
\alpha_{\perp} & =\begin{bmatrix}1 & 0\\
0 & 0\\
0 & 1
\end{bmatrix} & \beta_{\perp}^{\pm} & =\begin{bmatrix}1 & 0\\
0 & 1\\
0 & 1
\end{bmatrix},
\end{align*}
where it will be observed that the $\beta_{\perp}^{\pm}$ indeed have
the form required by \eqref{betaperpfn}. Since $\Gamma_{i}=0$ for
all $i\geq1$ in this first-order model, we have $\Gamma^{\pm}(1)=\Phi_{0}^{\pm}$,
and thus from computing
\[
\alpha_{\perp}^{\trans}\Gamma^{\pm}(1)\beta_{\perp}^{\pm}=\begin{bmatrix}1 & 0 & 0\\
0 & 0 & 1
\end{bmatrix}\begin{bmatrix}1 & 0 & 0\\
-\delta^{\pm} & 1 & 0\\
0 & 0 & 1
\end{bmatrix}\begin{bmatrix}1 & 0\\
0 & 1\\
0 & 1
\end{bmatrix}=\begin{bmatrix}1 & 0\\
0 & 1
\end{bmatrix}=I_{2}
\]
we see immediately that $\det\alpha_{\perp}^{\trans}\Gamma^{+}(1)\beta_{\perp}^{+}=1=\det\alpha_{\perp}^{\trans}\Gamma^{-}(1)\beta_{\perp}^{-}$,
whence \assref{coclas}\ref{enu:coclas:det} holds.

The final condition to verify is \assref{coclas}\ref{enu:coclas:jsr}
-- or rather \ref{enu:coclas:jsr-struct}, because the CSKVAR \eqref{ABH-appendix}
is in structural form. We must therefore compute $\tilde{\Pi}^{\pm}$
for the associated canonical form, the parameters can be obtained
via the mapping \eqref{canon-vars}--\eqref{Q-canon}. From \eqref{PhipmABH}
we have (noting that here $y_{t}=\bar{\pi}_{t}$ and $x_{t}=(\Delta y_{t},g_{t})^{\trans}$),
\begin{align*}
\phi_{0,yy}^{\pm} & =1 & \phi_{0,xy}^{\pm} & =\begin{bmatrix}-\delta^{\pm}\\
0
\end{bmatrix} & \phi_{0,yx}^{\trans} & =\begin{bmatrix}0 & 0\end{bmatrix} & \Phi_{0,xx}= & \begin{bmatrix}1 & 0\\
0 & 1
\end{bmatrix}
\end{align*}
whence
\[
\bar{\phi}_{0,yy}^{\pm}=\phi_{0,yy}^{\pm}-\phi_{0,yx}^{\trans}\Phi_{0,xx}^{-1}\phi_{0,xy}^{\pm}=1
\]
and so
\[
P^{-1}=\begin{bmatrix}\bar{\phi}_{0,yy}^{+} & 0 & 0\\
0 & \bar{\phi}_{0,yy}^{-} & 0\\
\phi_{0,xy}^{+} & \phi_{0,xy}^{-} & \Phi_{0,xx}
\end{bmatrix}=\begin{bmatrix}1 & 0 & 0 & 0\\
0 & 1 & 0 & 0\\
-\delta^{+} & -\delta^{-} & 1 & 0\\
0 & 0 & 0 & 1
\end{bmatrix}
\]
and
\[
Q=\begin{bmatrix}1 & -\phi_{0,yx}^{\trans}\Phi_{0,xx}^{-1}\\
0 & I_{p-1}
\end{bmatrix}=\begin{bmatrix}1 & 0\\
0 & I_{2}
\end{bmatrix}=I_{3}.
\]
Hence, $\tilde{\Pi}=[\tilde{\pi}^{+},\tilde{\pi}^{-},\Pi^{x}]=-\tilde{\Phi}(1)$
is given by
\[
\tilde{\Pi}=Q\begin{bmatrix}\pi^{+} & \pi^{-} & \Pi^{x}\end{bmatrix}P=\begin{bmatrix}0 & 0 & 0 & 0\\
0 & 0 & -1 & 1\\
0 & 0 & 0 & 0
\end{bmatrix}\begin{bmatrix}1 & 0 & 0 & 0\\
0 & 1 & 0 & 0\\
\delta^{+} & \delta^{-} & 1 & 0\\
0 & 0 & 0 & 1
\end{bmatrix}=\begin{bmatrix}0 & 0 & 0 & 0\\
-\delta^{+} & -\delta^{-} & -1 & 1\\
0 & 0 & 0 & 0
\end{bmatrix}.
\]
It follows that the matrices $\tilde{\Pi}^{\pm}=[\tilde{\pi}^{\pm},\Pi^{x}]$
\emph{are} regime dependent, and factorise as
\[
\tilde{\Pi}^{\pm}=\begin{bmatrix}0 & 0 & 0\\
-\delta^{\pm} & -1 & 1\\
0 & 0 & 0
\end{bmatrix}=\begin{bmatrix}0\\
1\\
0
\end{bmatrix}\begin{bmatrix}-\delta^{\pm} & -1 & 1\end{bmatrix}\eqdef\tilde{\alpha}\tilde{\beta}^{\pm\trans}
\]
whence
\[
I_{r}+\tilde{\beta}^{\pm\trans}\tilde{\alpha}=1+\begin{bmatrix}-\delta^{\pm} & -1 & 1\end{bmatrix}\begin{bmatrix}0\\
1\\
0
\end{bmatrix}=0
\]
and thus \ref{enu:coclas:jsr-struct} is satisfied.

\end{document}